\documentclass[11pt,a4paper]{amsart}

\newtheorem{theorem}{Theorem}[section]

\newtheorem{corollary}[theorem]{Corollary}
\newtheorem{lemma}[theorem]{Lemma}
\newtheorem{proposition}[theorem]{Proposition}

\newtheorem*{conjecture}{Conjecture}

\theoremstyle{definition}
\newtheorem*{definition}{Definition}

\newtheorem{example}[theorem]{Example}
\newtheorem{remark}[theorem]{Remark}
\newtheorem{remarks}[theorem]{Remarks}
\newtheorem{examples}[theorem]{Examples}

\usepackage[english]{babel}
\usepackage{graphicx}
\usepackage[square,sort,comma,numbers]{natbib}
\usepackage[all,cmtip,2cell,graph]{xy}
\usepackage{pstricks}
\usepackage{enumerate}
\usepackage{amsfonts,amssymb,amsmath,pinlabel,array,hhline}
\usepackage{slashed}
\usepackage{tabulary}
\usepackage{a4wide}
\usepackage{mathrsfs}
\usepackage{calrsfs}
\usepackage{bbm,dsfont}
\usepackage{mathtools}
\UseAllTwocells

\newcommand{\Z}{\mathbb{Z}}

\newcommand{\R}{\mathbb{R}}
\newcommand{\C}{\mathbb{C}}
\newcommand{\M}{\mathcal{M}}
\newcommand{\D}{\mathcal{D}}
\newcommand{\G}{\Gamma}
\newcommand{\T}{\mathbb{T}}
\newcommand{\K}{\mathcal{K}}

\newcommand{\Pf}{\mathrm{Pf}}
\newcommand{\GL}{\operatorname{GL}}

\DeclarePairedDelimiter\floor{\lfloor}{\rfloor}

\DeclareSymbolFont{EulerScript}{U}{eus}{m}{n}
\DeclareSymbolFontAlphabet\mathscr{EulerScript}

\begin{document}

\title{The dimer and Ising models on Klein bottles}

\author{David Cimasoni}
\address{Universit\'e de Gen\`eve, Section de math\'ematiques, 2 rue du Li\`evre, 1211 Gen\`eve, Switzerland}
\email{david.cimasoni@unige.ch}


\begin{abstract}
We study the dimer and Ising models on a finite
planar weighted graph with periodic-antiperiodic
boundary conditions, i.e. a graph~$\G$ in the Klein bottle~$\K$.
Let~$\G_{mn}$ denote the graph obtained
by pasting~$m$ rows and~$n$ columns of copies of~$\G$, which embeds in~$\K$ for~$n$
odd and in the torus~$\mathbb{T}^2$ for~$n$ even.
We compute the dimer partition function~$Z_{mn}$ of~$\G_{mn}$ for~$n$ odd,
in terms of the well-known characteristic polynomial~$P$
of~$\G_{12}\subset\mathbb{T}^2$ together with a new characteristic
polynomial~$R$ of~$\G\subset\K$.

Using this result together with work of Kenyon, Sun and Wilson,
we show that in the bipartite case, this partition function has the asymptotic
expansion
\[
\log Z_{mn}=mn\frac{\mathbf{f}_0}{2}+\mathsf{fsc}+o(1)
\]
for~$m,n$ tending to infinity and~$m/n$ bounded below and above,
where~$\mathbf{f}_0$ is the bulk free energy for~$\G_{12}\subset\mathbb{T}^2$
and~$\mathsf{fsc}$ an explicit finite-size correction term.
The remarkable feature of this later term is its {\em universality\/}: it does not depend on the graph~$\G$,
but only on the zeros of~$P$ on the unit torus and on an explicit (purely imaginary) conformal shape parameter.
A similar expansion is also obtained in the non-bipartite case, assuming a conjectural
condition on the zeros of~$P$.

We then show that this asymptotic expansion holds for the Ising partition function as well, with~$\mathsf{fsc}$ taking a particularly simple form: it vanishes in the subcritical regime, is equal to~$\log(2)$ in the supercritical regime, and to an explicit function of the shape parameter at criticality.
These results are in full agreement with the conformal field theory predictions of
Bl\"ote, Cardy and Nightingale. 
\end{abstract}

\maketitle


\section{Introduction}
\label{sec:intro}

\subsection{Background on the dimer model}
\label{sub:intro-dimer}

A {\em dimer configuration\/} on a finite graph~$\G=(V,E)$ is a perfect matching on~$\G$, i.e.
a family of edges~$M\subset E$ such that each vertex~$v\in V$ is adjacent to exactly one edge of~$M$.
Given a non-negative edge weight system~$\nu=(\nu_e)_{e\in E}$ on~$\G$,
a probability measure on the set~$\M(\G)$ of dimer configurations on~$\G$ is given by
\[
\mathbb{P}(M)=\frac{\nu(M)}{Z}\;,\quad\text{with }\;\;\nu(M)=\prod_{e\in E}\nu_e\;\text{ and }\;
Z=\sum_{M\in\M(\G)}\nu(M)\,.
\]
The normalization constant~$Z=Z(\G,\nu)$ is the {\em partition function\/}
of the {\em dimer model\/} on~$\G$.

The first breakthrough in the study of this model came with the foundational work of Kasteleyn~\cite{Ka1,Ka3},
Temperley and Fisher~\cite{T-F,Fisher}. They showed that, in the case of a planar graph~$\G$,
the partition function~$Z$ can be expressed as
the Pfaffian of a signed, weighted skew-adjacency matrix of~$\G$, now called a {\em Kasteleyn matrix\/}.
In the case of a general graph, embedded in an orientable surface of arbitary genus~$g$, this method extends,
but~$Z$ is equal to an alternating sum of~$2^{2g}$ Pfaffians~\cite{Tes,C-R}. In particular, if~$\G$ embeds in
the torus~$\mathbb{T}^2$, then~$Z$ is equal to an alternating sum of~$4$ Pfaffians.
Using this method, Kasteleyn~\cite{Ka1}, Fisher~\cite{Fisher} and
Ferdinand~\cite{Fer} were able to compute successive terms in the asymptotic expansion of
the partition function for the weighted square lattice with various boundary conditions
(see e.g.~\cite[Section~1.1]{KSW} for details).
In the case of periodic-periodic boundary conditions, i.e. when the lattice is embedded in the
torus~$\mathbb{T}^2$,
the final result reads as follows.
For the~$m\times n$ square lattice with horizontal (resp. vertical) edges-weights equal to~$x$
(resp.~$y$), periodic-periodic boundary conditions and~$mn$ even,
the corresponding partition function~$Z_{mn}$ satisfies
\[
\log Z_{mn}=mn\,\mathbf{f}_0(x,y)+\mathsf{fsc}_{(-1)^{m+n}}(\textstyle\frac{nx}{my})+o(1)\,,
\]
where the {\em bulk free energy\/}~$\mathbf{f}_0$
depends on the weights~$x,y$ in an explicit but complicated way,
while the constant order {\em finite size correction\/} term~$\mathsf{fsc}_\pm$ only depends on the {\em shape parameter\/}~$\frac{nx}{my}$ of the torus, together with the parity of~$m+n$.

For conformally invariant two-dimensional models
on a closed surface~$\Sigma$ of vanishing Euler characteristic,
such an asymptotic expansion is believed to hold for arbitrary graphs, with the finite-size correction
term depending only on the {\em universality class\/} of the model at criticality and on the topology
of the surface, but not on the underlying graph~\cite{BCN,Ca-P}.
There are exactly two closed surfaces with~$\chi(\Sigma)=0$, namely the torus~$\T^2$ and the Klein
bottle~$\K$, corresponding to periodic-periodic and periodic-antiperiodic boundary conditions, respectively.
Generalizing the asymptotic expansion displayed above from the square lattice to arbitrary weighted graphs
in~$\T^2$ and~$\K$ is no easy task, and was out of reach with the tools available in the 1960s.

The extension to bipartite toric graphs was made possible with the second breakthrough in the study of the dimer model, namely the work of 
Kenyon and coauthors, in particular the seminal article of Kenyon, Okounkov and Sheffield~\cite{KOS}.
In a nutshell, it was discovered that many large scale properties of the dimer model on a doubly
periodic bipartite graph~$\G$ can be understood from the behavior of the associated
{\em characteristic polynomial\/}~$P(z,w)\in\R[z^{\pm 1},w^{\pm 1}]$, defined as the determinant of
a twisted Kasteleyn matrix for~$\G\subset\mathbb{T}^2$.
(Here, we make use of the notation of~\cite{KSW}, where~$P$ stands for the full polynomial
which factors as~$P(z,w)=Q(z,w)Q(z^{-1},w^{-1})$.)
In particular, the Pfaffian formula of~\cite{Tes}
can now be reformulated as
\begin{equation}
\label{equ:Tesler}
Z=\textstyle{\frac{1}{2}}\left(\pm P(1,1)^{1/2}\pm P(-1,1)^{1/2}\pm P(1,-1)^{1/2}\pm P(-1,-1)^{1/2}\right)\,,
\end{equation}
where the signs can be given a natural geometric interpretation~\cite{C-R}.
A crucial role is played by
the intersection of the corresponding {\em spectral curve\/}, i.e. the zeros of~$P$, with the unit torus~$S^1\times S^1$. As proved in~\cite{KOS,K-O}, there is at most two such zeros, and they are positive nodes.

Another important feature of this polynomial is that it behaves well with respect to enlargement of the
fundamental domain. In other words, if one considers the graph~$\G_{mn}\subset\mathbb{T}^2$
obtained by pasting~$m$ rows and~$n$ columns of copies of~$\G$, then the associated characteristic
polynomial~$P_{mn}$ can be computed from~$P_{11}=P$ via
\begin{equation}
\label{equ:Pmn}
P_{mn}(\zeta,\xi)=\prod_{z^n=\zeta}\prod_{w^m=\xi}P(z,w)
\end{equation}
for any~$\zeta,\xi\in\C^*$ (see~\cite[Theorem~3.3]{KOS}).
Note that the definition of the characteristic polynomial extends
to arbitrary (possibly non-bipartite) graphs~$\G\subset\mathbb{T}^2$, and the formula~\eqref{equ:Pmn} still
holds, but the corresponding spectral curve is not well-understood.
It is believed to intersect the unit torus in at most two
points which are real positive nodes (see~\cite[Section~1.2]{KSW}), but this remains a conjecture.

With these tools in hand, the only hurdle left in the computation of the asymptotic expansion
of  the partition function~$Z_{mn}$
for arbitrary toric graphs~$\Gamma_{mn}$ was the determination of the asymptotic behavior of~$P_{mn}$ for a
non-negative analytic function~$P$ on the unit torus whose only zeros are positive nodes.
This was done by Kenyon, Sun and Wilson in~\cite[Theorem~1]{KSW}, see Section~\ref{sub:KSW} below for a summary.
The resulting asymptotic expansion for~$Z_{mn}$ is too involved to be stated in detail here,
so we refer the reader to~\cite[Theorem~2]{KSW} and mention that it is in full agreement with the
conformal field theory (CFT) predictions of~\cite{BCN}.
These results are proved for arbitrary bipartite graphs in~$\T^2$, and for non-bipartite graphs as well
assuming the aforementioned conjectural condition on the zeros of the characteristic polynomial.

\medskip

While the toric case is now fully settled, the case of the Klein bottle remains very poorly understood.
To the best of our knowledge, the only available results deal with the square lattice~\cite{L-W,Lu-Wu99,IOH}, and are partially contradictory (see Examples~\ref{ex:fsc-sq} and~\ref{ex:sqbip} below, where we point out inaccuracies in~\cite{Lu-Wu99} and~\cite{IOH}, as well as in~\cite{Lu-Wu}).

The main goal of the present article is to fill this gap, i.e. to compute the asymptotic expansion
of the dimer partition function, and in particular the finite size correction term, for arbitrary 
weighted graphs in the Klein bottle.

\subsection{Results on the dimer model}
\label{sub:results-dimer}

To understand the technical difficulties of this endeavour, let us go once again
through the list of tools used in the toric case.

The first fundamental tool is Kasteleyn's theorem extended to toric graphs, which can be stated as
Equation~\eqref{equ:Tesler}. Fortunately, such a formula is also available for non-orientable
surfaces: if~$\G$ is embedded in a closed (possibly non-orientable) surface~$\Sigma$,
then~$Z$ can be computed as a linear combination of~$2^{2-\chi(\Sigma)}$ Pfaffians of (possibly
complex-valued) Kasteleyn matrices~\cite{Tes}.
Furthermore, the coefficients in this linear combination can still be interpreted geometrically~\cite{Cim}.
In particular, the partition function of any graph embedded in~$\K$ is given by~$4$ Pfaffians,
which turn out to be two pairs of conjugate complex numbers, so~$2$ well-chosen Pfaffians are sufficient.

The second tool is the characteristic polynomial~$P$, defined for graphs in~$\T^2$.
Intrinsically, this polynomial should be understood as an element of the group ring~$\R[H_1(\T^2;\Z)]$,
the choice of a basis of~$H_1(\T^2;\Z)\simeq\Z^2$ then leading to the more familiar
polynomial ring~$\R[z^{\pm 1},w^{\pm 1}]$. Therefore, since~$H_1(\K;\Z)\simeq\Z\oplus\Z/2\Z$ and the Kasteleyn matrices are now complex valued, a naturally defined characteristic polynomial~$R$ for
graphs in the Klein bottle should be an element of the quotient
ring~$\C[H_1(\K;\Z)]\simeq\C[z^{\pm 1},w]/(w^2-1)$,
i.e. two 1-variable polynomials~$R(z,1),R(z,-1)\in\C[z^{\pm 1}]$.
The definition of such a polynomial is the first technical step of this work, see Section~\ref{sub:Klein}.
The Pfaffian formula of~\cite{Tes} reinterpreted in the spirit of~\cite{Cim} now reads
\begin{equation}
\label{equ:Pfintro}
Z=\big|{\rm Im}\big(R(\pm 1,1)^{1/2}\big)\big|+\big|{\rm Re}\big(R(\pm 1,-1)^{1/2}\big)\big|\,,
\end{equation}
see Proposition~\ref{prop:Pf}. Remarkably, the polynomials~$R$ of~$\G$ and~$P$ of its~$2$-cover~$\widetilde{\G}\subset\T^2$
(see Figure~\ref{fig:Gmn}) are related via
\begin{equation}
\label{equ:RPintro}
R(z,w)R(-z,w)=P(z^2,w)
\end{equation}
for~$w=\pm 1$, see Proposition~\ref{prop:P}. Moreover, the order~$2$ symmetry of~$\widetilde{\G}$
implies the formula~$P(z,w)=P(z,w^{-1})$ for the associated polynomial.

If~$\G$ is bipartite, then much more can be said, constituting the first technically
challenging results of this article. As in the toric case, the Kleinian
characteristic polynomial factors as~$R(z,w)=S(z,w)S(z^{-1},w)$ for~$w=\pm 1$,
and one can use the powerful tools of~\cite{KOS}, namely amoebas of Harnack curves,
to show that all the zeros of~$P(z,w)$ on the unit torus satisfy~$z=-1$.
Furthermore, we prove that all the zeros of~$S(z,1)$ and~$S(z,-1)$ are simple, purely imaginary,
and interlaced along the imaginary axis (Proposition~\ref{prop:rootS}).
We also determine their behavior as one moves along the associated amoeba (or phase diagram),
see Lemma~\ref{lemma:S}.

\begin{figure}[tb]
\centering
\includegraphics[width=0.8\textwidth]{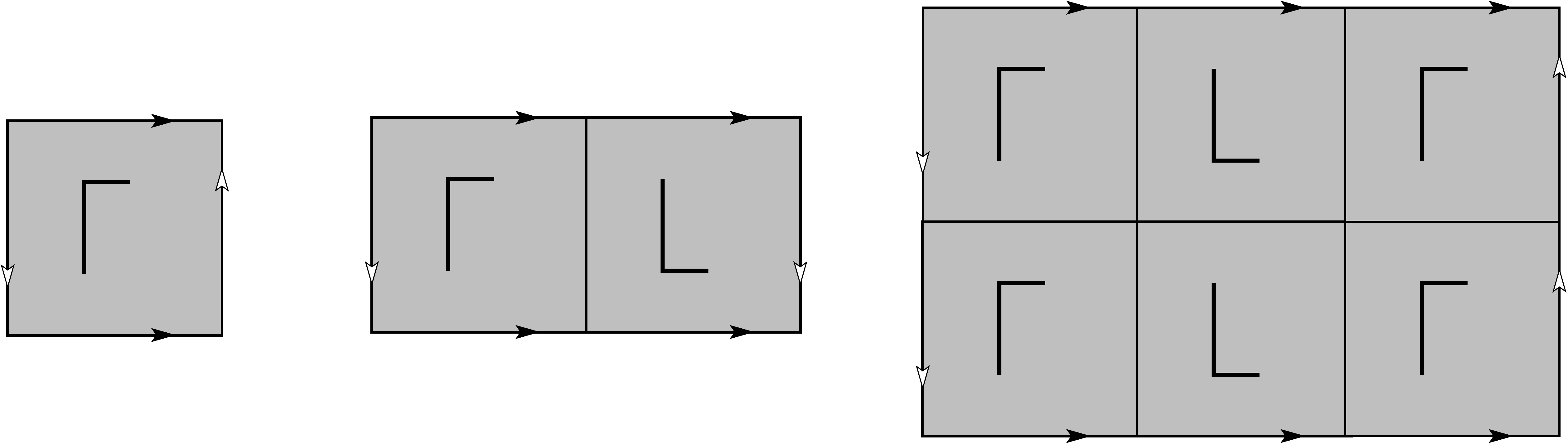}
\caption{Left: the graph~$\G\subset\K$, pictured with periodic boundary conditions in the
vertical direction (black arrows) and antiperiodic boundary conditions in the horizontal one (white arrows).
Center: the~$2$-cover~$\widetilde{\G}=\G_{12}\subset\T^2$.
Right: the associated graph~$\G_{mn}\subset\K$, here with~$m=2$ and~$n=3$.}
\label{fig:Gmn}
\end{figure}

The third tool used in the toric case is the fact that~$P$ behaves well under enlargement of
fundamental domains, as expressed in Equation~\eqref{equ:Pmn}. The extension of this result to the Kleinian
case is the main technical novelty of the present work, and the subject of the whole of Section~\ref{sec:enlarge}.
But first of all, let us clarify what we mean by ``enlargement of fundamental domain'' in the Klein bottle.
Given a weighted graph~$\G\subset\K$, consider the graph~$\G_{mn}$ obtained by pasting~$m$ rows
and~$n$ columns of copies of~$\G$ as illustrated in Figure~\ref{fig:Gmn}.
Observe that if~$n$ is even, then~$\G_{mn}$ embeds in the torus and is nothing but the~$m\times n/2$ enlargement
of~$\widetilde{\G}=\G_{12}\subset\T^2$,
a case well-understood. However, if~$n$ is odd, then~$\G_{mn}$ embeds in the Klein bottle,
and this is the case we will focus on.
Defined more intrinsically, we are looking at non-trivial covers~$\K_{mn}\to\K$ of the Klein bottle by
itself (which incidently only exist since~$\K$ has vanishing Euler characteristic), see also Section~\ref{sub:rem} below.

The idea now is to relate the two Kasteleyn matrices of~$\G_{mn}\subset\K_{mn}$ used in~\eqref{equ:Pfintro},
understood as discrete operators twisted by~$1$-dimensional representations~$\rho,\rho'$
of~$\pi_1(\K_{mn})<\pi_1(\K)$,
with the associated Kasteleyn operators on~$\G\subset\K$ twisted by the corresponding
{\em induced representations\/}~$\rho^\#,(\rho')^\#$ of~$\pi_1(\K)$.
This uses a general result, Theorem~\ref{thm:C-K} below,
which is probably well-known to the experts, but whose precise statement we have not been able
to find in the literature.
It is therefore the subject of a separate note with Adrien Kassel~\cite{C-K}, where we give a detailed proof
together with applications to other models of statistical physics.
Unlike that of the torus, the fundamental group of the Klein bottle is not abelian,
so the representations~$\rho^\#,(\rho')^\#$ need not split as products of~$1$-dimensional
representations as in Equation~\eqref{equ:Pmn}.
It turns out that they split as products of representations of dimension~$1$ and~$2$. Furthermore,
the determinant of the Kasteleyn matrices for~$\G$ twisted by the~$2$-dimensional representations
can be expressed as evaluations of the toric characteristic polynomial~$P$ of~$\widetilde{\G}\subset\T^2$.
The final result is somewhat cumbersome, so we will not state it here but refer the reader to Theorem~\ref{thm:Rmn}. Together with Equation~\eqref{equ:Pfintro}, it yields the following result.

\begin{theorem}
\label{thmintro:Z}
For positive integers~$m,n$ with~$n$ odd, we have
\[
Z_{mn}=\left|\sin(\alpha_n/2)\right|P_{mn}(1,1)^{1/4}+\left|\cos(\alpha'_n/2)\right|P_{mn}(1,-1)^{1/4}
\]
if~$m$ is odd, and
\[
Z_{mn}=\left|\sin((\alpha_n-\alpha'_n)/2)\right|P_{mn}(1,1)^{1/4}+P_{mn}(1,-1)^{1/4}
\]
if~$m$ is even, where
\[
\alpha_n=\mathrm{Arg}\Big(\prod_{z^n=1}R(z,1)\Big)\,,\quad\alpha'_n=\mathrm{Arg}\Big(\prod_{z^n=1}R(z,-1)\Big)\,,
\]
and~$P_{mn}(1,\pm 1)^{1/4}$ denotes the non-negative fourth root of~$P_{mn}(1,\pm 1)\ge 0$.
\end{theorem}

\medskip

These expressions are handy for the determination of the asymptotics of~$Z_{mn}$,
which now boils down to two distinct problems: the computation of the asymptotics of~$P_{mn}(1,\pm 1)$,
and of~$\mathrm{Arg}\left(\prod_{z^n=1}R(z,\pm 1)\right)$. The first question being answered by~\cite[Theorem~1]{KSW}, we are left with the second.
As it turns out, the limit of the
coefficients~$\left|\sin(\alpha_n/2)\right|$,~$\left|\cos(\alpha'_n/2)\right|$
and~$\left|\sin((\alpha_n-\alpha'_n)/2)\right|$ for~$n$ odd tending to~$\infty$ can only take the three
possible values~$0,\frac{\sqrt{2}}{2}$ and~$1$.
Furthermore, they are determined by the number of roots of~$R(z,1)$
and~$R(z,-1)$ outside the unit disc (see Lemmas~\ref{lemma:alpha} and~\ref{lemma:cst} for the precise statements). Using~\eqref{equ:RPintro}, one can then show that these limits are constant if the
dimer weights vary continuously without~$P(-1,\pm 1)$ vanishing.
This leads to an asymptotic expansion of~$Z_{mn}$, valid for general graphs, with the conjectural assumption
that all the zeros of~$P(z,w)$ on the unit torus are positive nodes with~$z=-1$ (see Theorem~\ref{thm:as-gen}).

In the bipartite case, this assumption is known to hold. Furthermore, as outlined above,
we have a good understanding of the locations of the zeros of~$S(z,\pm 1)$ along the imaginary axis.
This leads to the following result (Theorem~\ref{thm:as-bip}), where~$\vartheta_{00},\vartheta_{01}$
denote Jacobi theta functions and~$\eta$ is the Dedekind eta function, see
Section~\ref{sub:KSW}.

\begin{theorem}
\label{thmintro:as-bip}
Let~$\G\subset\K$ be a weighted bipartite graph embedded in the Klein bottle. Then, we have the asymptotic expansion
\[
\log Z_{mn}=mn\frac{\mathbf{f}_0}{2}+\mathsf{fsc}+o(1)
\]
for~$m$ and~$n$ tending to infinity with~$n$ odd and~$m/n$ bounded below and above, with
\[
\mathbf{f}_0=\iint_{S^1\times S^1}\log |Q(z,w)|\frac{dz}{2\pi iz}\frac{dw}{2\pi iw}
\]
and~$\mathsf{fsc}=\log\mathsf{FSC}$ given as follows:
\begin{enumerate}
\item{If~$Q(z,w)$ has no zeros in the unit torus, then~$\mathsf{FSC}=1$.}
\item{If~$Q(z,w)$ has two zeros~$(-1,w_0)\neq(-1,\overline{w}_0)$ in the unit torus
with~$w_0=\exp(2\pi i\psi)$, then
\[
\mathsf{FSC}=\frac{\vartheta_{00}(m\psi|\tau)}{\eta(\tau)}+\frac{\vartheta_{01}(m\psi|\tau)}{\eta(\tau)}
\quad\text{with}\quad
\tau=i\,\frac{m}{n}\left|\frac{\partial_zQ(-1,w_0)}{\partial_wQ(-1,w_0)}\right|\,.
\]}
\item{If~$Q(z,w)$ has a single (real) node at~$(-1,w_0)$ in the unit torus, then
\[
\mathsf{FSC}=\frac{\vartheta_{00}(\tau)}{\eta(\tau)}+\frac{\vartheta_{01}(\tau)}{\eta(\tau)}
\quad\text{with}\quad
\tau=i\,\frac{m}{n}\left|\frac{\partial^2_zQ(-1,w_0)}{\partial^2_wQ(-1,w_0)}\right|^{1/2}\,.
\]}
\end{enumerate}
\end{theorem}

Without surprise, the bulk free energy~$\mathbf{f}_0$ is that of the toric graph~$\widetilde{\G}\subset\T^2$.
On the other hand, the finite size correction~$\mathsf{fsc}$ is different from the one obtained in the toric case, and
does not seem to be related to it in a simple way (see also Section~\ref{sub:cons}).
The most remarkable aspect of this result is the universality of~$\mathsf{fsc}$, a term which only depends on the
phase of the model and on the (purely imaginary) conformal shape parameter~$\tau$.
To the best of our knowledge, this universality feature for the Klein
bottle cannot be derived from the corresponding result for the torus. 
Another fact worth mentioning is that even though~$Z_{mn}$ depends on both polynomials~$P$ and~$R$ (recall Theorem~\ref{thmintro:Z}), its asymptotic expansion is determined by~$P$.


As an illustration, we compute the explicit example of the square lattice.
Note that the~$(M\times N)$-square lattice in the Klein bottle is always ``locally bipartite'',
in the sense that all faces have even degree, but it is bipartite if and only if~$M$ is even and~$N$ odd:
in the other cases, there exist (homologically non-trivial) cycles of odd length.
Therefore, our result is a blend of the bipartite case (see Example~\ref{ex:sqbip})
and of the non-bipartite case (see Example~\ref{ex:fsc-sq}):
the finite size correction in the asymptotic expansion of the dimer partition function for
the~$(M\times N)$-square lattice in the Klein bottle is given by~$\mathsf{fsc}=\log\mathsf{FSC}$, with
\[
\mathsf{FSC}=\left\{
\begin{array}{ll}
\frac{\vartheta_{00}(\tau)}{\eta(\tau)} & \text{for~$M$ and~$N$ even,}\\
\frac{\vartheta_{00}(\tau)}{\eta(\tau)}+\frac{\vartheta_{01}(\tau)}{\eta(\tau)} & \text{for~$M$ even and~$N$ odd,}\\
\left(\frac{2\vartheta_{01}(2\tau)}{\eta(2\tau)}\right)^{1/2} & \text{for~$M$ odd and~$N$ even,}
\end{array}\right.
\]
where~$\tau=i\frac{Mx}{2Ny}$
and~$x$ (resp.~$y$) denotes the weight of the horizontal (resp. vertical) edges.
These functions are illustrated in Figure~\ref{fig:fsc}.
This recovers (and sometimes, corrects) the aforementioned results of~\cite{L-W,Lu-Wu99,IOH}. 
We also compute new examples, such as the hexagonal and triangular lattices, see
Examples~\ref{ex:hexa} and~\ref{ex:tri-sq}.

\begin{figure}[tb]
\labellist\small\hair 2.5pt
\pinlabel {even~$\times$ odd} at 640 550
\pinlabel {odd~$\times$ even} at 640 490
\pinlabel {even~$\times$ even} at 640 420
\endlabellist
\centering
\includegraphics[width=0.4\textwidth]{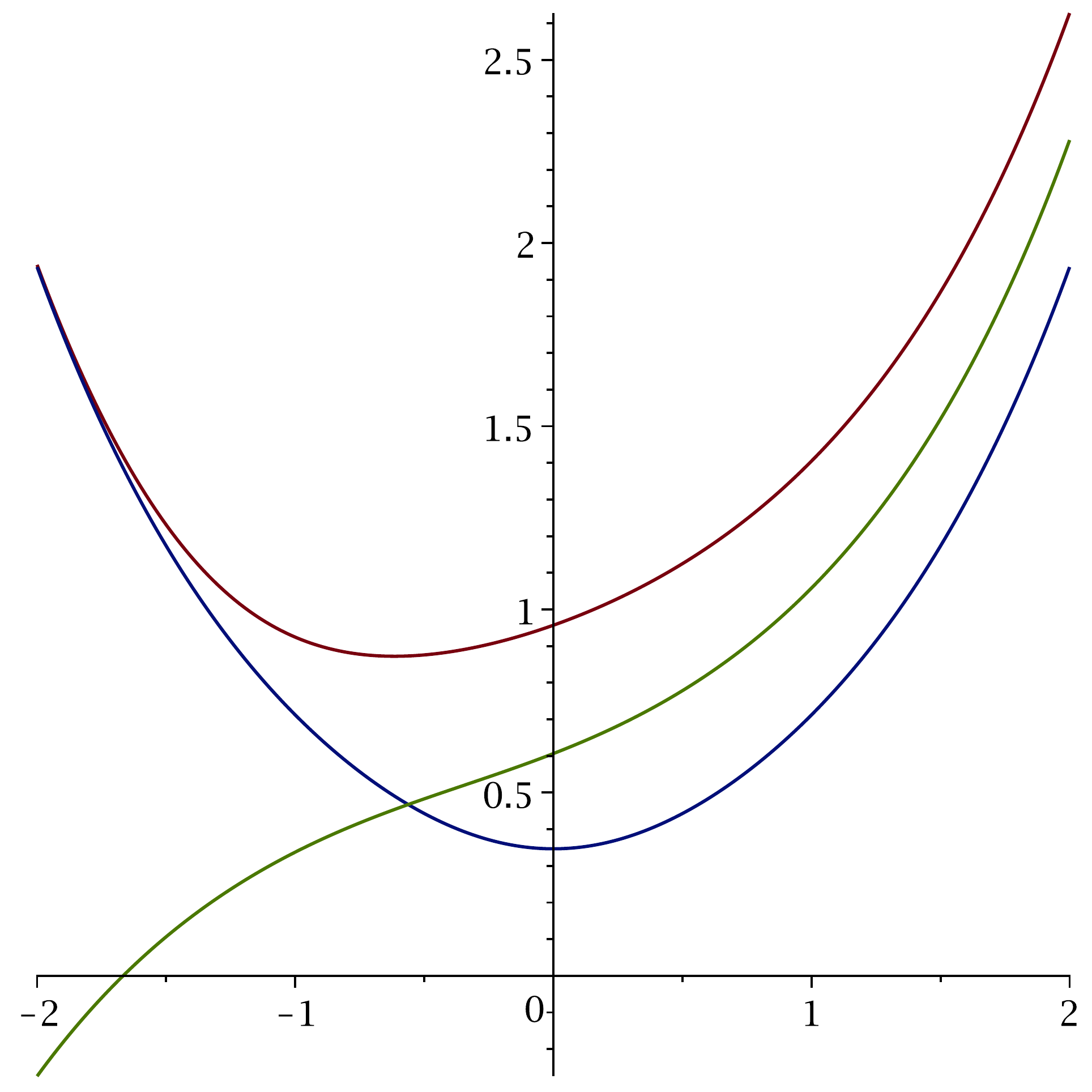}
\caption{Finite size corrections~$\mathsf{fsc}$ for the~$M\times N$ square lattice on the Klein bottle,
as a function of~$\log(\frac{Mx}{2Ny})$, with curves labeled according to the parity of~$(M,N)$.}
\label{fig:fsc}
\end{figure}

A couple of consequences are discussed in Section~\ref{sub:cons}. For instance, we show that the limit~$\lim_{m,n\to\infty}\frac{Z(\G_{mn})^2}{Z(\widetilde{\G}_{mn})}$
is universal, in the same sense as the finite size corrections. Also,
motivated by the CFT predictions of~\cite{BCN}, we compute the~$\tau_\text{im}\to\infty$
asymptotic of~$\mathsf{fsc}$ and check that the result is in agreement with~\cite{BCN}.
In particular, both cases~$(2)$ and~$(3)$ in Theorem~\ref{thmintro:as-bip}
yield the same value~$c=1$ for the central charge of a
conformal field theory describing the bipartite dimer model in the liquid phase.

\medskip

We now turn to the Ising model.

\subsection{Background on the Ising model}
\label{sub:intro-Ising}

The Ising model, first introduced by Lenz~\cite{Len20} in an attempt to understand Curie's
temperature for ferromagnets, is one of the most famous models in statistical physics.
It can be defined as follows.
Given a finite graph~$G=(V(G),E(G))$ endowed with a positive edge weight system~$J=(J_e)_{e\in E(G)}$, the
{\em energy\/} of a {\em spin configuration\/}~$\sigma\in\{-1,+1\}^{V(G)}$ is
defined by~$\mathcal{H}(\sigma)=-\sum_{e=uv\in E(G)}J_e\sigma_u\sigma_v$.
Fixing an {\em inverse temperature}~$\beta\ge 0$ determines a probability measure on the set~$\Omega(G)$ of spin configurations by
\[
\mu_{G,\beta}(\sigma)=\frac{e^{-\beta\mathcal{H}(\sigma)}}{Z_\beta^J(G)},\quad\text{with}
\quad Z_\beta^J(G)=\sum_{\sigma\in\Omega(G)}e^{-\beta\mathcal{H}(\sigma)}\,.
\]
The normalization constant~$Z_\beta^J(G)$ is called the {\em partition function\/}
of the {\em Ising model on~$G$ with coupling constants~$J$\/}.

Once again, we are interested in the asymptotic expansion of~$Z_\beta^J(G_{mn})$
for an arbitrary weighted graph~$G$ embedded in the torus or the Klein bottle.
As we shall see, this asymptotic expansion depends on the position of the parameter~$\beta$ with respect to
some {\em critical inverse temperature\/}, whose definition we now briefly recall.
Let~$\mathcal G$ be the infinite weighted planar graph obtained as the universal cover of~$G$ (i.e., as~$G_{mn}$ with~$m,n\to\infty$). Ising probability measures can be constructed on~$\mathcal G$ as limits of finite volume probability measures~\cite{McCW}: let us denote by~$\mu_{\mathcal G,\beta}^+$ the Ising measure at inverse temperature~$\beta$ on~$\mathcal G$ with $+$ boundary conditions. 
Let us assume that the embedded graph~$G$ is {\em non-degenerate}, i.e. that  the complement of its edges consists of topological discs. A Peierls argument~\cite{Pei36} and the GKS inequality~\cite{Gri67,KS68} then classically imply that the Ising model on~$\mathcal G$ exhibits a phase transition at some critical value~$\beta_c\in(0,\infty)$:
\begin{itemize}
\item for~$\beta<\beta_c$, we have~$\mu_{\mathcal G,\beta}^+(\sigma_v)=0$ for any~$v\in V(\mathcal G)$,
\item for~$\beta>\beta_c$, we have~$\mu_{\mathcal G,\beta}^+(\sigma_v)>0$ for any~$v\in V(\mathcal G)$.
\end{itemize}
We refer to~\cite{Li2,CDC} for a computation of the critical inverse temperature for arbitrary
non-degenerate doubly periodic weighted graphs.

There is a classical two-step method to apply dimer technology to the Ising model.
First, the Ising partition function can be expresses via the {\em high-temperature expansion\/}~\cite{vdW}
\[
Z_\beta^J(G)=\Big(\prod_{e\in E(G)}\cosh(\beta J_e)\Big)2^{|V(G)|}\sum_{\gamma\in\mathcal{E}(G)}\prod_{e\in\gamma}\tanh(\beta J_e)\,,
\]
where~$\mathcal{E}(G)$ denotes the set of even subgraphs of~$G$, that is, the set of subgraphs~$\gamma$ of~$G$ such that every vertex of~$G$ is adjacent to an even number of edges of~$\gamma$.
The second step is the so-called {\em Fisher correspondence\/}~\cite{Fisher}, which has a long and interesting history with several variations on the same theme (see e.g.~\cite[Section~3.1]{CCK} and references therein).
Let us consider a weighted graph~$(G,x)$ embedded in a surface, and denote by~$(G^F,x^F)$ the associated weighted graph obtained from~$(G,x)$ as illustrated in Figure~\ref{fig:Fisher}. As one easily checks, the high-temperature expansion of the Ising partition function on~$G$ is related to the dimer partition function on~$G^F$ via
\begin{equation}
\label{equ:Fisher}
2^{|V(G)|}\sum_{\gamma\in\mathcal{E}(G)}\prod_{e\in\gamma}x_e=Z(G^F,x^F)\,.
\end{equation}
In conclusion, we have the relation
\begin{equation}
\label{equ:Ising-dimer}
Z_\beta^J(G)=\Big(\prod_{e\in E(G)}\cosh(\beta J_e)\Big)Z(G^F,x^F)
\end{equation}
between the Ising partition function on~$G$ and the dimer partition function on~$G^F$, where the associated
weights are given by~$x_e=\tanh(\beta J_e)$. It allows to study the Ising model on a graph via the dimer model on the associated Fisher graph.

\begin{figure}[tb]
\labellist\small\hair 2.5pt
\pinlabel {$x_e$} at 145 290
\pinlabel {$x_e$} at 515 300
\pinlabel {$1$} at 515 260
\pinlabel {$1$} at 470 260
\pinlabel {$1$} at 495 200
\endlabellist
\centering
\includegraphics[width=0.5\textwidth]{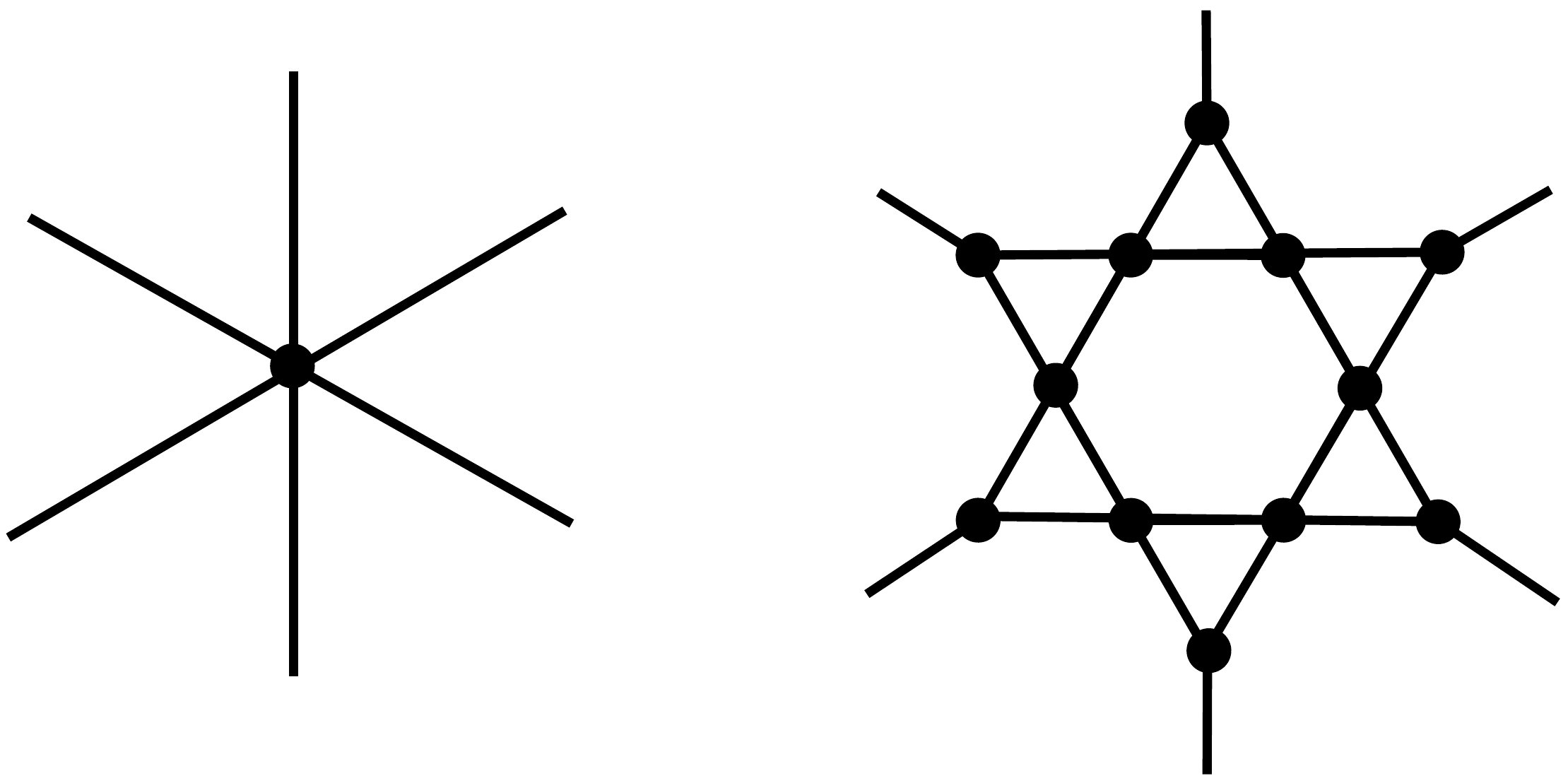}
\caption{The weighted graph~$(G,x)$ near a vertex, and the associated weighted graph~$(G^F,x^F)$ obtained via the Fisher correspondence.}
\label{fig:Fisher}
\end{figure}
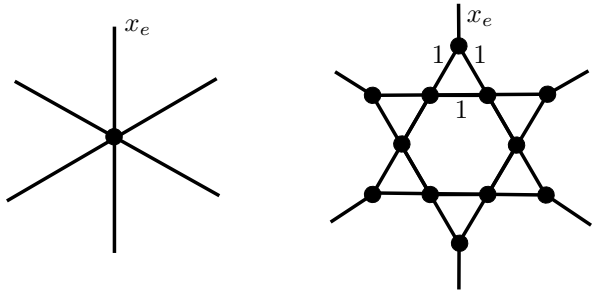

However, the Fisher graph fails to be bipartite, so the very powerful tools of~\cite{KOS} are {\em a priori\/} not available.
As it turns out, a more technical mapping from the Ising model to a bipartite
dimer model exists~\cite{Dub}, thus allowing to understand the spectral curve of the dimer model
on the Fisher graph~\cite{Li,CDC}: it is disjoint from the unit torus for~$\beta\neq\beta_c$,
and meets it at a single real positive node for~$\beta=\beta_c$.
The tools of~\cite{KSW} then make it a routine task to compute the asymptotic expansion of~$Z_\beta^J(G_{mn})$
in the toric case, see Remark~\ref{rem:Ising}~(i) below.

\medskip

What about the case of the Klein bottle? Here again, it is poorly understood: so far,
the only known results deal with the critical isotropic square lattice~\cite{LWIsing,Ch-P,Iz12},
and are partially contradictory
(see Example~\ref{ex:Ising} below, where we recover the formula
of~\cite{LWIsing} and of~\cite{Ch-P}).

\subsection{Result on the Ising model}
\label{sub:results-Ising}

We compute the asymptotic expansion of the Ising partition function for an arbitrary planar
graph with periodic-antiperiodic boundary conditions, as follows.

\begin{theorem}
\label{thmintro:Ising}
Let~$(G,J)$ be a non-degenerate weighted graph embedded in the Klein bottle, and let~$P(z,w)$ be the
characteristic polynomial of the associated Fisher graph~$\widetilde{G^F}\subset\mathbb{T}^2$.
Then, the Ising partition function on~$G_{mn}$ satisfies
\[
\log Z^J_\beta(G_{mn})= mn\frac{\mathbf{f}_0}{2} + \mathsf{fsc}+o(1)
\]
for~$m$ and~$n$ tending to infinity with~$n$ odd and~$m/n$ bounded below and above, with
\[
\mathbf{f}_0=2\sum_{e\in E(G)}\log\cosh(\beta J_e)+\frac{1}{2} \int_{\mathbb T^2} \log P(z,w)\frac{dz}{2\pi iz}\frac{dw}{2\pi iw}
\]
and~$\mathsf{fsc}=0$ in the subcritical regime~$\beta<\beta_c$,~$\mathsf{fsc}=\log(2)$ in the supercritical regime~$\beta>\beta_c$, and
\[
\mathsf{fsc}=\log\left(\left(\frac{\vartheta_{00}(\tau)}{\eta(\tau)}\right)^{1/2}+\left(\frac{\vartheta_{01}(\tau)}{2\eta(\tau)}\right)^{1/2}\right)\,,\quad\text{where}\quad\tau=i\,\frac{m}{n}\left|\frac{\partial^2_zP(-1,1)}{\partial^2_wP(-1,1)}\right|^{1/2}
\]
in the critical regime~$\beta=\beta_c$.\end{theorem}

Once again, the most remarkable feature of this result
is the universality of the finite-size correction term: it only depends on the regime of the model,
and at criticality, on the shape parameter~$\tau$ (see Figure~\ref{fig:Ising}).
The fact that the finite-size correction does not vanish for~$\beta>\beta_c$ might come as a surprise,
see Example~\ref{ex:tri-sq}.
Another fact worth mentioning is that for~$\beta=\beta_c$ and~$\tau_\text{im}\to\infty$, the 
asymptotic behavior of~$\mathsf{fsc}$ matches the CFT predictions of~\cite{BCN},
yielding the value~$c=\frac{1}{2}$ for the central charge of a conformal field theory describing the Ising model
(see Remark~\ref{rem:Ising}~(iii)).

\begin{figure}[tb]
\includegraphics[width=0.4\textwidth]{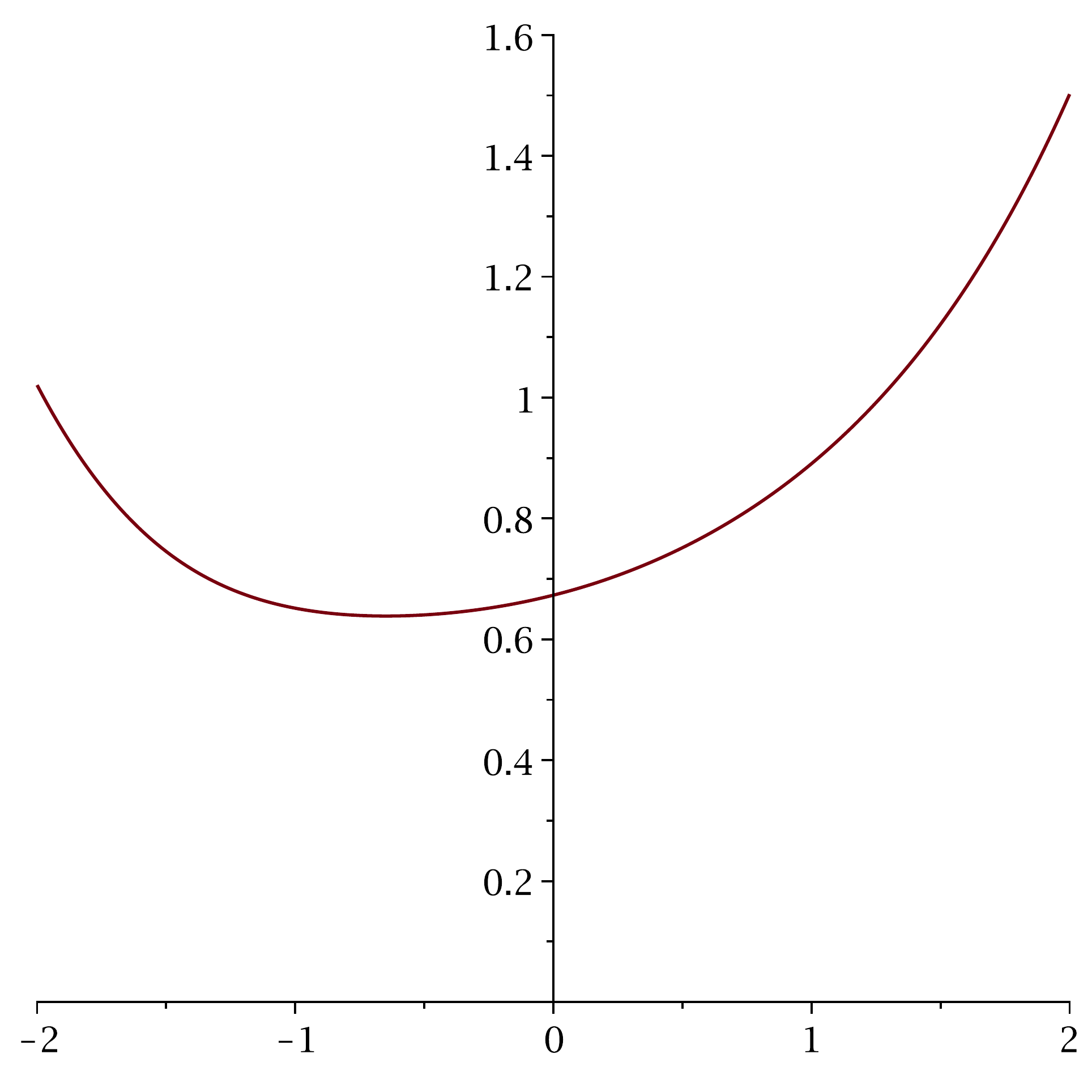}
\caption{Finite size correction term~$\mathsf{fsc}$ to the Ising partition function
for an arbitrary graph on the Klein bottle, in the critical regime~$\beta=\beta_c$,
as a function of~$\log(|\tau|)$.}
\label{fig:Ising}
\end{figure}

\subsection{Further directions}
\label{sub:rem}

The techniques developed in this article could
help to
produce further results, that we now briefly outline.

\subsubsection*{More general fundamental domains}

When enlarging the fundamental domain, we have restricted our attention to rectangular ones, i.e. domains
spanned by the vectors expressed as~$(n,0)$ and~$(0,m)$, with~$n$ odd,
in the basis of the plane given by the two vectors spanning the
fundamental domain of~$\G$. One could study the more general quadrangular domains spanned 
by vectors of the form~$(n,p)$ and~$(0,m)$ with~$n$ odd. These integers also describe a finite covering of
the Klein bottle by itself, so our methods apply.
For~$p\neq 0$, we expect more general finite size correction terms to appear, as in~\cite{KSW}, where the
most general form of finite coverings of the torus by itself is considered. However, we expect
the associated shape parameter~$\tau$ to remain purely imaginary, see Figure~\ref{fig:tau}.

Incidentally, it is an amusing exercice in combinatorial group theory to show that any
subgroup~$H$ of finite index of~$\pi_1(\K)=\left<a,b\,|\,aba^{-1}b\right>$ with~$H$ isomorphic to~$\pi_1(\K)$
is of the form~$H=\left<a^nb^p,b^m\right>$ with~$n$ odd. As a consequence, all finite coverings of the Klein
bottle by itself are of the form described above.

\subsubsection*{Loop statistics}

If~$\G$ is a bipartite graph embedded in a surface~$\Sigma$, then
the difference of two dimer configurations on~$\G$ gives a collection of oriented loops on~$\G$, and therefore a homology class in~$H_1(\Sigma;\Z)$.
In~\cite[Section~4]{KSW}, the authors consider the case of a bipartite graph~$\G\subset\Sigma=\mathbb{T}^2$ whose spectral curve intersects the unit torus in two distinct zeros, and describe the asymptotic distribution of these classes in~$H_1(\mathbb{T}^2;\Z)=\Z^2$ for the covers~$\G_{mn}$ of~$\G$, as~$m,n\to\infty$. We refer the reader to~\cite[Theorem~4]{KSW} for the more
precise, complete and general statement, and to~\cite{BdT09} for the previously studied case of the
hexagonal lattice.

In our context, one can fix a bipartite graph~$\G\subset\Sigma=\K$ whose characteristic polynomial has two conjugate zeros in the unit torus, and use the second case of Theorem~\ref{thmintro:as-bip} to determine the asymptotic distribution
of the homology classes in~$H_1(\K;\Z)=\Z\oplus\Z/2\Z$ for the
covers~$\G_{mn}$ of~$\G$, as~$m,n\to\infty$ with~$n$ odd.

\subsubsection*{Asymptotic expansion beyond the constant order}

In~\cite{IIH}, Ivashkevich, Izmailian and Hu study the asymptotic expansion of the dimer and Ising partition functions
for the square lattice embedded in the torus beyond the bulk free energy and the constant order term
(see also~\cite{BEP2}).
Actually, they consider the full asymptotic expansion
\[
\log Z_{mn}=mn\,\mathbf{f}_0+\mathsf{fsc}(\tau)+\sum_{p\ge 2}\frac{\mathbf{f}_p(\tau)}{(mn)^{1-p}}\,,
\]
and express all the terms~$\mathbf{f}_p$ using elliptic theta functions evaluated at the conformal parameter~$\tau$.
This work is extended in~\cite{IOH} to the dimer
model on the square lattice with various boundary conditions (see however Example~\ref{ex:sqbip} below).

It would be a worthy endeavor to use Theorem~\ref{thmintro:Z} and the good understanding of
the corresponding characteristic polynomials to try to compute subleading terms in the asymptotic expansion
of the bipartite dimer and Ising partition functions.
Determining which terms are universal and which ones are not would be of particular interest, see the introduction of~\cite{IIH}.

\subsubsection*{Non-bipartite dimers on the torus and Klein bottle}

Together with Theorem~2 of~\cite{KSW}, the present work settles the question
of finite-size corrections for bipartite dimers on the torus and Klein bottles.
For non-bipartite dimers however, there is still some work to be done.

As mentioned above, it is believed that
the corresponding spectral curve intersects
the unit torus in at most two points which are real positive nodes,
but this remains to be rigorously demonstrated.
Even then, the precise form of the finite-size
correction term would require some further study,
both in the torus and Klein bottle cases
(see e.g. Example~\ref{ex:tri-sq} below).

\subsubsection*{Beyond the flat case}

As stated at the beginning of this introduction, the asymptotic expansion
of the dimer and Ising partition functions is believed to take a
particularly simple form when the graph is embedded in a closed surface~$\Sigma$
with vanishing Euler characteristic~\cite{Ca-P}.
The torus and Klein bottles being the only two such surfaces,
this very favorable case is now settled.
But what about closed surfaces with non-vanishing Euler characteristic?
For square and triangular lattices on a genus~$2$ surface,
there is numerical evidence of the finite-size correction terms being
naturally expressed as sums of Riemann theta functions~\cite{CSC1,CSC2},
drawing a striking parallel with the conformal field theory results of~\cite{ABMNV}.
(See also~\cite{BLR} for recent advances on conformal invariance of
dimers on Riemann surfaces, and~\cite{I-K} for related results on
the determinant of discrete Laplacians.)

In the spirit of the present work, the study of irreducible representations
of finite quotients of~$\pi_1(\Sigma)$ could lead to an explicit
expression for the term~$\mathbf{f}_0$ in the asymptotic expansion
(work in progress with Adrien Kassel).
The presence of curvature makes it unlikely for our methods
alone to determine the universal finite-correction terms, and thus solve the outstanding problem stated above. However, it is our hope that
they will serve as stepping stone towards this goal.

\subsection*{Organisation of the article}

In Section~\ref{sec:char}, we describe the general set up of our work (a weighted graph~$\G$
embedded in the Klein bottle), define the associated Kleinian characteristic polynomial~$R$, and relate it to
the well-known toric characteristic polynomial~$P$ of~$\widetilde{\G}\subset\T^2$. In the bipartite case, we also study the location
of the roots of~$R$.

In Section~\ref{sec:enlarge}, we study how the polynomial~$R_{mn}$ of~$\G_{mn}$ can be computed from
the polynomials~$R$ and~$P$ (Theorem~\ref{thm:Rmn}), and prove Theorem~\ref{thmintro:Z}.
This requires subtle modifications of the twisted Kasteleyn matrices to ensure periodicity of
all the ingredients (Section~\ref{sub:period}), as well as the identification
of induced representations of~$\pi_1(\K)$ and their factorisation into irreducible ones (Section~\ref{sub:repr}).

Section~\ref{sec:as} deals with the resulting asymptotic expansion of the dimer and Ising partition functions.
We first consider the general (possibly non-bipartite) dimer model in Section~\ref{sub:as-gen},
before focusing on bipartite graphs and proving Theorem~\ref{thmintro:as-bip} in Section~\ref{sub:as-bip}.
A couple of consequences are discussed in Section~\ref{sub:cons}.
Finally, in Section~\ref{sub:as-Ising}, we study the Ising model and prove Theorem~\ref{thmintro:Ising}.

\subsection*{Acknowledgments}
The author would like to thank Adrien Kassel and Anh Minh Pham for helpful discussions,
and the anonymous referee for several sensible suggestions.
Partial support by the Swiss National Science Foundation is thankfully acknowledged.


\section{Characteristic polynomials for dimers on Klein bottles}
\label{sec:char}

The aim of this section is to define and study two classes of characteristic polynomials that play a fundamental role in this article. We begin in Section~\ref{sub:gen} by introducing the general setup of a weighted graph embedded in the Klein bottle, together with an appropriate orientation of its edges. This allows us to define the associated Kleinian characteristic polynomial~$R$ in Section~\ref{sub:Klein}, which by~\cite{Cim} can be used
to compute the corresponding dimer partition function.
In Section~\ref{sub:toric}, we recall the definition of the toric characteristic polynomial~$P$ of~\cite{KOS,KSW}, and explain its relation with its Kleinian counterpart. Finally, in Sections~\ref{sub:bip} and~\ref{sub:S}, we consider the special case of bipartite graphs and show how the results of~\cite{KOS} on~$P$ imply strong conditions on~$R$.

\subsection{The general setup}
\label{sub:gen}

Throughout this section,~$\G$ denotes a finite connected graph with vertex set~$V$ of even cardinality, edge set~$E$, and non-negative edge weights~$\nu=(\nu_e)_{e\in E}\in[0,\infty)^E$. This weighted graph is embedded in the Klein bottle~$\K$ in such a way that~$\K\setminus\G$ consists in topological discs. To represent the pair~$\G\subset\K$ conveniently, we cut~$\K$ open along two well-chosen oriented simple closed curves~$a,b$. In this way, one obtains a rectangular {\em fundamental domain\/}~$\mathcal{D}$ with horizontal sides corresponding to~$a$ and vertical sides corresponding to~$b$, as illustrated in Figure~\ref{fig:Klein} (left).
Let us write~$a'\subset\K$ for the simple closed curve corresponding to the horizontal line cutting~$\D$ in two, oriented from left to right. Note that~$a,a'$ generate the first homology group of~$\K$, and that~$a+a'$ is homologous to~$b$ and hence of order~$2$ in~$H_1(\K;\Z)$. We assume that~$\G$ is in general position with respect to~$\D$, in the sense that~$V$ is disjoint from~$a,a'$ and~$b$, while each edge of~$\G$ intersects each of these three curves at most once, and if so, transversally.

Let us write~$\widetilde{\G}\subset\T^2$ for the {\em orientation cover\/} of~$\G\subset\K$, i.e. the pair obtained by gluing two copies of the fundamental domain~$\D$ along a vertical side as illustrated in Figure~\ref{fig:Klein} (right).
We denote by~$\widetilde{V}$,~$\widetilde{E}$,~$\widetilde{\nu}$ and~$\widetilde{\D}$ the corresponding vertex set, edge set, edge weights and fundamental domain, respectively, and endow the torus~$\T^2$ with an orientation that is pictured counterclockwise.

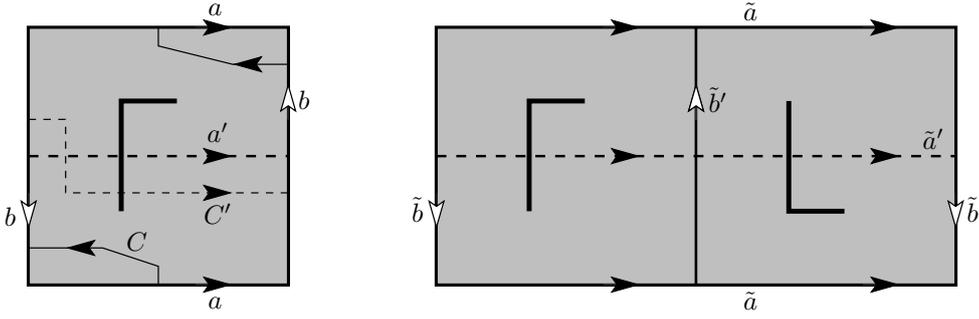
\begin{figure}[tb]
\labellist\small\hair 2.5pt
\pinlabel {$a$} at 150 -7
\pinlabel {$a'$} at 152 125
\pinlabel {$C'$} at 152 62
\pinlabel {$C$} at 90 40
\pinlabel {$a$} at 150 220
\pinlabel {$b$} at 218 150
\pinlabel {$b$} at -7 60
\pinlabel {$\tilde{b}$} at 305 65
\pinlabel {$\tilde{b}$} at 730 65
\pinlabel {$\tilde{b}'$} at 535 150
\pinlabel {$\tilde{a}$} at 560 218
\pinlabel {$\tilde{a}$} at 560 -5
\pinlabel {$\tilde{a}'$} at 700 120
\endlabellist
\centering
\includegraphics[width=0.8\textwidth]{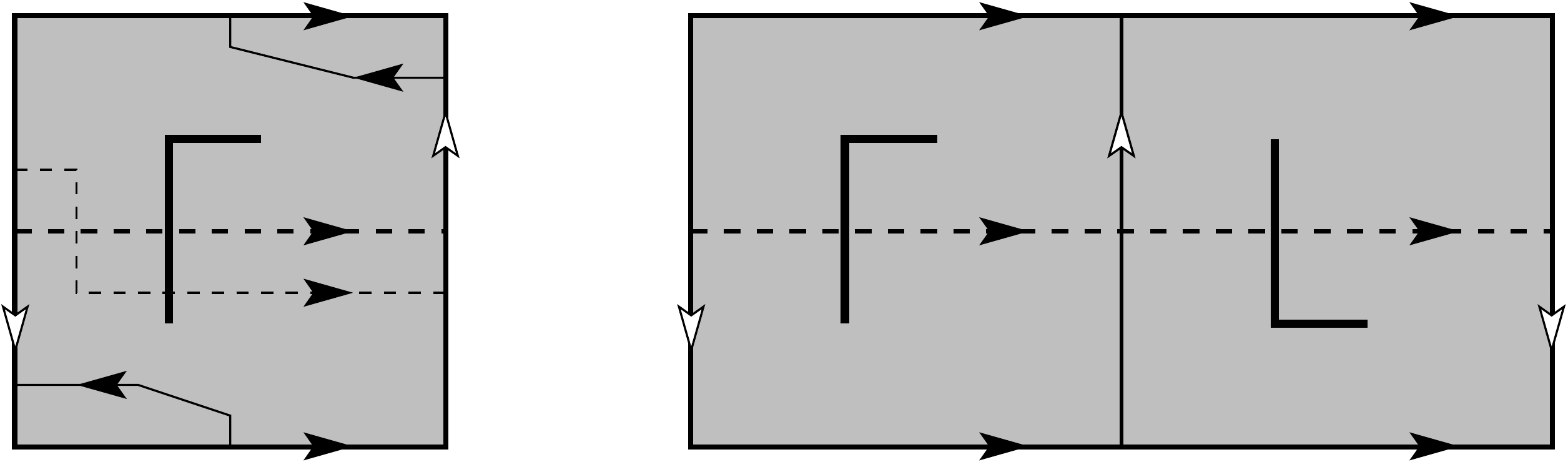}
\caption{Left: the graph~$\G\subset\K$ pictured in the fundamental domain~$\D$ delimited by the curves~$aba^{-1}b$, together with the associated curves~$C,C'\subset\G$. Right: the orientation cover~$\widetilde{\G}\subset\T^2$ pictured in the fundamental domain~$\widetilde{\D}$ delimited by the curves~$\tilde{b}\tilde{a}\tilde{b}^{-1}\tilde{a}^{-1}$.}
\label{fig:Klein}
\end{figure}

Following~\cite{Cim}, let us fix an orientation~$K$ on the edges of~$\G$ satisfying the following two conditions.
\begin{enumerate}[(i)]
\item{If one lifts~$K$ to~$\widetilde{\G}$ and then inverts the orientation of all the edges whose endpoints are both contained in the upper half part of~$\widetilde{\D}$, the resulting orientation~$\widetilde{K}$ is a {\em Kasteleyn orientation\/} on~$\widetilde{\G}\subset\T^2$. This means that each face of~$\widetilde{\G}\subset\T^2$ has an odd number of edges in its boundary that are oriented in the clockwise direction. (This makes sense as~$\T^2$ is oriented.)}
\item{Let~$C$ be the oriented closed curve in~$\G$ (homologous to~$a$) given by one edge~$e\in E$ intersecting~$a$ together with the oriented curve in~$\G$ joining the endpoints of~$e$, having~$a$ to its left in the lower half of~$\D$, to its right in the upper half of~$\D$, and meeting every vertex of~$\G$ adjacent to~$a$ on this side. Let~$C'\subset\G$ be defined in the same way for~$a'$. These curves are illustrated in the left part of Figure~\ref{fig:Klein}. We require the total number of edges of~$C$ and~$C'$ where~$K$ disagrees with the orientation of these curves to be even. (The curves~$C,C'$ are not uniquely defined, but it follows from the general theory of~\cite{Cim} that the parity of this number does not depend on the choices made.)}
\end{enumerate}

\begin{remarks}
\label{rem:K}
\begin{enumerate}[(i)]
\item{By~\cite[Theorem~4.3]{Cim}, such an orientation exists if and only if~$\G$ has an even number of vertices, which we assumed.}
\item{By this same result, such an orientation is unique up to flipping the edge orientations around a set of vertices, and up to reversing the orientations of all edges meeting the curve~$b$.}
\item{Deforming the curves~$a,a'$ (or equivalently, deforming the graph~$\G$ inside the fundamental
domain~$\mathcal{D}$) leads to natural local transformations of the orientation~$K$ which keep
the lifted orientation~$\widetilde{K}$ unchanged. (In case the deformation sweeps an odd number of vertices, one
also needs to invert all the edges meeting~$a$ in order for condition~(ii) above to hold.) On the other hand, the orientation~$K$ does not depend on the curve~$b\subset\K$.}
\end{enumerate}
\end{remarks}

Let us illustrate these conditions with the three lattices that provide the running examples of this article.

\begin{figure}[h]
\labellist\small\hair 2.5pt
\pinlabel {$a$} at 176 470
\pinlabel {$a'$} at 180 250
\pinlabel {$x_1$} at 58 368
\pinlabel {$x_2$} at 58 140
\pinlabel {$y_1$} at 138 271
\pinlabel {$y_2$} at 145 405
\pinlabel {$C$} at 1050 403
\pinlabel {$C'$} at 1055 227
\endlabellist
\centering
\includegraphics[width=0.8\textwidth]{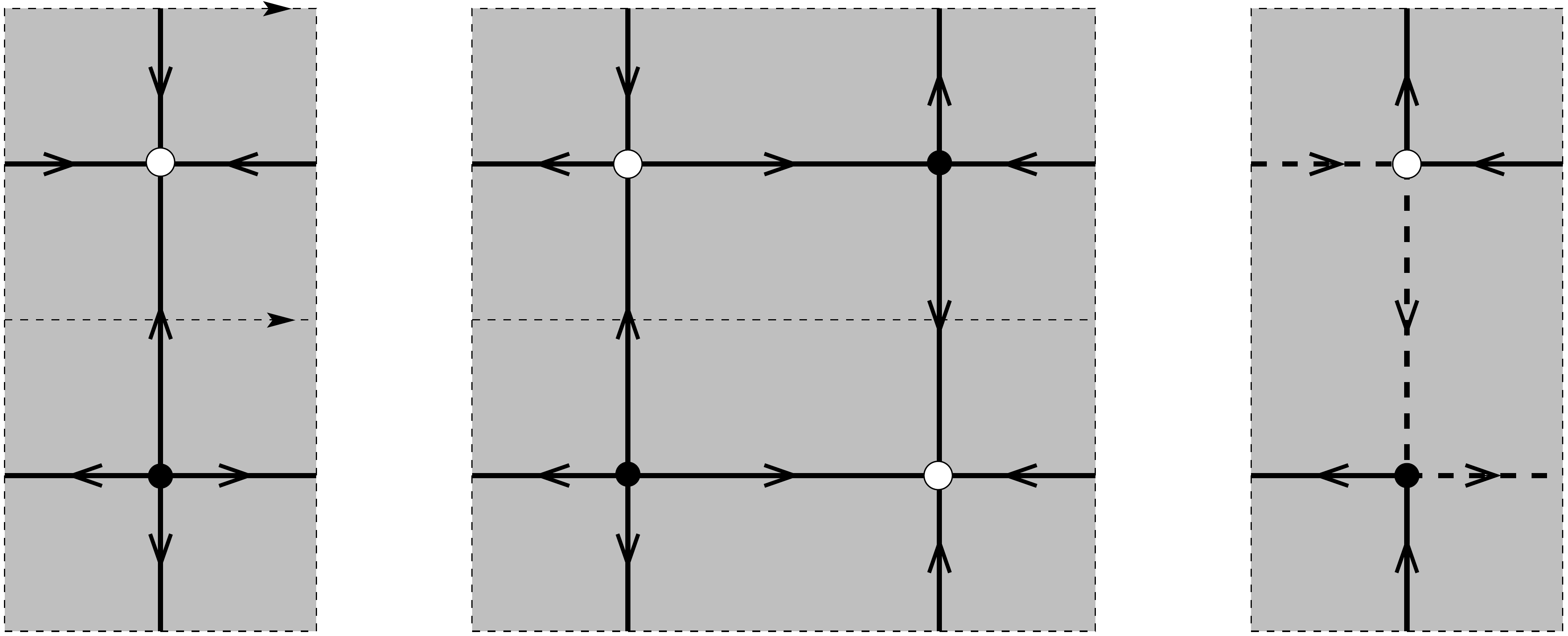}
\caption{The~$2\times 1$ square lattice of Example~\ref{ex:K}. Left: the orientation~$K$ on~$\G\subset\K$. Middle: The Kasteleyn orientation~$\widetilde{K}$ on~$\widetilde{\G}\subset\T^2$. Right: the oriented curves~$C$ and~$C'$ in~$\G$.}
\label{fig:ex1}
\end{figure}
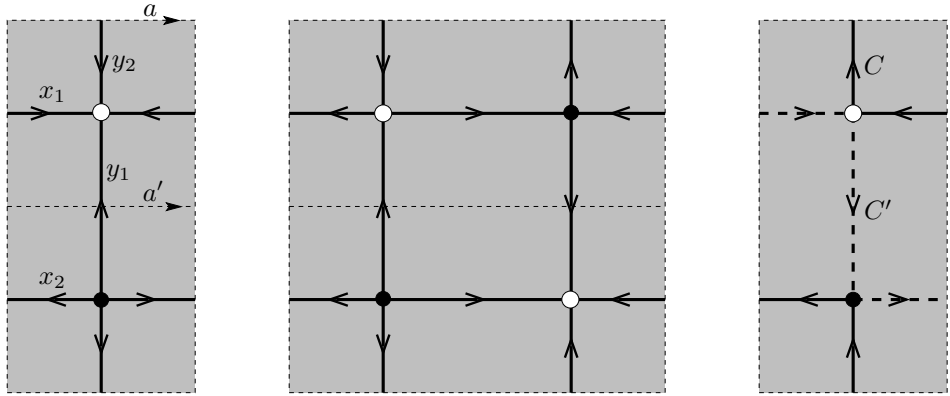

\begin{examples}
\label{ex:K}
Consider the~$2\times 1$ square lattice~$\G$ naturally embedded in the Klein bottle as illustrated in the left part of Figure~\ref{fig:ex1}. (The weights are represented as well for later use.) The orientation~$K$ pictured on this graph satisfies condition~(i) above: this is easily checked using the middle part of~Figure~\ref{fig:ex1} which represents the Kasteleyn orientation~$\widetilde{K}$ on~$\widetilde{\G}\subset\T^2$. (On this example, the two upper horizontal edges of~$\widetilde{\G}$ have their orientation inverted to obtain~$\widetilde{K}$.)
Finally, the orientation~$K$ satisfies condition~(ii) as well: the curves~$C,C'$ are illustrated on the right part of this same figure, and the number of edges of~$C$ and~$C'$ where~$K$ disagrees with the orientation of these curves is equal to~$2$. 

Consider now the~$1\times 2$ square lattice embedded in the Klein bottle as illustrated in the left part of Figure~\ref{fig:ex2-4}. (The curve~$a'$ is not represented as a straight line in order for it to intersect the graph transversally.) One easily checks that the orientation pictured there satisfies conditions~(i) and~(ii).

Finally, fundamental domains for the hexagonal and triangular lattices are pictured in the center and right parts of Figure~\ref{fig:ex2-4}, together with orientations satisfying both conditions.
\end{examples}

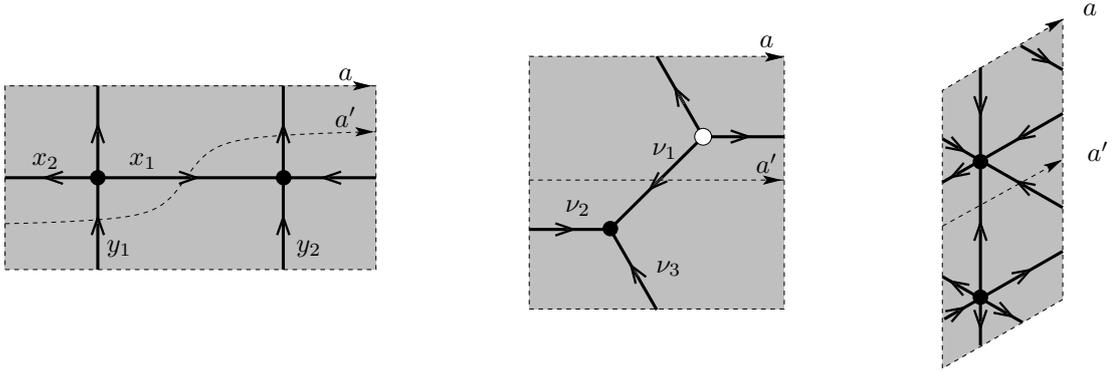
\begin{figure}[tb]
\labellist\small\hair 2.5pt
\pinlabel {$a$} at 420 365
\pinlabel {$a'$} at 420 312
\pinlabel {$x_1$} at 172 258
\pinlabel {$x_2$} at 53 258
\pinlabel {$y_1$} at 143 150
\pinlabel {$y_2$} at 375 150
\pinlabel {$a$} at 935 405
\pinlabel {$a'$} at 935 255
\pinlabel {$\nu_1$} at 810 270
\pinlabel {$\nu_2$} at 704 200
\pinlabel {$\nu_3$} at 815 125
\pinlabel {$a$} at 1330 445
\pinlabel {$a'$} at 1340 270
\endlabellist
\centering
\includegraphics[width=0.9\textwidth]{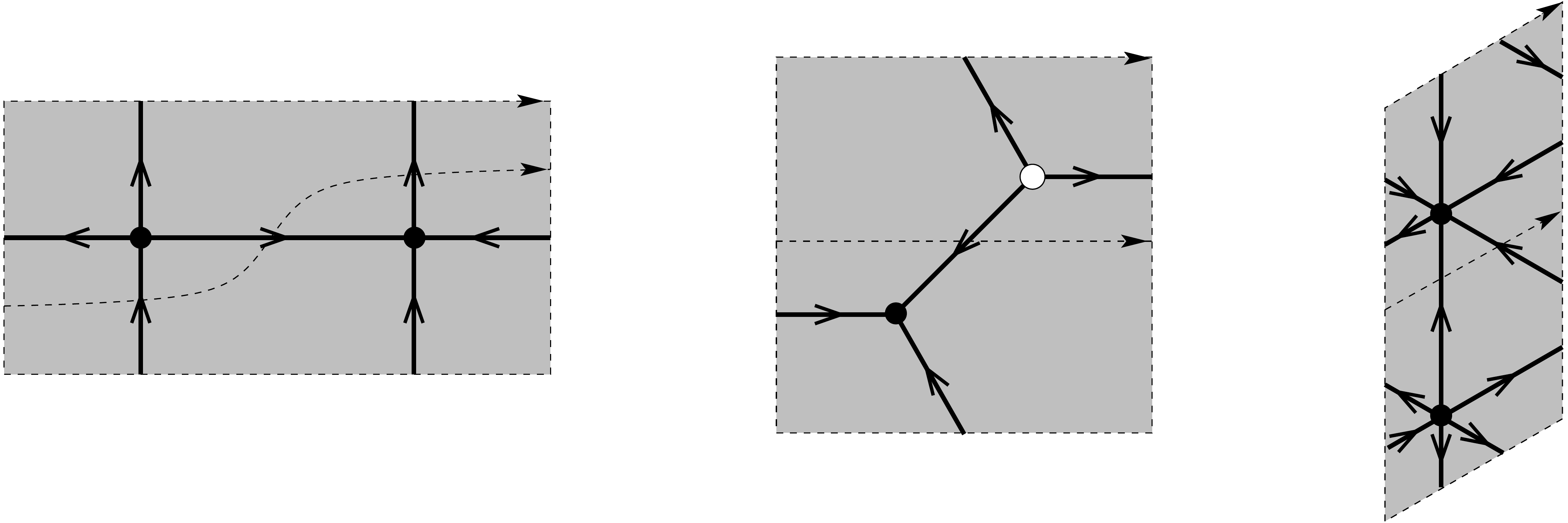}
\caption{The square, hexagonal and triangular lattices of Examples~\ref{ex:K} together with appropriate orientations.}
\label{fig:ex2-4}
\end{figure}

\begin{remark}
\label{rem:C}
Condition~(ii) is needed to normalize the orientation~$K$, as not all orientations satisfying condition~(i) can be used to compute the dimer partition function (see Proposition~\ref{prop:Pf} below). For simplicity, we shall furthermore assume that the curve~$C$ (resp.~$C'$) used in this normalization can be chosen disjoint from~$a'$ (resp.~$a$). This is easily seen to hold in the first, third and fourth lattices of Example~\ref{ex:K},
and can be assumed without loss of generality via vertical extension of the fundamental domain.
However, if one considers graphs that are ``too small'', such as the second lattice of Example~\ref{ex:K},
then this assumption is not satisfied. As we progress, we shall explain the small modifications that need to be made in such cases (see Remarks~\ref{rem:Kmn},~\ref{rem:perK} and~\ref{rem:as-gen}~(iii) below).
\end{remark}

\subsection{The Kleinian characteristic polynomial}
\label{sub:Klein}

We now have a weighted graph~$\G\subset\K$ together with an orientation~$K$ of its edges.
Let us order the vertex set~$V$ of~$\G$ and for~$z\in\C^*$ and~$w=\pm 1$, denote
by~$A(z,w)=\left(A(z,w)_{v,v'}\right)_{v,v'\in V}$ the associated (complex-valued, twisted) {\em Kasteleyn matrix\/} given by
\[
A(z,w)_{v,v'}=\sum_{e=(v,v')}\varepsilon_{vv'}^K(e)\,i^{e\cdot a+e\cdot a'}\,\nu_e\,z^{e\cdot b} w^{e\cdot a}\,,
\]
where the sum is over all the oriented edges~$e$ from~$v\in V$ to~$v'\in V$, the sign~$\varepsilon_{vv'}^K(e)$ is~$+1$ if~$K$ orients~$e$ from~$v$ to~$v'$ and~$-1$ otherwise, and~$e\cdot c$ denotes the {\em algebraic intersection number\/} of the oriented edge~$e$ with the oriented curve~$c\in\{a,a',b\}$ in~$\K$. Here some caution is needed. We shall say that~$e\cdot b=+1$ if~$e$ crosses the vertical side of the fundamental domain from left to right, while~$e\cdot b=-1$ it crosses it from right to left and~$e\cdot b=0$ if~$e$ and~$b$ are disjoint. (An integral intersection number with~$b$ is indeed well-defined in  the Klein bottle.) On the other hand, the intersection number~$e\cdot a$ does not carry a sign: we have~$e\cdot a=1$ if~$e$ and~$a$ meet (once) and~$e\cdot a=0$ else. The same holds for~$a'$. (Only a~$\Z/2\Z$-valued intersection number with~$a$ and~$a'$ is defined in~$\K$.)

\begin{definition}
The {\em characteristic polynomial\/} of~$\G\subset\K$ is
\[
R(z,w)=\det A(z,w)\in\C[z^{\pm 1},w]\,.
\]
\end{definition}

Several remarks are in order.

\begin{remarks}
\label{rem:R}
\begin{enumerate}[(i)]
\item This polynomial is really an element of~$\C[z^{\pm 1},w]/(w^2-1)$, the group ring of~$H_1(\K;\Z)=\Z\oplus\Z/2\Z$. In other words, it only carries relevant information for~$w=\pm 1$ and can be considered as two 1-variable Laurent polynomials~$R(z,1),R(z,-1)\in\C[z^{\pm 1}]$. 
\item Given~$\G\subset\K$ with curves~$a,a',b$, the polynomial~$R(z,w)$ is uniquely defined up to complex conjugaison. This follows from Remark~\ref{rem:K}~(ii) together with the fact that~$a+a'$
is homologous to~$b$. Moreover, Remark~\ref{rem:K}~(iii) implies that deforming the curves~$a,a'$ (and~$b$) leads to
the additional transformation~$R(z,w)\mapsto -R(z,-w)$.

\item The polynomial~$R(iz,w)$ lies in~$\R[z^{\pm 1},w]$. Indeed,~$e\cdot a+e\cdot a'$ and~$e\cdot b$ have the same parity since~$[a+a']=[b]\in H_1(\K;\Z/2\Z)$. 
Equivalently, we have~$R(-z,w)=\overline{R(z,w)}$.
\end{enumerate}
\end{remarks}

We now compute this characteristic polynomial for our running examples.

\begin{examples}
\label{ex:R}
Let us order the vertices of the~$2\times 1$ square lattice of Figure~\ref{fig:ex1} from bottom to top. The corresponding Kasteleyn matrix is given by
\[
A(z,w)=\begin{pmatrix}
0&iy_1+iy_2w+x_1z+x_2z^{-1}\cr
-iy_1-iy_2w-x_1z^{-1}-x_2z&0
\end{pmatrix}\,,
\]
so~$R(z,w)=S(z,w)S(z^{-1},w)$ with~$S(z,w)=iy_1+iy_2w+x_1z+x_2z^{-1}$.

Ordering the vertices of the~$1\times 2$ square lattice of Figure~\ref{fig:ex2-4} from left to right, we have
\[
A(z,w)=\begin{pmatrix}
i^2y_1w-i^2y_1w&ix_1+x_2z^{-1}\cr
-ix_1-x_2z&i^2y_1w-i^2y_1w
\end{pmatrix}
=
\begin{pmatrix}
0&ix_1+x_2z^{-1}\cr
-ix_1-x_2z&0
\end{pmatrix}
\,.
\]
(Note that loops are counted twice, as the sum is over all {\em oriented\/} edges from~$v$ to~$v'$.) Hence, we obtain~$R(z,w)=(ix_1+x_2z)(ix_1+x_2z^{-1})=(x_2^2-x_1^2)+ix_1x_2(z+z^{-1})$.

For the hexagonal lattice of Figure~\ref{fig:ex2-4}, we have
\[
A(z,w)=\begin{pmatrix}
0&i\nu_1+\nu_2z+i\nu_3w\cr
-i\nu_1-\nu_2z^{-1}-i\nu_3w&0
\end{pmatrix}\,,
\]
so~$R(z,w)=S(z,w)S(z^{-1},w)$ with~$S(z,w)=i\nu_1+\nu_2z+i\nu_3w$.

Finally, for the isotropic triangular lattice of Figure~\ref{fig:ex2-4},
we have
\[
A(z,w)=\begin{pmatrix}
iz^{-1}-iz&-i-iw-z-z^{-1}\cr
i+iw+z+z^{-1}&iwz-iwz^{-1}
\end{pmatrix}\,,
\]
leading to~$R(z,w)=(z^2+z^{-2}+2iz+2iz^{-1})(1+w)-(1+w)^2+2(1-w)$.
\end{examples}

Using this characteristic polynomial, the Pfaffian formula of~\cite{Cim} (see also~\cite{Tes}) can now be reformulated in the following way.

\begin{proposition}
\label{prop:Pf}
The dimer partition function of~$\G$ is given by
\[
Z=\big|{\rm Im}\big(R(\pm 1,1)^{1/2}\big)\big|+\big|{\rm Re}\big(R(\pm 1,-1)^{1/2}\big)\big|\,.
\]
\end{proposition}
\begin{proof}
Let us first assume that~$K$ is such that the number~$n^K(C)$ of edges of~$C$ where~$K$ disagrees with the orientation of~$C$ is odd. (By condition~(ii) above,~$n^K(C')$ is then odd as well.) Then, we are precisely in the setting of~\cite[Theorem~6.3]{Cim}, which in our case amounts to the Pfaffian formula
\[
Z=\left|{\rm Im}(\Pf(A(1,1)))+{\rm Re}(\Pf(A(1,-1)))\right|\,.
\]
Considering the expansion of~$\Pf(A(1,\pm 1)$ given in~\cite[p.174]{Cim}, on easily checks that the formula
\[
Z=\left|{\rm Im}(\Pf(A(1,1)))\right|+\left|{\rm Re}(\Pf(A(1,-1)))\right|
\]
holds as well, as there is no possible cancellation between these two terms. The statement now follows from the identity~$\Pf(A(1,w))^2=\det(A(1,w))=R(1,w)$ for~$w=\pm 1$, together with the equality~$R(-z,w)=\overline{R(z,w)}$ of Remark~\ref{rem:R}~(iii).
If~$K$ is such that~$n^K(C)$ and~$n^K(C')$ are both even, then it can be changed to~$K'$ with~$n^{K'}(C)$ and~$n^{K'}(C')$ both odd by reversing the orientation of all edges crossing~$b$. By Remark~\ref{rem:R}~(iii), this amounts to replacing~$R(1,w)$ by~$R(-1,w)=\overline{R(1,w)}$ for~$w=\pm 1$, and the proposition still holds.
\end{proof}

\subsection{The toric characteristic polynomial}
\label{sub:toric}

A second polynomial plays a crucial role in our study: it is the characteristic polynomial introduced for bipartite graphs by Kenyon, Okounkov and Sheffield in their seminal paper~\cite{KOS}, and extended to general toric graphs by Kenyon, Sun and Wilson~\cite{KSW}. Let us recall its definition in our context.

For~$z,w\in\C^*$, let~$\widetilde{A}(z,w)$ be the Kasteleyn matrix associated with the weighted graph~$\widetilde{\G}\subset\T^2$ and the Kasteleyn orientation~$\widetilde{K}$ (recall condition~(i) above). In other words, for~$v,v'\in\widetilde{V}$, we have
\[
\widetilde{A}(z,w)_{v,v'}=\sum_{e=(v,v')}\varepsilon_{vv'}^{\widetilde{K}}(e)\,\widetilde{\nu}_e\,z^{e\cdot\tilde{b}} w^{e\cdot\tilde{a}}\,,
\]
where~$e\cdot\tilde{a}\in\Z$ denotes the algebraic intersection number in~$\T^2$ of the oriented edge~$e$ of~$\widetilde{\G}$ with the oriented simple closed curve~$\tilde{a}$, and similarly for~$e\cdot\tilde{b}$ (recall Figure~\ref{fig:Klein}, right). Concretely, we have~$e\cdot\tilde{a}=+1$ if~$e$ crosses~$\tilde{a}$ from bottom to top~($e\cdot\tilde{a}=-1$ if~$e$ goes top to bottom), and~$e\cdot\tilde{b}=+1$ if~$e$ crosses~$\tilde{b}$ from left to right~($e\cdot\tilde{b}=-1$ if~$e$ goes right to left).

\begin{definition}
The {\em characteristic polynomial\/} of~$\widetilde{\G}\subset\T^2$ is
\[
P(z,w)=\det\widetilde{A}(z,w)\in\R[z^{\pm 1},w^{\pm 1}]\,.
\]
\end{definition}

\begin{remarks}
\label{rem:P}
\begin{enumerate}[(i)]
\item{In the general setting of~\cite{KOS,KSW}, this characteristic polynomial depends on the choice of a Kasteleyn orientation. In our setting, it is defined with respect to a specific orientation (recall Remark~\ref{rem:K})
and therefore does not depend on such a choice.}
\item{By definition, the matrix~$\widetilde{A}$ is of even dimension and satisfies~$\widetilde{A}(z,w)^T=-\widetilde{A}(z^{-1},w^{-1})$, which implies the equality~$P(z,w)=P(z^{-1},w^{-1})$. Since~$P(z,w)$ has real coefficients, this in turn implies that for~$z,w\in S^1$, we have~$P(z,w)=P(z^{-1},w^{-1})=P(\overline{z},\overline{w})=\overline{P(z,w)}$ which is real.}
\end{enumerate}
\end{remarks}

\begin{examples}
\label{ex:P}
For the~$2\times 1$ square lattice of Figure~\ref{fig:ex1} and the hexagonal lattice of Figure~\ref{fig:ex2-4}, we have 
\[
\widetilde{A}(z,w)=\begin{pmatrix}
0&\widetilde{A}_{\circ\bullet}(z,w)\cr
-\widetilde{A}_{\circ\bullet}(z^{-1},w^{-1})^T&0
\end{pmatrix}\,,
\]
with
\[
\widetilde{A}_{\circ\bullet}(z,w)=\begin{pmatrix}
y_1+y_2w^{-1}&x_1+x_2z^{-1}\cr
-x_2-x_1z&y_1+y_2w
\end{pmatrix}
\quad\text{and}\quad
\widetilde{A}_{\circ\bullet}(z,w)=\begin{pmatrix}
\nu_1+\nu_3w&-\nu_2\cr
\nu_2z&\nu_1+\nu_3w^{-1}
\end{pmatrix}\,,
\]
respectively. Hence, we obtain a factorization~$P(z,w)=Q(z,w)Q(z^{-1},w^{-1})$ with
\[
Q(z,w)= y_1^2+y_2^2+2x_1x_2+y_1y_2(w+w^{-1})+x_1^2z+x_2^2z^{-1}
\]
and
\[
Q(z,w)= \nu_1^2+\nu_3^2+\nu_1\nu_3(w+w^{-1})+\nu_2^2z\,,
\]
respectively.

Finally, for the~$1\times 2$ square lattice with edge weights~$x_1=x_2=:x$ and~$y_1=y_2=:y$, we get
\[
P(z,w)=y^4(w-w^{-1})^4-4x^2y^2(w-w^{-1})^2+x^4(2+z+z^{-1})\,,
\]
while the isotropic triangular lattice yields
\[
P(z,w)=(z^2+z^{-2})(w+w^{-1}+2)+10(w+w^{-1})+w^2+w^{-2}+34\,.
\]
\end{examples}

The following proposition appears to be folklore (see~\cite[p.~974]{KSW}), holds for arbitrary toric graphs (not necessarily covers of graphs on the Klein bottle), and is immediate in the case of bipartite graphs. As we were unable to find a proof of the general case in the literature, we include one here for completeness.

\begin{proposition}
\label{prop:Ppos}
Given any toric graph~$G\subset\T^2$ and any non-negative edge weights~$\nu\in[0,\infty)^E$, the corresponding characteristic polynomial~$P$ takes non-negative values on~$S^1\times S^1$.
\end{proposition}

\begin{proof}
First note that~$P=P_{(G,\nu)}$ is left unchanged by adding edges of weight~$0$, and by subdividing an edge~$e$ of weight~$\nu_e$ into three edges of weights~$\nu_e$,~$1$, and~$1$. Using these transformations, any weighted graph~$(G,\nu)\subset\T^2$ can be modified to obtain a weighted graph with the same polynomial, but admitting a perfect matching. Using the continuity of~$\nu\mapsto P_{(G,\nu)}(z,w)$, we can assume that all edge weights are positive. Hence, it can be assumed that the dimer partition function of~$(G,\nu)$ does not vanish. Note that for~$z,w=\pm 1$, we have~$P(z,w)=\Pf(\widetilde{A}(z,w))^2\ge 0$. An appropriate linear combination of these four Pfaffians gives the dimer partition function~\cite{Tes,C-R}, which we assumed not to vanish. Hence, we have~$P(z,w)>0$ for at least one~$(z,w)$ in~$\{\pm 1\}^2$. Applying this fact to the~$\Z/m\Z\times\Z/n\Z$ covering of~$G\subset\T^2$ and using~\cite[Theorem~3.3]{KOS} (see also Equation~\eqref{equ:Pmn} below), we find that~$P(z,w)\neq 0$ for~$(z,w)$ in some dense subset~$T\subset S^1\times S^1$. For any fixed element~$(z,w)$ of~$T$, we hence have the equality
\[
\{\nu\in(0,\infty)^E\,|\,P_{(G,\nu)}(z,w)\ge 0\}=\{\nu\in(0,\infty)^E\,|\,P_{(G,\nu)}(z,w)> 0\}\,,
\]
which is open and closed in~$(0,\infty)^E$ by continuity of~$\nu\mapsto P_{(G,\nu)}(z,w)$. It is also non-empty, as the choice for~$\nu$ of the indicator function of a perfect matching gives the value~$P_{(G,\nu)}(z,w)=1$. By connectedness of~$(0,\infty)^E$, we conclude that~$P_{(G,\nu)}(z,w)\ge 0$ for all~$\nu$ and any fixed~$(z,w)\in T$. The statement follows from the density of~$T$ and the continuity of~$(z,w)\mapsto P(z,w)$.
\end{proof}

We need further properties of these polynomials for covers of graphs on the Klein bottle.

\begin{proposition}
\label{prop:P}
\begin{enumerate}[(i)]
\item{$P(z,w)=P(z,w^{-1})\in\R[z^{\pm 1},w^{\pm 1}]$.}
\item{For~$w=\pm 1$, we have the equality~$P(z,w)=R(z^{1/2},w)R(-z^{1/2},w)$ in~$\C[z^{\pm 1/2}]$.}
\item{$P(1,w)=|R(\pm 1,w)|^2$ for~$z\in\C^*$ and~$w=\pm 1$.}
\end{enumerate}
\end{proposition}

\begin{proof}
To show the first point, fix~$z,w\in\C^*$ and an arbitrary square root~$z^{1/2}$ of~$z$. First observe that~$P(z,w)=\det\widetilde{A}(z,w)$ is left unchanged when replacing~$z^{e\cdot\tilde{b}}$ by~$z^{1/2(e\cdot\tilde{b}-e\cdot\tilde{b}')}$, where~$\tilde{b},\tilde{b}'\subset\T^2$ denote the two lifts of~$b\subset\K$ (recall Figure~\ref{fig:Klein}). Also, multiplying by~$i$ the rows and columns of~$\widetilde{A}(z,w)$ corresponding to a vertex in the upper half of~$\widetilde{\D}$ amounts to multiplying its determinant by~$(-1)^{|V|}=1$, so the resulting matrix~$\widetilde{A}'(z,w)$ still has determinant equal to~$P(z,w)$. However, this new matrix is now symmetric in the following sense: if~$f$ denotes the involution of~$\C^{\widetilde{V}}$ corresponding to the non-trivial deck transformation of the covering~$(\T^2,\widetilde{\G})\to(\K,\G)$, we have the~$f\widetilde{A}'(z,w)f=\widetilde{A}'(z,w^{-1})$. The equality~$P(z,w)=P(z,w^{-1})$ follows.

To prove the second point, consider again the modified matrix~$\widetilde{A}'(z,w)$ above, which for~$w=\pm 1$ is invariant under the involution~$f$ of~$\C^{\widetilde{V}}$. Following the standard arguments of~\cite[Theorem~3.3]{KOS}, we obtain that in the right basis,~$\widetilde{A}'(z,w)$ is given by~$A(z^{1/2},w)\oplus A(-z^{1/2},w)$. The statement follows.

The third point is a consequence of the second one and of Remark~\ref{rem:R}~(iii).
\end{proof}

\subsection{The bipartite case: basics}
\label{sub:bip}

Let us now assume that the graph~$\G$ is {\em bipartite\/}, i.e. that its vertices can be partitioned into two sets (say, sets~$B$ and~$W$ of {\em black\/} and {\em white\/} vertices, respectively) so that no edge joins two vertices of the same set. For such a graph to admit a perfect matching, it is necessary to have~$|B|=|W|$, which we assume.
In such a case,
the vertices can be ordered so that the matrix~$A(z,w)$ is block off-diagonal,
leading to the characteristic polynomial factorizing as
\begin{equation}
\label{equ:factorRS}
R(z,w)=S(z,w)S(z^{-1},w)\,,
\end{equation}
where~$S(z,w)\in\C[z^{\pm 1},w]$ is the {\em bipartite characteristic polynomial\/} of~$\G\subset\K$.

\begin{remark}
\label{rem:S}
Given~$\G\subset\K$ and curves~$a,a',b$, the polynomial~$S$ is well-defined up to a sign and complex conjugation:
as before, this follows from Remark~\ref{rem:K}~(ii) together with the fact that~$a+a'$ is homologous to~$b$.
Moreover, by Remark~\ref{rem:K}~(iii), deforming the curves~$a,a',b$ leads to the additional
transformations~$S(z,w)\mapsto iS(z,-w)$ and~$S(z,w)\mapsto z^{\pm 1}S(z,w)$.
\end{remark}

As one easily checks, such transformations are coherent with the properties listed below, which can be obtained using
Remark~\ref{rem:R}~(iii) and Proposition~\ref{prop:Pf}.

\begin{proposition}
\label{prop:S}
The polynomial~$S(z,w)$ satisfies the equalities
\begin{enumerate}[(i)]
\item{$S(-z,w)=\pm\overline{S(z,w)}\in\C[z^{\pm 1},w]$, and}
\item{$Z=\left|{\rm Im}(S(\pm 1,1))\right|+\left|{\rm Re}(S(\pm 1,-1))\right|$.\qed}
\end{enumerate}
\end{proposition}

A bipartite structure on~$\G\subset\K$ lifts to a bipartite structure~$\widetilde{V}=\widetilde{B}\sqcup\widetilde{W}$ on~$\widetilde{\G}\subset\T^2$. As above, we then have a factorization
\begin{equation}
\label{equ:factorPQ}
P(z,w)=Q(z,w)Q(z^{-1},w^{-1})\in\R[z^{-1},w^{-1}]\,,
\end{equation}
where~$Q(z,w)$ is the {\em bipartite characteristic polynomial\/} of~$\widetilde{\G}\subset\T^2$, defined as the determinant of corresponding bipartite Kasteleyn matrix~$\widetilde{A}_{\circ\bullet}(z,w)$.

\begin{remark}
\label{rem:Q}
The polynomial~$Q$ is uniquely defined from~$\G\subset\K$ and the curves~$a,a',b$, up to a global sign depending
on the ordering of the vertex set~$\widetilde{V}$.
Moreover, deforming the curves~$a,a'$ and~$b$ leads to the transformations~$Q(z,w)\mapsto Q(z,-w)$
and~$Q(z,w)\mapsto z^{\pm 1}Q(z,w)$.
\end{remark}

In this bipartite case, the proof of Proposition~\ref{prop:P} extends to give the following statement.

\begin{proposition}
\label{prop:Q}
The polynomial~$Q(z,w)\in\R[z^{\pm 1},w^{\pm 1}]$ satisfies the equalities
\begin{enumerate}[(i)]
\item{$Q(z,w)=Q(z,w^{-1})\in\R[z^{\pm 1},w^{\pm 1}]$, and}
\item{$Q(z,w)=\pm S(z^{1/2},w)S(-z^{1/2},w)\in\C[z^{\pm 1/2}]$ for~$w=\pm 1$.}\qed
\end{enumerate}
\end{proposition}

By Remark~\ref{rem:S} and Proposition~\ref{prop:S}~(i), one can choose the curves~$a,a'$ so
that
\begin{equation}
\label{equ:S}
S(-z,w)=\overline{S(z,w)}\in\C[z^{\pm 1},w]\,.
\end{equation}
Moreover, by Remark~\ref{rem:Q} and
Proposition~\ref{prop:Q}~(ii), one can order the vertex set~$\widetilde{V}$ so that the equality
\begin{equation}
\label{equ:QS}
Q(z,w)=S(z^{1/2},w)S(-z^{1/2},w)\in\C[z^{\pm 1/2}]
\end{equation}
holds for~$w=\pm 1$. As a consequence, we also have the equality~$Q(z,w)=|S(z^{1/2},w)|^2$
for~$w=\pm 1$. From now, one we will assume these normalizations of~$Q$ and~$S$.

\subsection{The bipartite case: roots of~$S$}
\label{sub:S}

This section contains the first technical results of this article. They play a crucial role in our proof of Theorem~\ref{thmintro:as-bip} and of Theorem~\ref{thmintro:Ising}.

For a bipartite toric graph~$\widetilde{\G}\subset\T^2$, there is a natural action of~$\R^2$ on the set of edge weights (the \emph{magnetic field coordinates} of~\cite[Section~2.3.3]{KOS}). In the case of a bipartite graph~$\G\subset\K$, there is an analogous natural action of~$\R$ on edge weights, defined as follows: for~$B\in\R$ and~$\nu=(\nu_e)_{e\in E}\in[0,\infty)^E$, set
\[
(B\bullet\nu)_e=\exp((e\cdot b)B)\nu_e\,,
\]
where~$e\cdot b\in\{-1,0,1\}$ denotes the intersection number in~$\K$ of the edge~$e$, oriented from the white to the black vertex, with the oriented curve~$b$ (recall Figure~\ref{fig:Klein}). Writing~$S_B$ for the bipartite characteristic polynomial of the weighted graph~$(\G,B\bullet\nu)\subset\K$, one easily checks the equality
\[
S_B(z,w)=S(\exp(B)z,w)\in\C[z^{\pm 1},w]\,.
\]
Similarly, one obtains the equality
\[
Q_B(z,w)=Q(\exp(2B)z,w)\in\R[z^{\pm 1},w^{\pm 1}]\,.
\]

We recall a couple of concepts from~\cite{KOS}. The {\em Newton polygon\/} of~$Q(z,w)=\sum_{(i,j)\in\Z^2}a_{ij}z^iw^j$ is defined as the convex hull of the set~$\{(i,j)\in\Z^2\,|\,a_{ij}\neq 0\}\subset\R^2$. We shall say that~$\G\subset\K$ (and~$\widetilde{\G}\subset\T^2$) are {\em non-degenerate\/} if the Newton polygon
of the corresponding characteristic polynomial~$Q$ has positive area. For a non-degenerate bipartite toric graph, the associated {\em spectral curve\/}
\[
\mathcal{C}=\{(z,w)\in (\C^*)^2\,|\,Q(z,w)=0\}
\]
is extremely well understood thanks to the work of Kenyon, Okounkov and Sheffield~\cite{KOS,K-O}. In a nutshell, it belongs to a special class of curves known as {\em Harnack curves\/}~\cite{Mik}, for which the
map~$(\C^*)^2\to\R^2$ defined by~$(z,w)\mapsto(\log|z|,\log|w|)$ is at most two-to-one~\cite{M-R}.
The image of~$\mathcal{C}$ via this map is called the {\em amoeba\/} of~$Q$~\cite{GKZ}, and is denoted by~$\mathbb{A}(Q)$.

We now use these tools to study the zeros of the characteristic polynomials~$Q$ and~$S$
associated with a non-degenerate bipartite graph~$\G\subset\K$.

\begin{proposition}
\label{prop:Q(-1)}
If~$\G\subset\K$ is a non-degenerate bipartite graph and~$(z,w)\in S^1\times S^1$ belongs to the spectral curve~$Q(z,w)=0$, then we have~$z=-1$.
\end{proposition}
\begin{proof}
First note that, using the symmetry~$Q(z,w)=Q(z,w^{-1})$ of Proposition~\ref{prop:Q} together with the fact that~$Q$ has real coefficients, the elements~$(z,w)\in S^1\times S^1$ of the spectral curve come in groups of four:~$0=Q(z,w)=Q(z,\overline{w})=Q(\overline{z},w)=Q(\overline{z},\overline{w})$.
Since a Harnack curve intersects the unit torus in at most two points, we must have~$z=\pm 1$ or~$w=\pm 1$.

Next, observe that the symmetry~$Q(z,w)=Q(z,w^{-1})$ immediately implies that~$\partial_wQ(z,w)$ vanishes for all~$z\in\C^*$ and~$w=\pm 1$. Also, if~$Q(1,w)=0$ for some fixed~$w=\pm 1$, then Equations~\eqref{equ:S} and~\eqref{equ:QS} imply that both~$S(1,w)$ and~$S(-1,w)$ vanish. Using Equation~\eqref{equ:QS} again, it follows
that~$\partial_zQ(1,w)=0$. Therefore, if~$Q(1,w)=0$ for some fixed~$w=\pm 1$, then we have~$\partial_zQ(1,w)=\partial_wQ(1,w)=0$; in other words, this is a singularity of the spectral curve.

Having established these two facts, let us analyse the intersection of the unit torus with the spectral curve~$Q_B(z,w)=0$, as~$B$ varies in~$\R$; the aim is to check that any element~$(z,w)$ in this intersection satisfies~$z=-1$ (the case~$B=0$ giving the proposition). This amounts to analysing the intersection of the amoeba of~$Q$ along the horizontal axis. If~$B$ lies outside the amoeba, then the intersection is empty and the statement holds trivially. If~$B$ lies on the boundary of the amoeba, then the intersection consists in a single real point that is not a singularity (see e.g. the last sentence of~\cite[Theorem~1]{M-R}). By the second fact above, it is of the form~$(-1,\pm 1)$. Finally, as~$B$ travels inside the amoeba from one boundary point~$B_0$ to another boundary point~$B_1>B_0$, we have~$Q_B(z_B,w_B)=Q_B(\overline{z_B},\overline{w_B})=0$ for some~$(z_B,w_B)\in S^1\times S^1$ varying continuously,
and satisfying the following conditions:
\begin{itemize}
\item{$z_{B_0}=z_{B_1}=-1$, and~$w_{B_0}=\pm 1$,~$w_{B_1}=\pm 1$};
\item{$(z_B,w_B)\neq(\overline{z_B},\overline{w_B})$ for~$B_0<B<B_1$ (except possibly at isolated real nodes);}
\item{$z_B=\pm 1$ or~$w_B=\pm 1$ for all~$B\in[B_0,B_1]$ (by the first fact above).}
\end{itemize}
By continuity of~$B\mapsto(z_B,w_B)$, we either have~$z_B$ constant (equal to~$z_{B_0}=-1$, and we are done), or~$w_B$ constant (equal to some~$w_0=\pm 1$) for all~$B\in[B_0,B_1]$. In the later case, we have~$0=Q_B(z_B,w_0)=Q(\exp(2B)z,w_0)$, so the polynomial map~$z\mapsto Q(z,w_0)$ vanishes on the arc~$\{\exp(2B)z_B\,|\,B\in[B_0,B_1]\}$. This implies that this polynomial is zero, which is impossible for a Harnack curve. This concludes the proof.
\end{proof}

We use a detailed study of the amoeba of the spectral curve to show the following result.

\begin{proposition}
\label{prop:rootS}
All the roots of~$S(z,1)$ and~$S(z,-1)$ are purely imaginary, and simple.
\end{proposition}
\begin{proof}
Let us fix~$w_0=\pm 1$ and~$z\in\C^*$ such that~$S(z,w_0)=0$. For~$B\in\R$ such that~$\exp(B)=|z|$, we have~$S_B(\exp(-B)z,w_0)=S(z,w_0)=0$ with~$(\exp(-B)z,w_0)$ in the unit torus. By Equation~\eqref{equ:QS}, we have that~$Q_B(\exp(-2B)z^2,w_0)$ vanishes as well, so by Proposition~\ref{prop:Q(-1)}, we must have~$\exp(-2B)z^2=-1$. This implies that~$z$ is purely imaginary.

It remains to show that these roots are simple.
Since~$Q_B(z,w)$ has at most (real) nodes on the unit torus, Equation~\eqref{equ:QS} and the argument above imply that the roots of~$S(z,w_0)$ have order at most~$2$,
with possible double roots corresponding to nodes of~$Q_B$. More precisely, a node~$(-1,w_0)$ of~$Q_B$ 
either corresponds to two conjugate simple roots~$\exp(B)i$ and~$-\exp(B)i$ of~$S(z,w_0)$,
or to a single double root of~$S(z,w_0)$ at~$z_0=\pm\exp(B)i$.
Unfortunately, such double roots cannot be excluded using Equation~\eqref{equ:QS} alone,
so we will use a careful analysis of the amoeba~$\mathbb{A}(Q)$ of~$Q$ to rule them out.
When perturbing the edge weights, such a node of~$Q_B$ would yield an oval in the 
boundary of~$\mathbb{A}(Q)$, meeting the horizontal axis in two points close to~$B$
corresponding to two simple roots of~$S(z,w_0)$ close to~$z_0$. We now show that this cannot happen,
as each oval of~$\partial\mathbb{A}(Q)$ meeting the horizontal axis
in two points yields two (simple) roots of~$S(z,w_0)$ that are located on opposite sides of the imaginary axis.

To show this claim, let us consider~$Q$ without any node (they can be deformed into ovals),
and write
\[
B'_0<B_1<B_1'<B_2<B_2'<\dots<B_{n-1}<B_{n-1}'<B_n
\]
for the coordinates of the intersection points of the horizontal axis with~$\partial\mathbb{A}(Q)$,
as illustrated in Figure~\ref{fig:amoeba} (left).
Each~$B_\ell$ (resp.~$B_\ell'$) corresponds to a simple root~$z_\ell$ (resp.~$z_\ell'$)
of~$S_{B_\ell}(z,w_\ell)$ (resp.~$S_{B_\ell}(z,w_\ell')$) for some~$w_\ell,w_\ell'=\pm 1$.
Let us first study these~$w_\ell,w_\ell'$ before turning to the claim above.
Due to the particular configuration of ovals in Harnack curves, each pair of points~$(B_\ell,B'_\ell)$
is linked by an oval of~$\partial\mathbb{A}(Q)$ for~$1\le\ell\le n$. Hence, by continuity of the (real) zeros of~$Q_B$
corresponding to these ovals, we have~$w_\ell=w_\ell'$ for all~$1\le\ell\le n$.
Moreover, when moving from~$B'_\ell$ to~$B_{\ell+1}$ in the interior of~$\mathbb{A}(Q)$,
the corresponding
roots of~$Q_B$ in the unit torus are of the form~$(-1,w_B)\neq(-1,\overline{w_B})$ with~$w_B$ moving
along the unit circle from~$w_\ell$ to~$w_{\ell+1}$.
The amoeba map of a Harnack curve being at most two-to-one, we necessarily 
have~$w'_\ell\neq w_{\ell+1}$ for all~$0\le\ell\le n$. Assuming without loss of generality that~$w'_0=1$,
we now have~$w_\ell=w'_\ell=(-1)^\ell$ for all~$0<\ell< n$ and~$w_n=(-1)^n$, thus completing the determination of these variables.
(This is illustrated by black and white dots in Figure~\ref{fig:amoeba}.)

\begin{figure}[tb]
\labellist\small\hair 2.5pt
\pinlabel {$B$} at 470 232
\pinlabel {\scriptsize $B'_0$} at 125 205
\pinlabel {\scriptsize $B_1$} at 160 205
\pinlabel {\scriptsize $B'_1$} at 220 205
\pinlabel {\scriptsize $B_2$} at 255 205
\pinlabel {\scriptsize $B'_2$} at 317 205
\pinlabel {\scriptsize $B_3$} at 350 205
\pinlabel {$z'_0$} at 675 250
\pinlabel {$z_1$} at 675 278
\pinlabel {$z'_1$} at 675 150
\pinlabel {$z_2$} at 675 115
\pinlabel {$z'_2$} at 675 375
\pinlabel {$z_3$} at 675 400
\endlabellist
\centering
\includegraphics[width=0.7\textwidth]{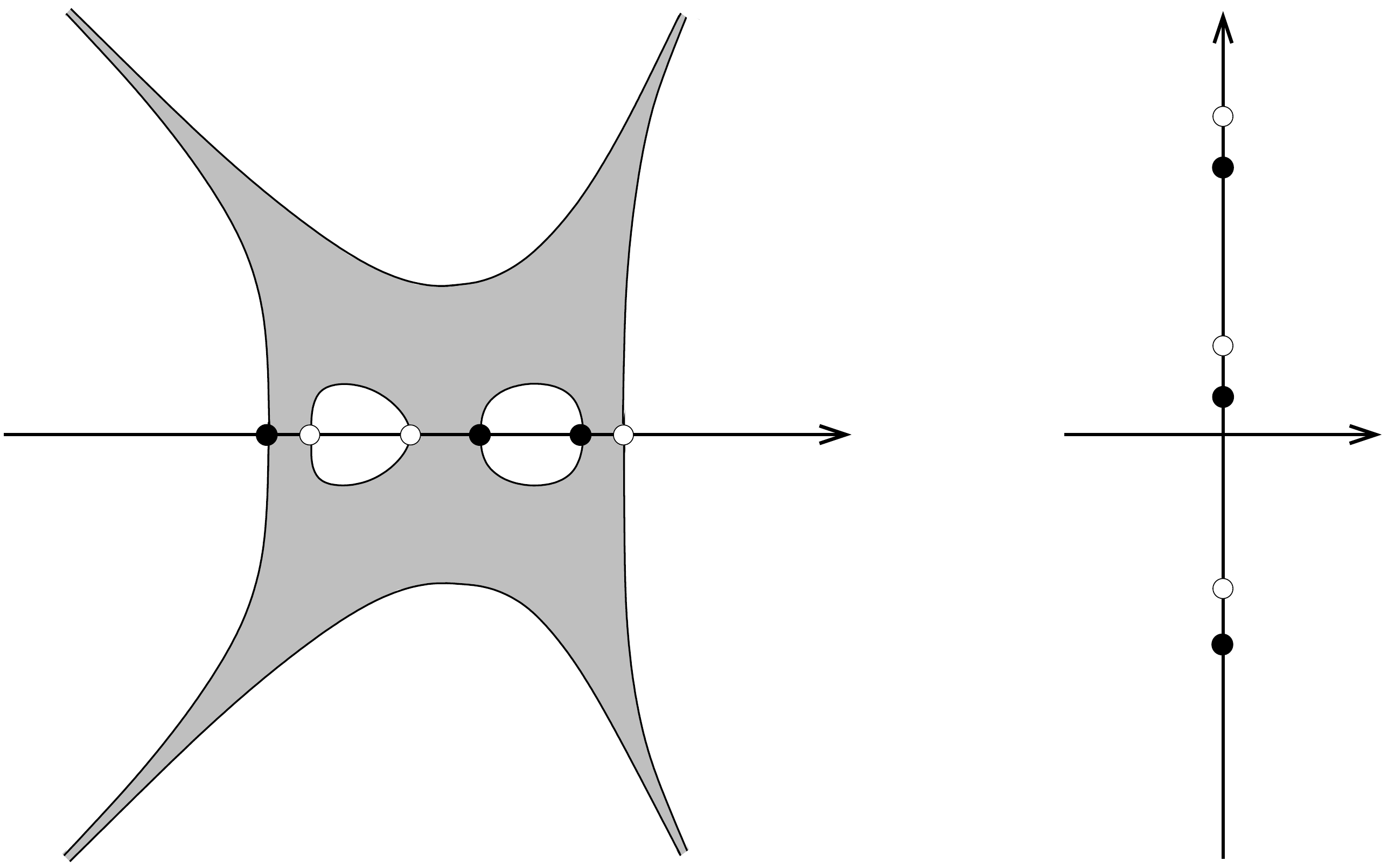}
\caption{Left: the intersection of a (schematized) amoeba with the horizontal axis, here with~$n=3$.
Right: the corresponding interlaced roots of~$S(z,1)$ and~$S(z,-1)$, respectively
represented by black and white dots.}
\label{fig:amoeba}
\end{figure}
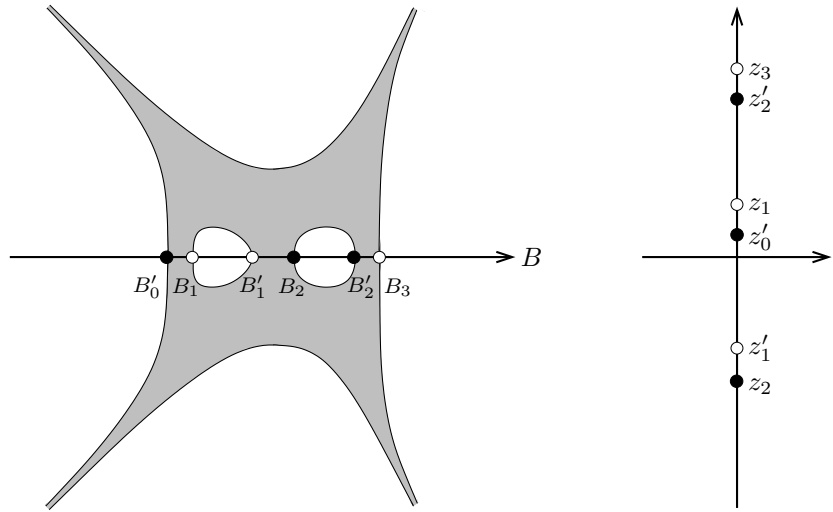

Let us turn to the roots~$z_\ell,z'_\ell$, and to our claim:
writing~$s_\ell$ and~$s_\ell'$ for the sign of~$iz_\ell$ and~$iz_\ell'$, respectively,
we wish to show that~$s_\ell$ and~$s_\ell'$ never coincide for~$1\le\ell\le n-1$.
To do so, let us consider the following~$1$-parameter deformation of the model. For~$t\in\R$,
let~$\nu^{(t)}$ denote the edge weights on~$\G\subset\K$ obtained via multiplication of~$\nu_e$ by~$t$
each time the edge~$e$ meets the curve~$a\subset\K$ (recall Figure~\ref{fig:Klein}).
Writing~$Q^{(t)}$ for the corresponding characteristic polynomial, we clearly have~$Q^{(1)}(z,w)=Q(z,w)=Q^{(-1)}(z,-w)$,
while~$Q^{(0)}$ only depends on~$z$. Hence, the amoebas of~$Q$ and~$Q^{(-1)}$ coincide, while the (degenerate)
one of~$Q^{(0)}$ consists in~$n$ vertical lines.
More precisely, when~$t$ decreases from~$1$ to~$0$, the topology of~$\partial\mathbb{A}(Q^{(t)})$ is unchanged, but the ovals become wider, the
points~$B'_\ell$ and~$B_{\ell+1}$ grow closer, eventually meeting at one of these vertical lines for~$t=0$.
As~$t$ decreases further from~$0$ to~$-1$, a symmetrical deformation is observed, ending up with the same
amoeba~$\mathbb{A}(Q^{(-1)})=\mathbb{A}(Q)$, but with the roles of~$w$ and~$-w$ exchanged.
(Note that this is consistent with the determination of~$w_\ell,w_\ell'$ above.)
The key observation is that during this continuous deformation,
the roots~$z'_\ell$ of~$S(z,(-1)^\ell)$ and~$z_{\ell+1}$ 
of~$S(z,(-1)^{\ell+1})$ are exchanged without additional collisions between any of the roots~$z_0',z_1,z_1',\dots,z_{n-1}',z_n$.
Since the norm of these roots are ordered as the corresponding~$B_\ell$ and~$B_\ell'$,
we are in one of the following two cases:
\begin{enumerate}[(i)]
\item{either~$s'_0=s_1=s_1'=\dots s_n'=s_0$, and the roots are ordered as their norms;}
\item{or~$s'_0=s_1\neq s_1'=  \dots = s_{n-1}\neq s_{n-1}'=s_n$, and the roots
of~$S(z,1)$ and~$S(z,-1)$ alternate along the imaginary axis (see Figure~\ref{fig:amoeba}, right).}
\end{enumerate}
The second case yields the claim, so we are left with ruling out the first one for~$n\ge 2$.
By means of contradiction, let us consider a non-degenerate bipartite graph~$\G\subset\K$ realising case~(i)
with~$n\ge 2$. Then, it necessarily contains two paths winding in the horizontal direction (since~$n\ge 2$), and one winding in the
vertical direction (since it is non-degenerate). Sending the weights of the edges not contained
in these paths to~$0$ and shrinking all degree~$2$ vertices of these paths (see e.g.~\cite{G-K}) leads to
the bipartite square lattice of Example~\ref{ex:R}, which is easily seen to display a root configuration as in case~(ii) above, with~$n=2$.
Therefore, this transformation produces a continuous deformation of the roots of~$S(z,1)$ and~$S(z,-1)$ from case~(i) with~$n\ge 2$ to case~(ii)
with~$n=2$, which is impossible to realise while staying in one of the two allowed families of configurations.
This concludes the proof.
\end{proof}

Note that the proof above yields the additional remarkable fact that the real Laurent
polynomials~$S(iz,1)$ and~$S(iz,-1)$ are {\em interlaced\/}, i.e. have only real roots which alternate
along the real line (see e.g.~\cite{interlace}).

The final result of this section requires the following notations.
Since~$S(-z,1)=\overline{S(z,1)}$, the leading coefficient of~$S(z,1)$ has argument~$\lambda\frac{\pi}{2}$ for some~$\lambda\in\Z$. Let us denote by~$r$ the total number of roots of~$S(z,1)$ with modulus~$>1$, counted with multiplicities, and set~$A:=\lambda+r$. We shall write~$A':=\lambda'+r'$ for the corresponding integers
associated with the polynomial~$S(z,-1)$.

\begin{lemma}
\label{lemma:S}
\begin{enumerate}
\item{If~$Q(z,w)$ has no zeros in the unit torus, then~$A+A'$ is even.}
\item{If~$Q(z,w)$ has two distinct zeros in the unit torus, then~$A$ is odd and~$A'$ is even.}
\item{If~$Q(z,w)$ has a node in the unit torus, then~$A$ is odd and~$A'$ is even.}
\end{enumerate}
\end{lemma}

\begin{proof}
The strategy of the proof is once again to fix~$\G$ and study the intersection of the unit torus with the spectral
curve~$Q_B(z,w)=0$ as~$B$ varies in~$\R$, i.e. to analyse the amoeba~$\mathbb{A}(Q)$ of~$Q$ along the horizontal
axis. More precisely, we shall start by checking that the statement holds for~$B$ big enough, and then show that it
remains true as~$B$ decreases. The crucial idea is the following one: the three cases in the statement of the
lemma correspond to three possible locations of~$B\in\R$, and passing from one to another corresponds to crossing
the boundary or a node of~$\mathbb{A}(Q)$, i.e. a real root~$(z_0,w_0)$ of~$Q_B$, with~$z_0=-1$ by
Proposition~\ref{prop:Q(-1)}. By Equation~\eqref{equ:QS} and Proposition~\ref{prop:rootS},
this corresponds to one or two simple roots of~$S_B(z,w_0)$ exiting the unit disc 
along the imaginary axis, which in turn corresponds
to a change in the integer~$A$ if~$w_0=1$, and in the integer~$A'$ if~$w_0=-1$.
Observe also that varying continuously the weights~$\nu\in(0,\infty)^E$
amounts to continuously deforming the amoeba~$\mathbb{A}(Q)$ without changing its topology, as the amoeba of Harnack
curve is severely constrained by the corresponding Newton polygon (see~\cite{Mik}, and more detail below).
In particular, this will not change the order in which we meet the roots of~$S_B(z,1)$ and~$S_B(z,-1)$
as~$B$ decreases; in other words, this will not change the order of the moduli of the roots of~$S(z,1)$
and~$S(z,-1)$.

Before implementing this idea in detail, let us start by studying the parity of the integers~$\lambda,\lambda'$.
Since~$S(-z,1)=\overline{S(z,1)}$, the leading coefficient of~$S(z,1)$ has argument~$\lambda\frac{\pi}{2}$
with~$\lambda$ of the same parity as the top-degree~$d$ of~$S(z,1)$. Similarly, the integer~$\lambda'$
has the same parity as the top-degree~$d'$ of~$S(z,-1)$.
By Equation~\eqref{equ:QS}, the square of the leading coefficient of~$S(z,1)$ is equal to the leading
coefficient of~$Q(z,1)$, which is the degree~$d$ coefficient, and similarly for~$S(z,-1)$.
To describe these coefficients explicitely, let us fix 
a reference matching~$M_0$ on~$\widetilde{\G}$, assuming without loss of generality that it is disjoint
from~$\tilde{a}$ and~$\tilde{b}$. By~\cite[Proposition~3.1]{KOS}, we have
\[
Q(z,w)=\sum_{M\in\mathcal{M}(\widetilde{G})}(-1)^{h_xh_y+h_y}\nu(M)z^{h_x}w^{h_y}\,,
\]
where~$(h_x,h_y)\in\Z^2$ denote the coordinates of~$[M-M_0]\in H_1(\mathbb{T}^2;\Z)=\Z\tilde{a}\oplus\Z\tilde{b}$.
(This formula is valid for a specific Kasteleyn orientation on~$\widetilde{\G}\subset\mathbb{T}^2$, but
one can check that our conventions for~$\widetilde{K}$ are coherent with this choice.)
Let us denote by~$h$ the maximal value of~$h_x$ over all~$M\in\mathcal{M}(\widetilde{G})$ and
write~$Z_{h0}$ (resp.~$Z_{h1}$) for the contribution to the partition function of~$\widetilde{\G}$ of the
matchings with~$h_x=h$ and~$h_y$ even (resp. odd).
By the equality displayed above, we have~$d\le h$, and the degree~$h$ coefficient of~$Q(z,1)$ is equal
to~$Z_{h0}-Z_{h1}$ if~$h$ is even and to~$Z_{h0}+Z_{h1}\neq 0$ if~$h$ is odd.
Similarly, the top degree~$d'$ of~$Q(z,-1)$ is at most~$h$, and its degree~$h$ coefficient is equal
to~$Z_{h0}+Z_{h1}\neq 0$ if~$h$ is even and to~$Z_{h0}-Z_{h1}$ if~$h$ is odd.
Let us first assume that~$h$ is even. In such a case, we have that~$d'=h$ is even, while~$d$ is equal to~$h$
(and therefore even as well) unless the weights~$\nu$ satisfy the equality~$Z_{h0}=Z_{h1}$, i.e.
\begin{equation}
\label{equ:deg}
\sum_{M\in\{\mathcal{M}(\widetilde{\G})\,|\,h_x=h,\;\text{$h_y$ even}\}}\!\!\nu(M)=\sum_{M\in\{\mathcal{M}(\widetilde{\G})\,|\,h_x=h,\;\text{$h_y$ odd}\}}\!\!\nu(M)\,.
\end{equation}
Therefore, we have that~$\lambda\equiv d$ and~$\lambda'\equiv d'$ are both even if the
weights~$\nu$ are \emph{generic}, in the sense
that they do not satisfy equality~\eqref{equ:deg} above. The case of~$h$ odd is similar, leading
to~$\lambda\equiv d\equiv h$ always odd, and~$\lambda'\equiv d'\equiv h$ odd as well for generic weights.

We now investigate the geometric meaning of Equation~\eqref{equ:deg}. Consider a path in the space of
generic weights ending in non-generic weights. By the discussion above, this corresponds to the degree~$h$
coefficient of~$Q(z,(-1)^h)$ (or equivalently, of~$S(z,(-1)^h)$) tending to zero, with all other coefficients of~$S(z,1)$ and~$S(z,-1)$ bounded away from zero.
This results in the modulus of the biggest root of~$S(z,(-1)^h)$ tending to infinity and all other roots of~$S(z,1)$ and~$S(z,-1)$ having bounded modulus.
By the observation at the beginning of the proof,
this implies that the biggest root of~$S(z,1)$ and~$S(z,-1)$ belongs to~$S(z,(-1)^h)$. In other words,
as~$B$ decreases within generic weights, the first time we hit the boundary of~$\mathbb{A}(Q)$ corresponds to
a root of~$S(z,(-1)^h)$.
Now, recall that for Harnack curves, the Newton polygon~$\Delta(Q)$ of~$Q$ allows to describe the associated
amoeba~$\mathbb{A}(Q)$ as follows: each interval between two adjacent points in~$\Z^2\cap\partial\Delta(Q)$
produces one tentacle of~$\mathbb{A}(Q)$ with asymptotic direction orthogonal to this interval.
The horizontal axis is generically not contained in one of these tentacles; the only way for this to happen is if
two adjacent horizontal tentacles from either sides of this axis merge to give a single tentacle, thus
sending the right-most intersection of~$\partial\mathbb{A}(Q)$ with the horizontal axis to infinity.
By the discussion above, this corresponds to the weights~$\nu$ varying so that the modulus of the biggest root
of~$S(z,(-1)^h)$ tends to infinity, i.e. to the non-generic case defined by Equation~\eqref{equ:deg}.
In summary, the non-genericity condition defined by Equation~\eqref{equ:deg} corresponds
precisely to some ray~$[B_0,\infty)$ of the horizontal axis being contained in~$\mathbb{A}(Q)$.  

We are finally ready to start the actual proof of the statement.
Let us first consider the case of~$B$ big enough on the horizontal axis, with generic weights.
By the discussion above, we are outside~$\mathbb{A}(Q)$, and therefore in case~(1).
Since~$B$ is big, we also have~$r=r'=0$, as all the roots of~$S_B(z,1)$ and~$S_B(z,-1)$ have
modulus~$<1$. Therefore, in the generic case for~$B$ big enough, we are in case~(1) and have~$A\equiv A'\equiv h$,
so the statement holds.
Let us now turn to the non-generic case for~$B$ big enough. This time, we are inside a horizontal tentacle
of~$\mathbb{A}(Q)$, and therefore in case~(2). As discussed above, such a case can be obtained as a limit of
generic weights, with the leading coefficient of~$S_B(z,(-1)^h)$ tending to zero and its biggest root
tending to infinity: this corresponds to changing the parity of the corresponding integer~$A$ or~$A'$, which
results in~$A$ odd and~$A'$ even, as claimed. In any case, we see that the lemma holds for~$B$ big enough.

We now study the behavior of~$A$ and~$A'$ as we decrease~$B$ along the horizontal axis. 
As explained earlier, the only way for~$A$ or~$A'$ to change is if (simple) roots of~$S_B(z,1)$ or~$S_B(z,-1)$
cross the values~$\pm i$ along the imaginary axis. This corresponds to~$Q_B(-1,1)$ or~$Q_B(-1,-1)$ vanishing,
i.e. to~$B$ crossing the boundary or a real node of~$\mathbb{A}(Q)$. Therefore, we are left with the proof
that the statement of the lemma is coherent with such phase transitions.
Starting with~$B$ big enough (in the generic case), we have~$A\equiv A'\equiv h$. Let us decrease~$B$ until
we first cross the boundary of~$\mathbb{A}(Q)$, thus transitioning from case~(1) to case~(2).
As discussed above, this corresponds to the biggest root of~$S(z,(-1)^h)$ exiting the unit disc,
leading to~$A$ odd and~$A'$ even. Continuing to decrease~$B$, we might cross once again~$\partial\mathbb{A}(Q)$,
thus exiting the amoeba, but perhaps through an oval this time (thus entering a \emph{gazeous phase}).
This results in a simple root of~$S(z,w_0)$ exiting the unit disc, for some~$w_0=\pm 1$,
and therefore a change in the parity of~$A+A'$.
(We know from the proof of Proposition~\ref{prop:rootS} that~$w_0=(-1)^{h+1}$, but this is not needed here.)
It corresponds to transitioning from case~(2) back to case~(1),
and we indeed have~$A+A'$ even once again. Note however that, due to the particular topology of oval arrangements
in Harnack curves, the next time we hit~$\partial\mathbb{A}(Q)$ must be through the same oval; therefore,
this corresponds to another simple root of~$S(z,w_0)$ exiting the unit disc, for the same~$w_0=\pm 1$ as before.
We thus return to~$A$ odd and~$A'$ even, which is once again consistent with the claimed statement.
The last possible phase transition is when we cross a real node~$(-1,w_0)$ inside~$\mathbb{A}(Q)$,
which by Equation~\eqref{equ:QS} and Proposition~\ref{prop:rootS} corresponds
to two simple roots of~$S(z,w_0)$ exiting the unit disc, and to a transition from case~(2) to case~(3).
The parity of the integers~$A$ and~$A'$ is obviously unchanged, concluding the proof.
\end{proof}


\section{Enlarging the fundamental domain}
\label{sec:enlarge}

The aim of this section is to show how the dimer partition function of a periodic weighted graph~$\G_{mn}\subset\K$ of arbitrary size can be computed from the characteristic polynomials of the original weighted graph~$\G_{11}=\G\subset\K$, see Theorem~\ref{thm:Zmn}. Via the Pfaffian formula (Proposition~\ref{prop:Pf}), this can be achieved if we understand how the Kleinian characteristic polynomial~$R_{mn}$ of~$\G_{mn}$ can be expressed in terms of~$R$ and~$P$. The answer is given in Theorem~\ref{thm:Rmn}, which is the main technical achievement of this section.

It is organised as follows. In Section~\ref{sub:CK}, we state a recent result of Kassel and the author~\cite{C-K}, probably folklore, which plays a crucial role in this discussion. In Section~\ref{sub:cover}, we state Theorem~\ref{thm:Rmn} and show how it implies Theorem~\ref{thm:Zmn}. The proof of Theorem~\ref{thm:Rmn} is contained in Sections~\ref{sub:period} to~\ref{sub:proof}.

\subsection{Covering spaces and twisted operators}
\label{sub:CK}

One of the main technical tools used in this article is a result due to Adrien Kassel and the author~\cite{C-K}, but probably known to the experts. The aim of the present section is to succinctly explain a special case of this result adapted to our context.

As in Section~\ref{sub:gen}, let us fix a connected graph~$\G$ together with edge weights~$\nu$ and an orientation~$K$. The embedding of~$\G$ in the Klein bottle~$\K$ endowed with the curves~$a,a'$ provides an additional structure: the map~$\omega\colon E\to\Z/2\Z$ given by~$\omega(e)=e\cdot a+e\cdot a'$. (Technically, this is a~$1$-cocycle representing the first Stiefel-Whitney class of~$\K$.) Finally, let us fix a base vertex~$v_0\in V$ and a finite-dimensional complex linear representation~$\rho\colon\pi_1(\G,v_0)\to\operatorname{GL}(W)$.

It is not difficult to show that any such homomorphism~$\rho$ can be represented by a {\em connection\/}, i.e. a family~$\Phi=(\varphi_e)_{e\in\mathbb{E}}\in\operatorname{GL}(W)^{\mathbb{E}}$ indexed by the set~$\mathbb{E}$ of oriented edges of~$\G$, such that~$\varphi_{\overline{e}}=\varphi_e^{-1}$ if~$e,\overline{e}$ denote the same edge with opposite orientations. This means that for each loop~$\gamma$ in~$\G$ based at~$v_0$, the composition of the corresponding automorphisms~$\varphi_e$ is equal to~$\rho(\gamma)$.
Using this data and in the spirit of~\cite{Ken}, one can define an associated {\em twisted Kasteleyn operator\/}~$A^\rho=A^\rho(\G,\nu,K,\omega)$ acting on the set~$W^V$ of~$W$-valued functions on~$V$ as follows: for~$f\in W^V$ and~$v\in V$, set
\[
(A^\rho f)(v)=\sum_{e=(v,v')}\varepsilon_{vv'}^K(e)\,i^{\omega(e)}\,\nu_e\,\varphi_e(f(v'))\,,
\]
where the notations are as in Section~\ref{sub:Klein}.

\begin{remarks}
\label{rem:twist}
\begin{enumerate}[(i)]
\item A fixed homomorphism~$\rho$ can be representated by various connections. However, one can show that any two such connections are {\em gauge equivalent\/}. This implies that the corresponding twisted Kasteleyn operators are conjugated by an automorphism of~$W^V$, and justifies the abuse of notation.
\item If~$\rho_1,\rho_2$ are two representations, then the operators~$A^{\rho_1\oplus\rho_2}$ and~$A^{\rho_1}\oplus A^{\rho_2}$ are clearly conjugated by an automorphism of~$W^V$.
\end{enumerate}
\end{remarks}

As a natural class of examples, consider the homomorphisms~$\rho\colon\pi_1(\G,v_0)\to\operatorname{GL}(W)$ given by the irreducible representations that factor through the inclusion induced homomorphism~$\pi_1(\G,v_0)\to\pi_1(\K,v_0)$ and the abelianization~$\pi_1(\K,v_0)\to H_1(\K;\Z)\simeq\Z[a]\oplus\Z/2\Z[b]$. Being abelian and irreducible, such a representation is~$1$-dimensional and fully determined by the image~$z\in\C^*$ of~$[a]$ and~$w\in\{\pm 1\}$ of~$[b]$. The resulting twisted Kasteleyn operator is nothing but~$A(z,w)$, as defined in Section~\ref{sub:Klein}.

The main technical novelty of our approach is that, in order to understand the dimer model on (bigger and bigger) Klein bottles, one needs to consider not only these operators, but the ones twisted by~$2$-dimensional representations as well.

To see this, let us consider a covering map~$p\colon\widehat{\G}\to\G$ with~$\widehat{\G}$ a finite connected graph. The additional data~$\nu,K,\omega$ on~$\G$ lifts uniquely to~$\widehat{\nu},\widehat{K},\widehat{\omega}$ on~$\widehat{\G}$; therefore, any representation~$\rho\colon\pi_1(\widehat{\G},\hat{v}_0)\to\operatorname{GL}(W)$ allows us to define~$\widehat{A}^\rho:=A^\rho(\widehat{\G},\widehat{\nu},\widehat{K},\widehat{\omega})$ as above.
Note that~$p$ being a covering map, it induces an injection~$p_*\colon\pi_1(\widehat{\G},\hat{v}_0)\hookrightarrow\pi_1(\G,v_0)$ on fundamental groups (see e.g.~\cite[Chapter~1]{Hat}). Hence, one can identify~$\pi_1(\widehat{\G},\hat{v}_0)$ with a subgroup of~$\pi_1(\G,v_0)$.  Finally, recall that given any linear representation~$\rho\colon H\to\operatorname{GL}(W)$ of a subgroup~$H$ of a group~$G$ (for example, of~$\pi_1(\widehat{\G},\hat{v}_0)<\pi_1(\G,v_0)$), there is an {\em induced representation\/}~$\rho^\#\colon G\to\operatorname{GL}(Z)$, well-defined up to isomorphism (see~\cite[Section~3.3]{Serre} and Section~\ref{sub:repr} below).

The following statement is a special case of the main theorem of~\cite{C-K}.

\begin{theorem}
\label{thm:C-K}
There is an isomorphism~$W^{\widehat{V}}\to Z^V$ that conjugates~$\widehat{A}^\rho$ and~$A^{\rho^\#}$.
\end{theorem}

To illustrate this result, consider the simpler case of a toric graph~$\G\subset\T^2$ and the associated (real-valued) Kasteleyn matrix. Let~$\widehat{\G}=\G_{mn}$ denote the lift of~$\G$ by the~$m\times n$ cover~$\T_{mn}^2\to\T^2$ of the torus by itself. This covering being normal, the trivial representation of~$\pi_1(\G_{mn})$ is easily seen to induce the representation of~$\pi_1(\G)$ given by the composition
\[
\pi_1(\G)\stackrel{i_*}{\longrightarrow}\pi_1(\T^2)\stackrel{\mathit{pr}}{\longrightarrow}\pi_1(\T^2)/\pi_1(\T_{mn}^2)\simeq\mathit{Gal}(\widehat{\G}/\G)\stackrel{\rho_{\text{reg}}}{\longrightarrow}\operatorname{GL}(Z)\,,
\]
where~$i_*$ denotes the inclusion induced homomorphism,~$\mathit{pr}$ the canonical projection, and~$\rho_{\text{reg}}$ the {\em regular representation\/} of the Galois group~$\mathit{Gal}(\widehat{\G}/\G)\simeq\Z/m\Z\times\Z/n\Z$ of this covering. For such a finite group, this regular representation is known to split as the direct sum of all irreducible representations of~$\mathit{Gal}(\widehat{\G}/\G)$ (see~\cite[Section~2.4]{Serre}). In our case, this group being abelian, all the irreducible representations are~$1$-dimensional so~$\rho_{\text{reg}}$ splits as
\[
\rho_{\text{reg}}=\bigoplus_{z^n=1}\bigoplus_{w^m=1}\rho(z,w)\,,
\]
with~$\rho(z,w)$ mapping a fixed generator of~$\Z/m\Z$ (resp.~$\Z/n\Z$) to~$w\in\C^*$ (resp.~$z\in\C^*$). Using the version of Theorem~\ref{thm:C-K} adapted to this context together with Remark~\ref{rem:twist}~(ii) above, we obtain the following fact: the (untwisted) Kasteleyn operator associated with~$\G_{mn}\subset\T^2$ is conjugate to the direct sum of the Kasteleyn operators associated with~$\G\subset\T^2$ twisted by~$\rho(z,w)$, the product being over all~$z,w\in\C^*$ such that~$z^n=1$ and~$w^m=1$. Taking the determinant, and writing~$P_{mn}$ for the characteristic polynomial of~$\G_{mn}\subset\T^2$, we get
\[
P_{mn}(1,1)=\prod_{z^n=1}\prod_{w^m=1}P(z,w)\,,
\]
which is nothing but a special case of Equation~\eqref{equ:Pmn}.

The section of~\cite{KOS} containing this latter statement is entitled ``enlarging the fundamental domain''. We borrowed this title for the present section, whose aim is to perform the same action, no longer on tori, but on Klein bottles.

\subsection{Covering the Klein bottle by itself}
\label{sub:cover}

As in the introduction, let us fix a weighted graph~$\G\subset\K$, two positive integers~$m$ and~$n$ with~$n$ odd, and denote by~$\G_{mn}\subset\K_{mn}$ the lift of~$\G\subset\K$ by the~$m\times n$ cover~$\K_{mn}\to\K$ of the Klein bottle by itself (recall Figure~\ref{fig:Gmn}). The edge weights on~$\G$ lift to edge weights on~$\G_{mn}$, so one can consider the associated Kleinian and toric characteristic polynomials~$R_{mn}(z,w)$ and~$P_{mn}(z,w)$, as explained in Section~\ref{sec:char}.

The main result of this section is the expression of~$R_{mn}$ in terms of~$R_{11}=R$ and~$P_{11}=P$,
as follows.

\begin{theorem}
\label{thm:Rmn}
For any positive integers~$m,n$ with~$n$ odd, we have
\[
R_{mn}(1,1)=
\begin{cases}
\prod_{z^n=1}\left(R(z,1)\prod_{1\le k\le m-1,\,k\text{ even}}P(z,\zeta^{k})\right)&\mbox{if~$m$ is odd;} \\
\prod_{z^n=1}\left(R(z,1)\overline{R(z,-1)}\prod_{1\le k\le m-1,\,k\text{ even}}P(z,\zeta^{k})\right)& \mbox{if~$m$ is even,}
\end{cases}
\]
and
\[
R_{mn}(1,-1)=
\begin{cases}
\prod_{z^n=1}\left(\overline{R(z,-1)}\prod_{1\le k\le m-1,\,k\text{ odd}}P(z,\zeta^{k})\right)&\mbox{if~$m$ is odd;} \\
\prod_{z^n=1}\prod_{1\le k\le m-1,\,k\text{ odd}}P(z,\zeta^{k}) & \mbox{if~$m$ is even,}
\end{cases}
\]
where~$\zeta$ stands for~$\exp(\pi i/m)$.
\end{theorem}

We postpone the proof of this theorem to Sections~\ref{sub:period} to~\ref{sub:proof}.

Writing~$Z_{mn}$ for the dimer partition function of the weighted graph~$\G_{mn}$, it allows us to prove
Theorem~\ref{thmintro:Z}, that we now recall for the reader's convenience.
Note that since~$P_{mn}$ can be computed in terms of~$P$ via Equation~\eqref{equ:Pmn},
this theorem shows that~$Z_{mn}$ can be expressed using the polynomials~$P$ and~$R$ alone.

\begin{theorem}
\label{thm:Zmn}
For positive integers~$m,n$ with~$n$ odd, we have
\[
Z_{mn}=\left|\sin(\alpha_n/2)\right|P_{mn}(1,1)^{1/4}+\left|\cos(\alpha'_n/2)\right|P_{mn}(1,-1)^{1/4}
\]
if~$m$ is odd, and
\[
Z_{mn}=\left|\sin((\alpha_n-\alpha'_n)/2)\right|P_{mn}(1,1)^{1/4}+P_{mn}(1,-1)^{1/4}
\]
if~$m$ is even, where
\[
\alpha_n=\mathrm{Arg}\Big(\prod_{z^n=1}R(z,1)\Big)\,,\quad\alpha'_n=\mathrm{Arg}\Big(\prod_{z^n=1}R(z,-1)\Big)\,,
\]
and~$P_{mn}(1,\pm 1)^{1/4}$ denotes the non-negative fourth root of~$P_{mn}(1,\pm 1)\ge 0$.
\end{theorem}
\begin{proof}[Proof of Theorem~\ref{thm:Zmn}]
Applying Proposition~\ref{prop:Pf} to~$\G_{mn}$ gives
\[
Z_{mn}=\big|{\rm Im}\big(R_{mn}(1,1)^{1/2}\big)\big|+\big|{\rm Re}\big(R_{mn}(1,-1)^{1/2}\big)\big|\,.
\]
In the case of~$m$ odd, Theorem~\ref{thm:Rmn} and Proposition~\ref{prop:Ppos} yield
\[
Z_{mn}=\left|\sin(\alpha_n/2)\right|\left|R_{mn}(1,1)\right|^{1/2}+\left|\cos(\alpha'_n/2)\right|\left|R_{mn}(1,-1)\right|^{1/2}\,,
\]
with~$\alpha_n,\alpha_n'$ as in the statement above. The final formula follows from
Proposition~\ref{prop:P}~(iii) applied to~$\G_{mn}$. The case of~$m$ even is similar.
\end{proof}

In the case of a bipartite graph,  Theorem~\ref{thm:Zmn} can be reformulated as follows.

\begin{corollary}
\label{cor:Zmn}
For positive integers~$m,n$ with~$n$ odd, we have
\[
Z_{mn}=\left|\sin(\beta_n)\right|\left|Q_{mn}(1,1)\right|^{1/2}+\left|\cos(\beta'_n)\right|\left|Q_{mn}(1,-1)\right|^{1/2}
\]
if~$m$ is odd, and
\[
Z_{mn}=\left|\sin(\beta_n-\beta'_n)\right|\left|Q_{mn}(1,1)\right|^{1/2}+\left|Q_{mn}(1,-1)\right|^{1/2}
\]
if~$m$ is even, where~$\beta_n=\mathrm{Arg}\big(\prod_{z^n=1}S(z,1)\big)$ and~$\beta'_n=\mathrm{Arg}\big(\prod_{z^n=1}S(z,-1)\big)$.
\end{corollary}
\begin{proof}
Equation~\eqref{equ:factorPQ} applied to~$\G_{mn}$ and~$z,w\in S^1$ leads to the equality
\[
P_{mn}(z,w)=Q_{mn}(z,w)Q_{mn}(\overline{z},\overline{w})=|Q_{mn}(z,w)|^2\,.
\]
Furthermore, Equation~\eqref{equ:factorRS} applied to~$\G$ and~$w=\pm 1$ yields
\[
\prod_{z^n=1}R(z,w)=\prod_{z^n=1}S(z,w)\prod_{z^n=1}S(z^{-1},w)=\prod_{z^n=1}S(z,w)^2\,.
\]
Corollary~\ref{cor:Zmn} is an immediate consequence of Theorem~\ref{thm:Zmn} together with the equalities displayed above.
\end{proof}

\begin{example}
\label{ex:Z}
Consider the bipartite square lattice illustrated in Figure~\ref{fig:ex1}, with weights~$x_1=x_2=:x$ and~$y_1=y_2=:y$.
As computed in Examples~\ref{ex:R} and~\ref{ex:P}, we have~$S(z,w)=x(z+z^{-1})+iy(1+w)$ and~$Q(z,w)=x^2(z+z^{-1}+2)+y^2(w+w^{-1}+2)$. Since~$S(z,-1)$ is always real for~$z\in S^1$, Corollary~\ref{cor:Zmn} now takes the simpler form
\[
Z_{mn}=\left|\sin(\beta_n)\right|\left|Q_{mn}(1,1)\right|^{1/2}+\left|Q_{mn}(1,-1)\right|^{1/2}\,.
\]
As a reality check, let us compute~$Z_{21}$ using this formula together with Equation~\eqref{equ:Pmn}. It yields
\[
Z_{21}=\frac{y}{\sqrt{x^2+y^2}}|4(x^2+y^2)4x^2|^{1/2}+|(4x^2+2y^2)(4x^2+2y^2)|^{1/2}=4x^2+4xy+2y^2\,,
\]
which can easily be checked by hand. Note that on this example of a~$4\times 1$ square lattice, Equation~(5) of~\cite{Lu-Wu} yields the incorrect result~$4x^2+2y^2$, an error that propagates to~\cite{IOH} (see Example~\ref{ex:sqbip} below).
\end{example}

\medskip

The rest of this section is devoted to the proof of Theorem~\ref{thm:Rmn}. It is divided into three parts. In Section~\ref{sub:period}, we show how the orientation~$K$ and cocycle~$\omega$ on~$\G_{mn}$ can be made periodic, so that Theorem~\ref{thm:C-K} can be used. In Section~\ref{sub:repr}, we analyse the two induced representations of~$\pi_1(\K)$ and show that they split as direct sums of~$1$- and~$2$-dimensional irreducible representations. Finally, Section~\ref{sub:proof} builds upon the two previous ones to complete the proof of Theorem~\ref{thm:Rmn}.

\subsection{Making the orientation and the cocycle periodic}
\label{sub:period}

Let us start by reformulating the definition of~$R_{mn}(1,\pm 1)$ using the langage of twisted operators introduced in Section~\ref{sub:cover}, in order to apply Theorem~\ref{thm:C-K}.

Let~$\G\subset\K$ be a graph endowed with edge weights~$\nu$, an orientation~$K$ satisfying conditions~(i) and~(ii) from Section~\ref{sub:gen}, and the~$1$-cocycle~$\omega\colon E\to\Z/2\Z$ given by~$\omega(e)=e\cdot(a+a')$. Let~$\rho$ be the trivial representation of~$\pi_1(\G)$, and~$\rho'$ denote the~$1$-dimensional representation determined by the connection~$(\varphi_e)_e$ given by~$\varphi_e=(-1)^{e\cdot a}$. Then, the associated twisted Kasteleyn matrices~$A^\rho=A^\rho(\G,\nu,K,\omega)$ and~$A^{\rho'}$ satisfy~$R(1,1)=\det(A^\rho)$ and~$R(1,-1)=\det(A^{\rho'})$.

If~$\widehat{\G}:=\G_{mn}$ is the graph considered above, then~$\nu$,~$K$ and~$\omega$ lift to edge weights~$\widehat{\nu}=\nu_{mn}$, an orientation~$\widehat{K}$, and a~$1$-cocycle~$\widehat{\omega}$ on~$\G_{mn}\subset\K_{mn}$. Theorem~\ref{thm:C-K} can be applied to compute~$\widehat{A}^\rho=A^\rho(\widehat{\G},\widehat{\nu},\widehat{K},\widehat{\omega})$ and~$\widehat{A}^{\rho'}$. However, in order to ensure that these matrices can be used to compute~$R_{mn}(1,\pm 1)$, we must ensure that~$\widehat{K}$ and~$\widehat{\omega}$ can be transformed to~$K_{mn}$ and~$\omega_{mn}$ satisfying the necessary properties explained in Section~\ref{sub:gen}.

To see this, let us denote by~$a_{mn},a'_{mn}\subset\K_{mn}$ the two parallel cycles generating~$H_1(\K_{mn};\Z)$ as described in Figure~\ref{fig:Klein}, and define~$\omega_{mn}\colon E(\G_{mn})=:E_{mn}\to\Z/2\Z$ by~$\omega_{mn}(e)=e\cdot(a_{mn}+a'_{mn})$. The two cycles~$a,a'\subset\K$ lift to~$2m$ parallel cycles~$\widehat{a},\widehat{a}'\subset\K_{mn}$ such that~$\widehat{\omega}(e)=e\cdot(\widehat{a}+\widehat{a}')$. Obviously, the~$1$-cycles~$a_{mn}+a'_{mn}$ and~$\widehat{a}+\widehat{a}'$ only coincide if~$m=1$. However, they are always homologous in~$H_1(\K_{mn};\Z/2\Z)$; indeed, their difference bounds a surface~$\Sigma$ consisting of~$\lfloor(m-1)/2\rfloor$ cylinders, plus one M\"obius strip if~$m$ is even. The cases~$m=2$ and~$m=3$ are illustrated in Figure~\ref{fig:M}. (Note that~$n$ is irrelevant in this argument.) Therefore, the cocycles~$\omega_{mn}$ and~$\widehat{\omega}$ are cohomologous: they can be obtained from each other by, for each vertex~$v$ in~$\Sigma$, flipping the value of all the edges adjacent to~$v$. At the level of Kasteleyn matrices, this amounts to multiplying by~$-i$ the rows and columns of~$\widehat{A}^{\rho}$ and~$\widehat{A}^{\rho'}$ corresponding to the vertices in~$\Sigma$. As~$\Sigma$ contains~$\lfloor m/2\rfloor n|V|$ vertices and~$|V|$ is even, the determinant is left unchanged by this operation.

\begin{figure}[htb]
\centering
\includegraphics[width=0.6\textwidth]{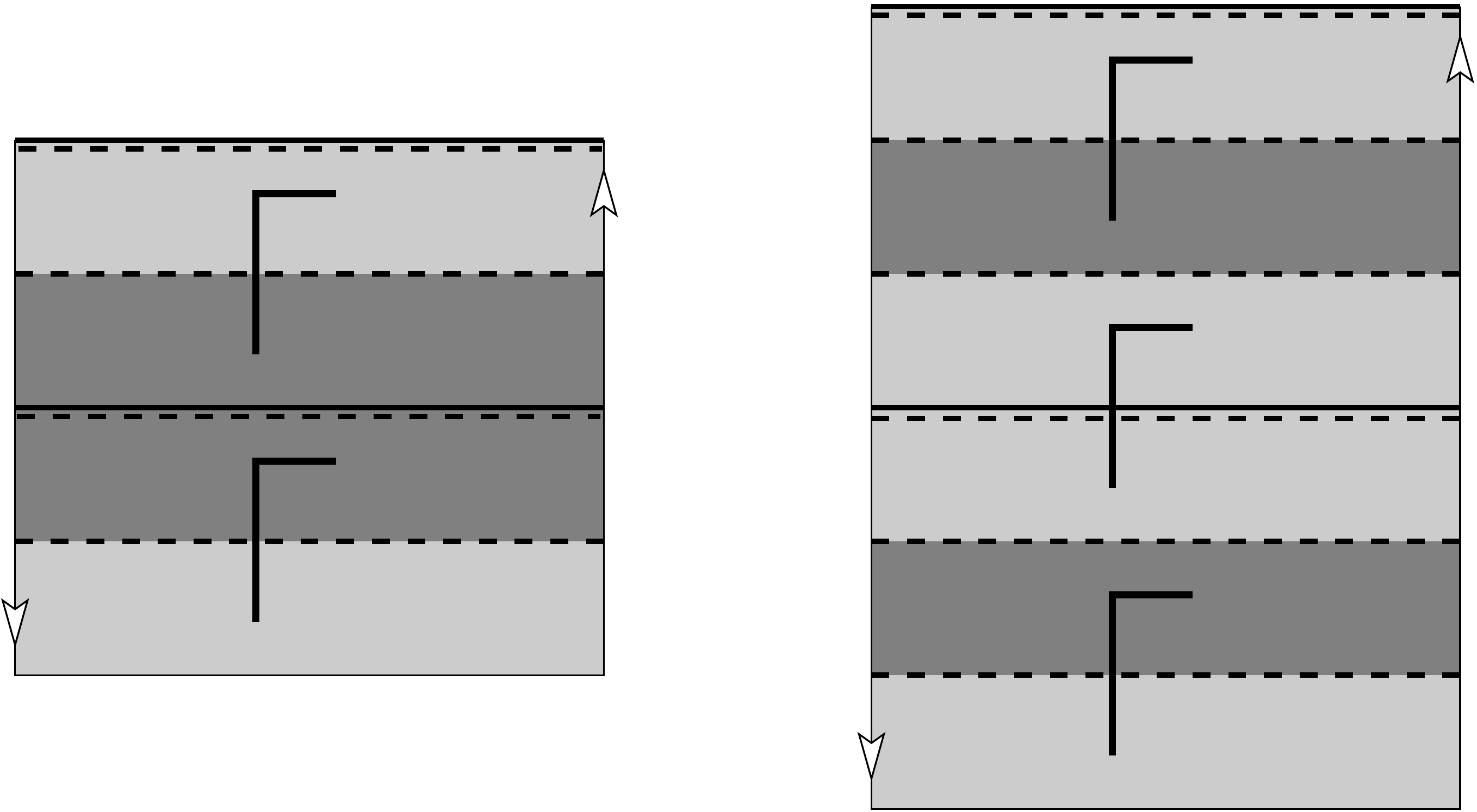}
\caption{The cycles~$a_{mn},a'_{mn}$ in full lines,~$\widehat{a},\widehat{a}'$ in dashed lines, in the cases~$m=2$ (left) and~$m=3$ (right). The surface~$\Sigma$ bounded by their difference is shaded.}
\label{fig:M}
\end{figure}

However, this operation also implies that the lifted orientation~$\widehat{K}$ is inverted along each edge both of whose endpoints are contained in~$\Sigma$, thus creating a new orientation~$K_{mn}$ on~$\G_{mn}\subset\K_{mn}$. It now remains to check that~$K_{mn}$ satisfies conditions~(i) and~(ii) of Section~\ref{sub:gen}, whose notations and terminology we assume. Once this is verified, we will able to use~$A_{mn}^\rho:=A^\rho(\G_{mn},\nu_{mn},K_{mn},\omega_{mn})$ and~$A_{mn}^{\rho'}$ to compute~$R_{mn}(1,\pm 1)$ and apply Theorem~\ref{thm:C-K} to obtain the equalities
\begin{align}
\label{equ:Zmn}
R_{mn}(1,1)&=\det(A_{mn}^\rho)=\det(\widehat{A}^\rho)=\det(A^{\rho^\#})\quad\text{and}\\
R_{mn}(1,-1)&=\det(A_{mn}^{\rho'})=\det(\widehat{A}^{\rho'})=\det(A^{(\rho')^\#})\,.
\label{equ:Rmn}
\end{align}

\medskip

To check the first condition, consider the commutative diagram of covering maps
\[
\xymatrix{
(\K_{mn},\G_{mn})\ar[d]& (\T^2_{mn},\widetilde{\G}_{mn})\ar[l]\ar[d]\\
(\K,\G)& \ar[l] (\T^2,\widetilde{\G})\,.
}
\]
Note that the orientation~$\widetilde{K_{mn}}$ on~$\widetilde{\G_{mn}}=\widetilde{\G}_{mn}\subset\T^2_{mn}$ is equal to the lift of~$K$ from the bottom-left to the upper-right of the above diagram, followed by the inversion of the edges both of whose endpoints belong to the lift of the upper half of~$\widetilde{\D}$ via~$\T^2_{mn}\to\T^2$. Hence,~$\widetilde{K_{mn}}$ is nothing but the lift of the orientation~$\widetilde{K}$ on~$\widetilde{\G}\subset\T^2$ via~$\T^2_{mn}\to\T^2$. The latter orientation being Kasteleyn by definition of~$K$, so is the former, and the first condition is satisfied.

\medskip

To check the second condition, let us denote by~$n^K(\delta)$ the (parity of the) number of edges of an oriented curve~$\delta$ where~$K$ disagrees with the orientation of~$\delta$. By assumption, we have that~$n^K(C)+n^K(C')$ is even, where~$C,C'$ are oriented curves in~$\G$ associated with~$a,a'$ as explained in Section~\ref{sub:gen} and illustrated in Figures~\ref{fig:Klein} and~\ref{fig:gamma}. We need to check that~$n^{K_{mn}}(C_{mn})+n^{K_{mn}}(C'_{mn})$ is even, with~$K_{mn}$ as above, and $C_{mn},C_{mn}'$ oriented curves in~$\G_{mn}$ associated with~$a_{mn},a'_{mn}$, respectively.

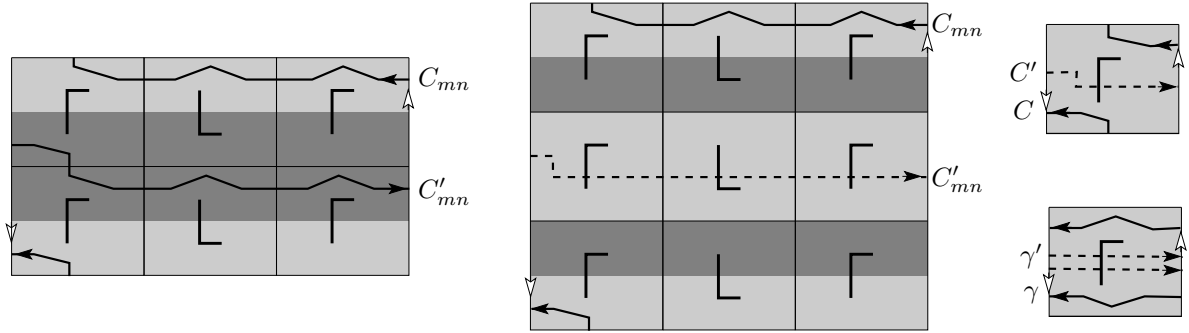
\begin{figure}[htb]
\labellist\small\hair 2.5pt
\pinlabel {$C_{mn}$} at 565 330
\pinlabel {$C'_{mn}$} at 565 184
\pinlabel {$C'_{mn}$} at 1225 200
\pinlabel {$C_{mn}$} at 1225 400
\pinlabel {$\gamma$} at 1320 50
\pinlabel {$\gamma'$} at 1320 100
\pinlabel {$C$} at 1310 290
\pinlabel {$C'$} at 1310 340
\endlabellist
\centering
\includegraphics[width=\textwidth]{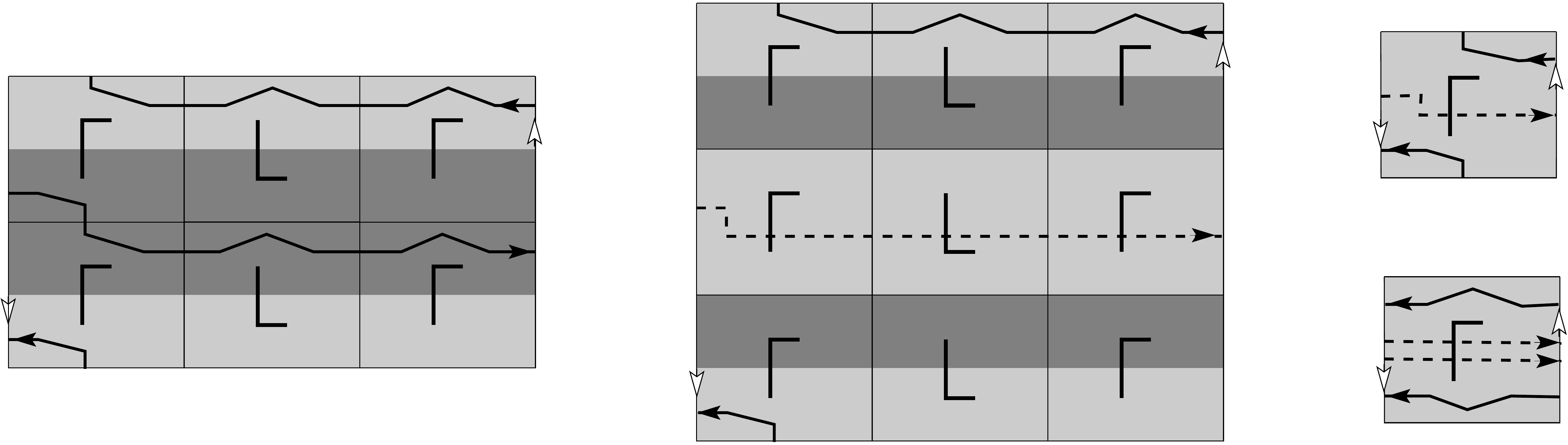}
\caption{The curves~$C_{mn}$ and~$C'_{mn}$ for~$n=3$ and~$m=2$ in the left part and~$m=3$ in the middle part. These curves can be constructed from copies of the curves~$C,C'$ and~$\gamma,\gamma'$ given to the right.}
\label{fig:gamma}
\end{figure}

By construction, the oriented curve~$C_{mn}\subset\G_{mn}\subset\K_{mn}$ can be obtained by a lift of~$C\subset\G$ together with~$(n-1)/2$ lifts of the oriented closed curve~$\gamma\subset\G\subset\K$ illustrated in the right part of Figure~\ref{fig:gamma}. (This curve can be defined as the oriented boundary of the M\"obius band obtained by all the~$2$-cells of~$\K$ meeting the curve~$a$.) By the assumption of Remark~\ref{rem:C}, the curve~$C_{mn}$ is located in the portion of~$\K_{mn}$ where the orientation~$K_{mn}$ coincides with the lift~$\widehat{K}$ of~$K$ (unshaded in Figure~\ref{fig:gamma}). Hence, we have
\[
n^{K_{mn}}(C_{mn})=n^{\widehat{K}}(C_{mn})=n^K(C)+\frac{n-1}{2}n^K(\gamma)\,.
\]
If~$m$ is even (Figure~\ref{fig:gamma}, left), then~$C'_{mn}$ can be obtained by a lift of~$-C$ together with~$(n-1)/2$ lifts of~$-\gamma$. But these curves are located in the portion of~$\K_{mn}$ where~$K_{mn}$ and~$\widehat{K}$ disagree (shaded in Figure~\ref{fig:gamma}). Hence, we have
\[
n^{K_{mn}}(C'_{mn})=n^{-\widehat{K}}(C'_{mn})=n^{-K}(-C)+\frac{n-1}{2}n^{-K}(-\gamma)=n^K(C)+\frac{n-1}{2}n^K(\gamma)\,.
\]
If~$m$ is odd (Figure~\ref{fig:gamma}, middle), then~$C'_{mn}$ can be obtained by a lift of~$C'$ and~$(n-1)/2$ lifts of the oriented curve~$\gamma'\subset\G\subset\K$ illustrated in right part of Figure~\ref{fig:gamma}. Since~$K_{mn}$ coincides with~$\widehat{K}$ along~$C'_{mn}$, we obtain
\[
n^{K_{mn}}(C'_{mn})=n^{\widehat{K}}(C'_{mn})=n^K(C')+\frac{n-1}{2}n^K(\gamma')\,.
\]
By the first two equations displayed above, we see that~$n^{K_{mn}}(C_{mn})+n^{K_{mn}}(C'_{mn})$ is always even for~$m$ even. For~$m$ odd, the first and third equations above together with the fact that~$n^K(C)+n^K(C')$ is even imply that we are left with the proof that~$n^K(\gamma)$ and~$n^K(\gamma')$ have the same parity.

This can be checked as follows. The curves~$\gamma,\gamma'\subset\G\subset\K$ lift to closed curves~$\widetilde{\gamma},\widetilde{\gamma}'\subset\widetilde{\G}$ in the lower half-part of~$\widetilde{\D}$, bounding a cylinder~$\mathcal{C}\subset\T^2$ containing the same number of vertices as~$\G$. Since~$K$ coincides with~$\widetilde{K}$ there, we have
\[
n^K(\gamma)+n^K(\gamma')=n^{\widetilde{K}}(\widetilde{\gamma})+n^{\widetilde{K}}(\widetilde{\gamma}')=n^{\widetilde{K}}(\partial\mathcal{C})\,.
\]
The fact that~$n^{\widetilde{K}}(\partial\mathcal{C})$ is even follows from the standard argument of Kasteleyn since~$\widetilde{K}$ is a Kasteleyn orientation, the cylinder~$\mathcal{C}$ contains an even number of vertices, and has even Euler characteristic (see e.g.~\cite{Ka3}).

\begin{remark}
\label{rem:Kmn}
When the graph is ``too small'', it might happen that the curves~$C$ and~$C'$ cannot be chosen to be disjoint from~$a'$ and~$a$, respectively (recall Remark~\ref{rem:C}). For~$m$ even, a mild extension of the argument above shows that~$n^{K_{mn}}(C_{mn})+n^{K_{mn}}(C'_{mn})$ is nevertheless always even. For~$m$ odd, this might no longer hold: the square lattice of Figure~\ref{fig:Klein} is an example of such a phenomenon. We shall deal with the necessary modifications in due time (see Remarks~\ref{rem:perK} and~\ref{rem:as-gen}~(iii) below). 
\end{remark}

\subsection{Identifying and factorising the induced representations}
\label{sub:repr}

We now proceed to the computation of the representations of~$\pi_1(\K)$ induced by the~$1$-dimensional representations of~$\pi_1(\K_{mn})$ arising in Equations~\eqref{equ:Zmn} and~\eqref{equ:Rmn}.

To do so, recall (following~\cite[Section~3.3]{Serre}) that given a representation~$\rho\colon H\to\operatorname{GL}(W)$ of a subgroup~$H<G$, the induced representation~$\rho^\#\colon G\to\operatorname{GL}(Z)$ is uniquely determined up to isomorphism by the following properties. If~$R\subset G$ denotes a system of representatives of~$G/H$ (i.e. each~$g\in G$ can be written uniquely as~$g=rh\in G$ with~$r\in R$ and~$h\in H$), then~$Z$ is given by the direct sum
\[
Z=\bigoplus_{r\in R}\rho_r^\#(W)\,,
\]
and for any~$g\in G$ and~$w\in W$, we have~$\rho_g^\#(\rho_r^\#(w))=\rho_{r'}^\#(\rho_h(w))$ where~$gr=r'h\in G$ with~$r'\in R$ and~$h\in H$.

In our case, the covering map~$p\colon\K_{mn}\to\K$ determines the inclusion of fundamental groups
\[
\pi_1(\K_{mn})=\left<a^n,b^m\,|\,a^nb^ma^{-n}b^m\right>\stackrel{p_*}{\longrightarrow}\left<a,b\,|\,aba^{-1}b\right>=\pi_1(\K)\,.
\]  
Furthermore, a natural system of representatives of the quotient~$\pi_1(\K)/\pi_1(\K_{mn})$ is given by~$R=\{b^ia^j\,|\,0\le i<m,\;0\le j<n\}$. Hence, we need for~$g=a$ and~$g=b$ to express~$g\cdot b^ia^j\in\pi_1(\K)$ in the form~$r\cdot h$ with~$r\in R$ and~$h\in\pi_1(\K_{mn})$. Using the relation~$aba^{-1}b=1$ (or equivalently, the relation~$a^{-1}b^{\pm 1}a=b^{\mp 1}$), we find
\[
b\cdot(b^ia^j)=\begin{cases}
b^{i+1}a^j\cdot 1&\mbox{if~$i<m-1$;} \\
b^ma^j=a^j(a^{-j}b^ma^j)=a^j(a^{-j}ba^j)^m=a^j\cdot b^{m(-1)^j}&\mbox{if~$i=m-1$,}
\end{cases}
\]
and for~$i>0$,
\[
a\cdot(b^ia^j)=b^{-i}a^{j+1}=b^{m-i}a^{j+1}(a^{-j-1}b^{-m}a^{j+1})=
\begin{cases}
b^{m-i}a^{j+1}\cdot b^{m(-1)^{j}}&\mbox{if~$j<n-1$;} \\
b^{m-i}\cdot a^n b^{m(-1)^{j}}&\mbox{if~$j=n-1$,}
\end{cases}
\]
while for~$i=0$,
\[
a\cdot(b^0a^j)=
\begin{cases}
a^{j+1}\cdot 1&\mbox{if~$j<n-1$;} \\
1\cdot a^n&\mbox{if~$j=n-1$.}
\end{cases}
\]
For~$\rho=1$ the trivial representation of~$\pi_1(\K_{mn})$ and~$\rho'$ given by~$\rho'_{a^n}=1$ and~$\rho'_{b^m}=-1$, we hence obtain the following result.

\begin{lemma}
\label{lemma:ind}
Let~$Z$ be the vector space with basis~$\{e(i,j)\,|\,0\le i<m,\;0\le j<n\}$. Then, the induced representations~$\rho^\#,(\rho')^\#\colon\pi_1(\K)=\left<a,b\,|\,aba^{-1}b\right>\to\GL(Z)$ are determined by
\begin{align*}
\rho_a^\#(e(i,j))&=\begin{cases}
e(m-i,j+1)&\mbox{if~$i>0,j<n-1$,}\\
e(m-i,0)&\mbox{if~$i>0,j=n-1$,}\\
e(0,j+1)&\mbox{if~$i=0,j<n-1$,}\\
e(0,0)&\mbox{if~$i=0,j=n-1$,}
\end{cases}
\\
\rho_b^\#(e(i,j))&=\begin{cases}
e(i+1,j)&\mbox{if~$i<m-1$,}\\
e(0,j)&\mbox{if~$i=m-1$,}
\end{cases}
\end{align*}
and
\begin{align*}
\pushQED{\qed}
(\rho')^\#_a(e(i,j))&=\begin{cases}
-e(m-i,j+1)&\mbox{if~$i>0,j<n-1$,}\\
-e(m-i,0)&\mbox{if~$i>0,j=n-1$,}\\
e(0,j+1)&\mbox{if~$i=0,j<n-1$,}\\
e(0,0)&\mbox{if~$i=0,j=n-1$,}
\end{cases}
\\
(\rho')^\#_b(e(i,j))&=\begin{cases}
e(i+1,j)&\mbox{if~$i<m-1$,}\\
-e(0,j)&\mbox{if~$i=m-1$.}
\end{cases}
\qedhere
\popQED
\end{align*}
\end{lemma}

Arranging the basis vectors~$\{e(i,j)\,|\,0\le i<m,\;0\le j<n\}$ as the vertices of an~$(m\times n)$-grid with periodic-antiperiodic boundary conditions, we can understand~$\rho^\#$ as the representation permuting these vertices as illustrated in Figure~\ref{fig:permute}. Similarly, we can understand~$(\rho')^\#$ as a signed permutation representation, with the signs given in Figure~\ref{fig:permute}.

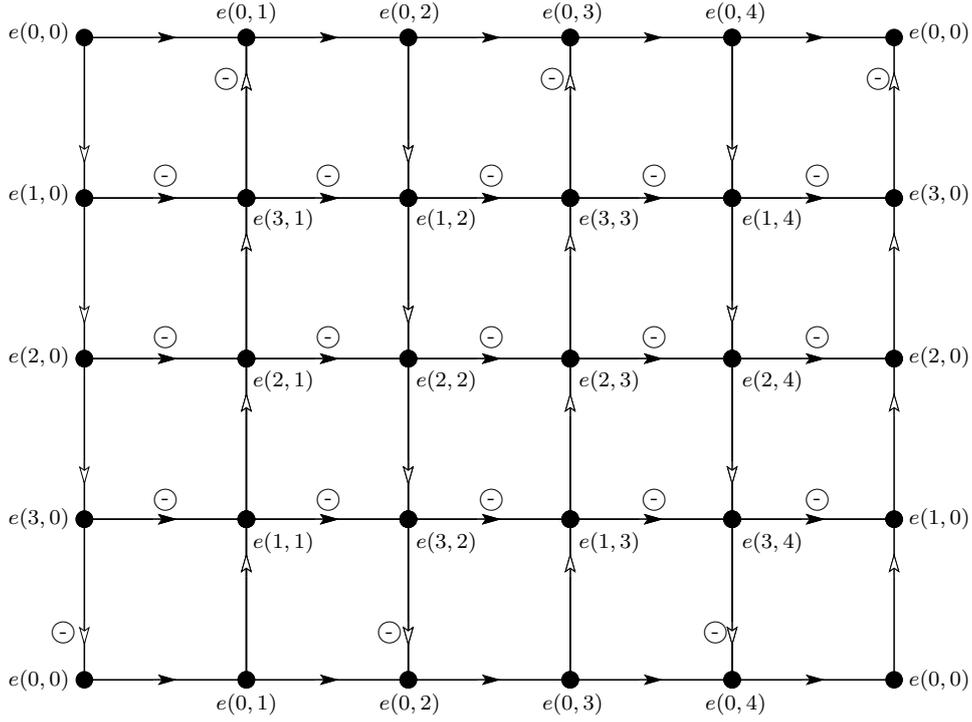
\begin{figure}[tb]
\labellist\small\hair 2.5pt
\pinlabel {\scriptsize \textcircled{-}} at 110 230
\pinlabel {\scriptsize \textcircled{-}} at 110 430
\pinlabel {\scriptsize \textcircled{-}} at 110 630
\pinlabel {\scriptsize \textcircled{-}} at 310 230
\pinlabel {\scriptsize \textcircled{-}} at 310 430
\pinlabel {\scriptsize \textcircled{-}} at 310 630
\pinlabel {\scriptsize \textcircled{-}} at 510 230
\pinlabel {\scriptsize \textcircled{-}} at 510 430
\pinlabel {\scriptsize \textcircled{-}} at 510 630
\pinlabel {\scriptsize \textcircled{-}} at 710 230
\pinlabel {\scriptsize \textcircled{-}} at 710 430
\pinlabel {\scriptsize \textcircled{-}} at 710 630
\pinlabel {\scriptsize \textcircled{-}} at 910 230
\pinlabel {\scriptsize \textcircled{-}} at 910 430
\pinlabel {\scriptsize \textcircled{-}} at 910 630
\pinlabel {\scriptsize \textcircled{-}} at -15 65
\pinlabel {\scriptsize \textcircled{-}} at 385 65
\pinlabel {\scriptsize \textcircled{-}} at 785 65
\pinlabel {\scriptsize \textcircled{-}} at 185 750
\pinlabel {\scriptsize \textcircled{-}} at 585 750
\pinlabel {\scriptsize \textcircled{-}} at 985 750
\pinlabel {\scriptsize $e(0,0)$} at -45 810
\pinlabel {\scriptsize $e(1,0)$} at -45 610
\pinlabel {\scriptsize $e(2,0)$} at -45 410
\pinlabel {\scriptsize $e(3,0)$} at -45 210
\pinlabel {\scriptsize $e(0,0)$} at -45 10
\pinlabel {\scriptsize $e(2,0)$} at 1060 410
\pinlabel {\scriptsize $e(1,0)$} at 1060 210
\pinlabel {\scriptsize $e(0,0)$} at 1060 10
\pinlabel {\scriptsize $e(3,0)$} at 1060 610
\pinlabel {\scriptsize $e(0,0)$} at 1060 810
\pinlabel {\scriptsize $e(0,1)$} at 210 835
\pinlabel {\scriptsize $e(0,2)$} at 410 835
\pinlabel {\scriptsize $e(0,3)$} at 610 835
\pinlabel {\scriptsize $e(0,4)$} at 810 835
\pinlabel {\scriptsize $e(0,1)$} at 210 -20
\pinlabel {\scriptsize $e(0,2)$} at 410 -20
\pinlabel {\scriptsize $e(0,3)$} at 610 -20
\pinlabel {\scriptsize $e(0,4)$} at 810 -20
\pinlabel {\scriptsize $e(3,1)$} at 255 580
\pinlabel {\scriptsize $e(1,2)$} at 455 580
\pinlabel {\scriptsize $e(3,3)$} at 655 580
\pinlabel {\scriptsize $e(1,4)$} at 855 580
\pinlabel {\scriptsize $e(2,1)$} at 255 380
\pinlabel {\scriptsize $e(2,2)$} at 455 380
\pinlabel {\scriptsize $e(2,3)$} at 655 380
\pinlabel {\scriptsize $e(2,4)$} at 855 380
\pinlabel {\scriptsize $e(1,1)$} at 255 180
\pinlabel {\scriptsize $e(3,2)$} at 455 180
\pinlabel {\scriptsize $e(1,3)$} at 655 180
\pinlabel {\scriptsize $e(3,4)$} at 855 180
\endlabellist
\centering
\includegraphics[width=0.7\textwidth]{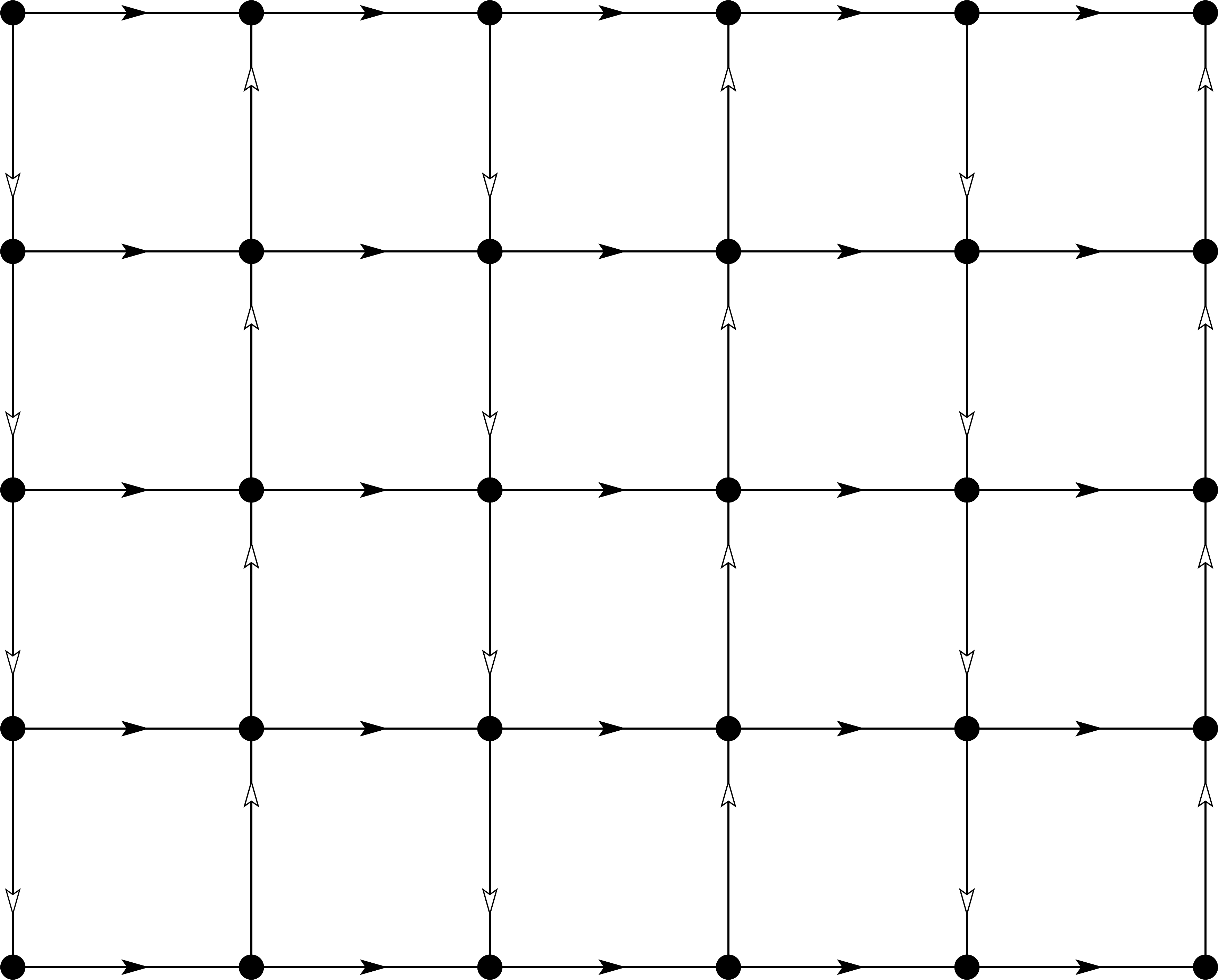}
\bigskip
\caption{The induced representation~$\rho^\#$ as a permutation representation, in the case~$m=4$,~$n=5$. The horizontal arrows correspond to the action of~$a$, the vertical ones to the action of~$b$. The circled minus signs {\scriptsize{\textcircled{-}}} indicate how~$(\rho')^\#$ differs from~$\rho^\#$.}
\label{fig:permute}
\end{figure}

\medskip

The next step is to determine the factorisation of the~$mn$-dimensional representations~$\rho^\#$ and~$(\rho')^\#$ into irreducible representations.
To do so, first observe that Lemma~\ref{lemma:ind} implies the equalities~$\rho_{a^{2n}}^\#=\rho_{b^m}^\#=\mathit{id}_Z$ (this is clear from Figure~\ref{fig:permute}). As a consequence, the representation~$\rho^\#$ factors through the natural projection of~$\pi_1(\K)$ onto the finite group given by the semi-direct product of two cyclic groups of order~$m$ and~$2n$:
\[
\pi_1(\K)=\left<a,b\,|\,aba^{-1}b\right>\twoheadrightarrow\left<a,b\,|\,a^{2n},b^m,aba^{-1}b\right>=C_m\rtimes C_{2n}\,.
\]
Similarly, one checks that~$(\rho')_{a^{2n}}^\#=(\rho')_{b^{2m}}^\#=\mathit{id}_Z$, so~$(\rho')^\#$ factors through the natural projection of~$\pi_1(\K)$ onto~$C_{2m}\rtimes C_{2n}$. By abuse of notation, we shall simply denote by~$\rho^\#$ (resp.~$(\rho')^\#$) the corresponding representation of~$C_m\rtimes C_{2n}$ (resp.~$C_{2m}\rtimes C_{2n}$), and we now need to understand the irreducible representations of the finite group
\[
\left<a,b\,|\,a^{2n},b^M,aba^{-1}b\right>=C_M\rtimes C_{2n}\,.
\]
Writing~$\zeta_M=\exp(2i\pi/M)$ and~$\zeta_{2n}=\exp(i\pi/n)$, the following assignements clearly define such representations.
\begin{itemize}
\item{For~$1\le\ell\le 2n$, the homomorphism~$\rho^\ell\colon C_M\rtimes C_{2n}\to\C^*$ given by~$a\mapsto\zeta_{2n}^\ell$ and~$b\mapsto 1$.}
\item{If~$M$ is even, for~$1\le\ell\le 2n$, the homomorphism~$\rho'^\ell\colon C_M\rtimes C_{2n}\to\C^*$ given by~$a\mapsto\zeta_{2n}^\ell$ and~$b\mapsto -1$.}
\item{For~$1\le k<M/2$ and~$1\le\ell\le n$, the homomorphism~$\rho^{k,\ell}\colon C_M\rtimes C_{2n}\to\operatorname{GL}_2(\C)$ given by
\[
a\mapsto\begin{pmatrix}
0&\zeta_{2n}^\ell\cr \zeta_{2n}^\ell&0
\end{pmatrix}
\quad\text{and}\quad
b\mapsto
\begin{pmatrix}
\zeta_{M}^k&0\cr 0&\overline{\zeta}_{M}^{k}
\end{pmatrix}\,.
\]}
\end{itemize}

To avoid very cumbersome notations, we do not include the index~$M$ in~$\rho^{k,\ell}$. However, the reader should keep in mind that the expression~$\rho^{k,\ell}$ will refer to the above representation, sometimes with~$M=m$ (when dealing with~$\rho^\#$), and sometimes with~$M=2m$ (when dealing with~$(\rho')^\#$).

\begin{lemma}
\label{lemma:irreps}
The representations~$\{\rho^\ell,\rho'^\ell\}_\ell$ and~$\{\rho^{k,\ell}\}_{k,\ell}$ above give the full list of irreducible representations of~$C_M\rtimes C_{2n}$ up to isomorphism.
\end{lemma}
\begin{proof}
The abelianisation of~$C_M\rtimes C_{2n}$ is given by~$\left<a\,|\,a^{2n}\right>$ if~$M$ is odd and by~$\left<a\,|\,a^{2n}\right>\times\left<b\,|\,b^{2}\right>$ if~$M$ is even. This immediately implies that the full list of~$1$-dimensional representations of~$C_M\rtimes C_{2n}$ is given by~$\{\rho^\ell,\rho'^\ell\}_\ell$ as above.

To analyse the~$2$-dimensional representations, let us denote by~$\chi_{k,\ell}\colon C_M\rtimes C_{2n}\to\C$ the character of~$\rho^{k,\ell}$ (see~\cite[Chapter~2]{Serre}). By definition, we get for~$0\le i\le M-1$ and~$0\le j\le 2n-1$
\begin{equation}
\label{equ:char}
\chi_{k,\ell}(b^ia^j)=\begin{cases}
\zeta_{2n}^{\ell j}(\zeta_M^{ki}+\overline{\zeta}_M^{ki})&\mbox{if~$j$ is even;} \\
0 & \mbox{if~$j$ is odd.}
\end{cases}
\end{equation}
For any~$1\le k,k'<M/2$ and~$1\le\ell,\ell'\le n$, this leads to
\begin{align*}
\left(\chi_{k,\ell}|\chi_{k',\ell'}\right)&:=\frac{1}{2nM}\sum_{i=0}^{M-1}\sum_{j=0}^{2n-1}\chi_{k,\ell}(b^ia^j)\overline{\chi_{k',\ell'}(b^ia^j)}\\
&=\frac{1}{n}\sum_{j=0}^{n-1}\zeta_n^{(\ell-\ell')j}\;\frac{1}{2M}\sum_{i=0}^{M-1}(\zeta_M^{(k+k')i}+\overline{\zeta}_M^{(k+k')i}+\zeta_M^{(k-k')i}+\overline{\zeta}_M^{(k-k')i})\\
&=\begin{cases}
1&\mbox{if~$k=k'$ and~$\ell=\ell'$;} \\
0 & \mbox{else.}
\end{cases}
\end{align*}
In other words, the characters~$\chi_{k,\ell}$ are orthogonal with respect to the usual scalar product.
By~\cite[Section~2.3]{Serre}, this implies that the representations~$\{\rho^{k,\ell}\}_{k,\ell}$ are irreducible and pairwise non-isomorphic.

Finally, note that the sum of the square of the degrees of the representations listed above gives, for~$M$ odd,
\[
2n\cdot 1^2+n\floor*{\frac{M-1}{2}}\cdot 2^2=2n+2n(M-1)=2nM=|C_M\rtimes C_{2n}|
\]
and for~$M$ even,
\[
4n\cdot 1^2+n\floor*{\frac{M-1}{2}}\cdot 2^2=4n+4n\left(\frac{M}{2}-1\right)=2nM=|C_M\rtimes C_{2n}|\,.
\]
By~\cite[Corollary~2]{Serre}, this shows that the above list is complete.
\end{proof}

We are now ready to compute the decomposition of the induced representations into irreducible ones. This is the content of the next lemma.

\begin{lemma}
\label{lemma:dec}
The decomposition of the representation~$\rho^\#$ of~$C_m\rtimes C_{2n}$ into irreducible representations is given by
\[
\rho^\#=\bigoplus_{\ell=1}^n\Bigg(\rho^{2\ell}\oplus\bigoplus_{1\le k<\frac{m}{2}}\rho^{k,\ell}\Bigg)
\]
if~$m$ is odd, and by
\[
\rho^\#=\bigoplus_{\ell=1}^n\Bigg(\rho^{2\ell}\oplus\rho'^{2\ell}\oplus\bigoplus_{1\le k<\frac{m}{2}}\rho^{k,\ell}\Bigg)
\]
if~$m$ is even. The decomposition of the representation~$(\rho')^\#$ of~$C_{2m}\rtimes C_{2n}$ into irreducible representations is given by
\[
(\rho')^\#=\bigoplus_{\ell=1}^n\Bigg(\rho'^{2\ell}\oplus\bigoplus_{\stackrel{1\le k\le m-1}{k\text{ odd}}}\rho^{k,\ell}\Bigg)
\]
if~$m$ is odd, and by
\[
(\rho')^\#=\bigoplus_{\ell=1}^n\bigoplus_{\stackrel{1\le k\le m-1}{k\text{ odd}}}\rho^{k,\ell}
\]
if~$m$ is even. 
\end{lemma}
\begin{proof}
The idea is to use once again the theory of characters. More precisely, we shall determine the character~$\chi$ of~$\rho^\#$ and compute its scalar product with the characters~$\chi_\ell,\chi'_\ell$ and~$\chi_{k,\ell}$ of the representations~$\rho^\ell,\rho'^\ell$ and~$\rho^{k,\ell}$, respectively. This scalar product is nothing but the number of times that the corresponding irreducible representation appears in the decomposition of~$\rho^\#$ (see~\cite[Chapter~2]{Serre}).

Understanding~$\rho^\#$ as a permutation representation (recall the square grid of Figure~\ref{fig:permute}), one sees that~$\chi(b^ia^j):=\operatorname{Tr}(\rho^\#_{b^ia^j})$ is simply given by the number of vertices of the grid that are fixed by the action of~$b^ia^j$. For~$0\le i<m$ and~$0\le j<2n$, this leads to
\[
\chi(b^ia^j)=\begin{cases}
mn&\mbox{for~$i=j=0$;}\\
n&\mbox{for~$0\le i<m$ and~$j=n$, if~$m$ is odd;}\\
2n&\mbox{for~$0\le i<m$ even and~$j=n$, if~$m$ is even;}\\
0&\mbox{else.}
\end{cases}
\]
Hence, for any character~$\psi$ of~$C_m\rtimes C_{2n}$, we have
\[
\left(\chi|\psi\right)=\frac{1}{2mn}\Big(mn\,\psi(b^0a^0)+n\sum_{i=0}^{m-1}\psi(b^ia^n)\Big)=\frac{1}{2}\psi(b^0a^0)+\frac{1}{2m}\sum_{i=0}^{m-1}\psi(b^ia^n)
\]
if~$m$ is odd, and
\[
\left(\chi|\psi\right)=\frac{1}{2}\psi(b^0a^0)+\frac{1}{m}\sum_{\stackrel{0\le i\le m-1}{i\text{ even}}}\psi(b^ia^n)
\]
if~$m$ is even. Applying this to~$\psi=\chi_\ell$ which satisfies~$\chi_\ell(b^0a^0)=1$ and~$\chi_\ell(b^ia^n)=(-1)^\ell$, we obtain
\[
(\chi|\chi_\ell)=\begin{cases}1&\mbox{if~$\ell$ is even;}\\ 0&\mbox{if~$\ell$ is odd.}\end{cases}
\]
If~$m$ is even, then the exact same computation holds for~$\chi'_\ell$ as well. Finally, Equation~\eqref{equ:char}  applied to~$M=m$ gives~$\chi_{k,\ell}(b^0a^0)=2$ and~$\chi_{k,\ell}(b^ia^n)=0$ since~$n$ is odd, so~$(\chi|\chi_{k,\ell})=1$ for all~$1\le\ell\le n$ and~$1\le k<m/2$. This concludes the proof of the first assertion. (As a reality check, note that the degree of the right-hand side is equal to~$n(1+2\frac{m-1}{2})=nm$ if~$m$ is odd and to~$n(2+2(\frac{m}{2}-1))=nm$ if~$m$ is even, which is indeed the degree of~$\rho^\#$.)

Let us now turn to the representation~$(\rho')^\#$ of~$C_{2m}\rtimes C_{2n}$. As described in Lemma~\ref{lemma:ind}, this is no longer a simple permutation representation, but a signed one. Therefore, the associated character~$\chi'$ evaluated at~$b^ia^j$ is not simply given by the number of vertices fixed by the action of~$b^ia^j$: it is equal to the signed sum of these fixed vertices, with the signs given by Lemma~\ref{lemma:ind} (see also Figure~\ref{fig:permute}). For~$0\le i<2m$ and~$0\le j<2n$, we obtain
\[
\chi'(b^ia^j)=\begin{cases}
mn&\mbox{for~$i=j=0$;}\\
-mn&\mbox{for~$i=m$ and~$j=0$;}\\
(-1)^in&\mbox{for~$0\le i<2m$ and~$j=n$, if~$m$ is odd;}\\
0&\mbox{else.}
\end{cases}
\]
Hence, for any character~$\psi$ of~$C_{2m}\rtimes C_{2n}$, we have
\[
\left(\chi'|\psi\right)=\begin{cases}\frac{1}{4}(\psi(b^0a^0)-\psi(b^ma^0))+\frac{1}{4m}\sum_{i=0}^{2m-1}(-1)^i\psi(b^ia^n)&\mbox{if~$m$ is odd;}\\
\frac{1}{4}(\psi(b^0a^0)-\psi(b^ma^0))&\mbox{if~$m$ is even.}
\end{cases}
\]
Applying this to~$\psi=\chi_\ell$ which satisfies~$\chi_\ell(b^0a^0)=\chi(b^ma^0)=1$ and~$\chi_\ell(b^ia^n)=(-1)^\ell$, we obtain~$(\chi'|\chi_\ell)=0$. On the other hand, the character~$\chi'_\ell$ satisfies~$\chi'_\ell(b^0a^0)=1$,~$\chi'(b^ma^0)=(-1)^m$ and~$\chi_\ell(b^ia^n)=(-1)^{i+\ell}$, so we get
\[
(\chi'|\chi'_\ell)=\begin{cases}1&\mbox{if~$m$ is odd and~$\ell$ is even;}\\ 0&\mbox{else.}\end{cases}
\]
Finally, Equation~\eqref{equ:char} applied to~$M=2m$ gives the values~$\chi_{k,\ell}(b^0a^0)=2$,~$\chi_{k,\ell}(b^ma^0)=2(-1)^k$ and~$\chi_{k,\ell}(b^ia^n)=0$ for all~$1\le\ell\le n$ and~$1\le k<m$. This leads to
\[
(\chi'|\chi_{k,\ell})=\begin{cases}1&\mbox{if~$k$ is odd;}\\ 0&\mbox{if~$k$ is even,}\end{cases}
\]
and concludes the proof of the lemma. (Once again, one easily checks that the degree of the right-hand side is equal~$nm$, which is the degree of~$(\rho')^\#$.)
\end{proof}

\subsection{Proof of Theorem~\ref{thm:Rmn}}
\label{sub:proof}

We start with the first part of Theorem~\ref{thm:Rmn}.
By Equation~\eqref{equ:Zmn}, Lemma~\ref{lemma:dec} and Remark~\ref{rem:twist}~(ii), we have
\[
R_{mn}(1,1)=\begin{cases}
\prod_{\ell=1}^n\left(\det(A^{\rho^{2\ell}})\prod_{1\le k<m/2}\det(A^{\rho^{k,\ell}})\right)&\mbox{if~$m$ is odd;} \\
\prod_{\ell=1}^n\left(\det(A^{\rho^{2\ell}})\det(A^{\rho'^{2\ell}})\prod_{1\le k<m/2}\det(A^{\rho^{k,\ell}})\right)& \mbox{if~$m$ is even,}
\end{cases}
\]
with~$\rho^{2\ell}$,~$\rho'^{2\ell}$ and~$\rho^{k,\ell}$ the representations of~$\pi_1(\K)=\left<a,b\,|\,aba^{-1}b\right>$ determined by
\[
\rho^{2\ell}_a=\rho'^{2\ell}_a=\zeta_n^\ell,\quad\rho^{2\ell}_b=1,\quad\rho'^{2\ell}_b=-1,\quad
\rho^{k,\ell}_a=\begin{pmatrix}
0&\zeta_{2n}^\ell\cr \zeta_{2n}^\ell&0
\end{pmatrix}\,,
\quad\text{and}\quad
\rho^{k,\ell}_b=
\begin{pmatrix}
\zeta_{m}^k&0\cr 0&\overline{\zeta}_{m}^{k}
\end{pmatrix}\,.
\]
It is easy to check that any representation~$\rho\colon\pi_1(\K)\to\operatorname{GL}(W)$ can be represented by the following connection~$\Phi=(\varphi_e)_{e\in\mathbb{E}}$ on the oriented edges of~$\G$:
\[
\varphi_e=\begin{cases}
\rho_{ab} &\mbox{if~$e$ crosses a vertical side of the fundamental domain~$\D$ from left to right;} \\
\rho^{-1}_{ab} &\mbox{if~$e$ crosses a vertical side of~$\D$ from right to left;} \\
\rho_b &\mbox{if~$e$ crosses a horizontal side of~$\D$ from bottom to top;} \\
\rho^{-1}_b &\mbox{if~$e$ crosses a horizontal side of~$\D$ from top to bottom;} \\
\mathit{id}_W &\mbox{else.}
\end{cases}
\]
(Observe that the loop based at the point of~$\K$ corresponding to the corners of the square and following the curve~$a'$ is homotopic to~$ab$, hence the first equality above.)
As mentioned in Remark~\ref{rem:twist}~(i), there are other possible choices of connections for~$\rho$, but the determinants of the resulting twisted operators will coincide. Using this choice of connection and the definition of~$R$, we obtain the equalities
\[
\det(A^{\rho^{2\ell}})=R(\zeta_n^\ell,1)\quad\text{and}\quad\det(A^{\rho'^{2\ell}})=R(-\zeta_n^\ell,-1)\,.
\]

The first part of Theorem~\ref{thm:Rmn} now follows from one final lemma.

\begin{lemma}
\label{lemma:P}
For any~$1\le\ell\le n$ and~$1\le k<m/2$, we have~$\det(A^{\rho^{k,\ell}})=P(\zeta_n^\ell,\zeta_m^k)$.
\end{lemma}
\begin{proof}
This proof is analogous to the demonstration of the second point of Proposition~\ref{prop:P}, whose notation we assume.
For any~$z,w\in\C^*$, observe that~$P(z^2,w)=\det\widetilde{A}(z^2,w)$ is left unchanged when replacing~$z^{2(e\cdot\tilde{b})}$ by~$z^{(e\cdot\tilde{b}-e\cdot\tilde{b}')}$, where~$\tilde{b},\tilde{b}'\subset\T^2$ denote the two lifts of~$b\subset\K$ illustrated in Figure~\ref{fig:Klein}. Also, multiplying by~$i$ the rows and columns of~$\widetilde{A}(z^2,w)$ corresponding to a vertex in the upper half of~$\widetilde{\D}$ amounts to multiplying its determinant by~$(-1)^{|V|}=1$, so the resulting matrix~$\widetilde{A}'(z^2,w)$ still has determinant equal to~$P(z^2,w)$. However, numbering the vertices of~$\widetilde{\G}$ in the right way, we see that~$\widetilde{A}'(z^2,w)$ is nothing but the twisted Kasteleyn matrix~$A^{\rho(z,w)}$, with~$\rho(z,w)$ the representation of~$\pi_1(\K)$ given by~$ab\mapsto\left(\begin{smallmatrix}0&z\cr z&0\end{smallmatrix}\right)$,~$b\mapsto\left(\begin{smallmatrix}w&0\cr 0&w^{-1}\end{smallmatrix}\right)$.
For any fixed~$z,w\in\C^*$, this representation is easily seen to be conjugate to~$\rho(z,w)'$ given by
\[
a\mapsto\begin{pmatrix}0&z\cr z&0\end{pmatrix}\,,\quad
b\mapsto\begin{pmatrix}
w&0\cr 0&w^{-1}\end{pmatrix}\,.
\]
Applying the resulting equality~$P(z^2,w)=\det(A^{\rho(z,w)'})$ to~$(z,w)=(\zeta_{2n}^\ell,\zeta_m^k)$ concludes the proof.
\end{proof}

The proof of the second part of Theorem~\ref{thm:Rmn} is similar: simply use Equation~\eqref{equ:Rmn} instead of Equation~\eqref{equ:Zmn}, the second part of Lemma~\ref{lemma:dec} instead of the first one, and replace~$\zeta_m$ by~$\zeta_{2m}$ in the definition of~$\rho^{k,\ell}$ (which is now defined on~$C_M\rtimes C_{2n}$ with~$M=2m$ and no longer~$m=M$). The statement corresponding to Lemma~\ref{lemma:P} now reads~$\det(A^{\rho^{k,\ell}})=P(\zeta_n^\ell,\zeta_{2m}^k)$, and the desired expression for~$R_{mn}(1,-1)$ follows readily.

This concludes the proof of Theorem~\ref{thm:Rmn}.

\begin{remark}
\label{rem:perK}
As mentioned in Remark~\ref{rem:C}, we made the additional assumption that the curve~$C$ and~$C'$ can be chosen disjoint from~$a'$ and~$a$, respectively, in order for Theorems~\ref{thm:Rmn} and~\ref{thm:Zmn} to hold as stated. When the graph is ``too small'', such as the~$1\times 2$ square lattice of Figure~\ref{fig:ex2-4}, this assumption is not satisfied, and these statements need some adjustments.

For~$m$ even, the orientation~$K_{mn}$ always satisfied Condition~(ii) of Section~\ref{sub:gen} (recall Remark~\ref{rem:Kmn}), so Theorems~\ref{thm:Rmn} and~\ref{thm:Zmn} hold unchanged. For~$m$ odd, it can happen that~$K_{mn}$ does not satisfy Condition~(ii). In such a case, the roles of~$R_{mn}(1,1)$ and~$R_{mn}(1,-1)$ are exchanged in the statement of Theorem~\ref{thm:Rmn}. As a direct consequence, Theorem~\ref{thm:Zmn} for~$m$ odd now reads
\begin{equation}
\label{equ:small}
Z_{mn}=\left|\sin(\alpha'_n/2)\right|P_{mn}(1,1)^{1/4}+\left|\cos(\alpha_n/2)\right|P_{mn}(1,-1)^{1/4}\,.
\end{equation}
We shall use this amended formula in Remark~\ref{rem:as-gen}~(iii) and Example~\ref{ex:fsc-sq} below.
\end{remark}


\section{On the asymptotics of the dimer and Ising models on Klein bottles}
\label{sec:as}

Our main result so far (Theorem~\ref{thm:Zmn}) gives an exact expression for the dimer partition function~$Z_{mn}$ for all~$m,n$ in terms of a finite set of data, namely the characteristic polynomials~$R$ and~$P$. This expression turns out to be well suited for the determination of the asymptotics of~$Z_{mn}$,
which is the subject of this section.

It is organised as follows. In Section~\ref{sub:KSW}, we recall~\cite[Theorem~1]{KSW} as well as the numerous notations required for this statement: this deals with the contribution of~$P$ to~$Z_{mn}$.
Section~\ref{sub:roots} contains some technical statements on the asymptotics of the
product of evaluations of a polynomial at roots of unity, dealing with the contribution of~$R$ to~$Z_{mn}$.
In Section~\ref{sub:as-gen}, we give the general form of the asymptotics of~$Z_{mn}$ for arbitrary weighted graphs in the Klein bottle, assuming a conjecture of~\cite{KSW} on the zeros of the characteristic polynomial~$P$ in the non-bipartite case. In Section~\ref{sub:as-bip}, we give the explicit form of this asymptotics for bipartite graphs. Section~\ref{sub:cons} deals with some consequences of these results for the dimer model. Finally, in Section~\ref{sub:as-Ising}, we compute the asymptotic expansion of the Ising partition function.

\subsection{The work of Kenyon, Sun and Wilson}
\label{sub:KSW}

The aim of this section is to recall a special case of~\cite[Theorem~1]{KSW}, namely an asymptotic expansion of
\[
P_{mn}(\zeta,\xi)=\prod_{z^n=\zeta}\prod_{w^m=\xi}P(z,w)\,,
\]
where~$P$ is an analytic non-negative function defined on the unit torus~$S^1\times S^1$.
We then explain the simpler form of this expansion when~$P(z,w)$ is a characteristic polynomial of
the form studied in our work.

It will be assumed that~$P$ does not vanish except at {\em positive nodes\/}, i.e. elements~$(e^{\pi ir_0},e^{\pi is_0})$ of the unit torus such that
\[
P(e^{\pi i(r_0+r)},e^{\pi i(s_0+s)})=\pi^2(A_zr^2+2Brs+A_ws^2)+O(\|(r,s)\|^3)
\]
with~$A_z,A_w>0$ and~$D:=\sqrt{A_zA_w-B^2}>0$. To such a node, let us associate the parameter
\begin{equation}
\label{equ:tau1}
\tau=\frac{-B+iD}{A_w}
\end{equation}
in the complex upper-half plane.
If~$\{(z_j,w_j)\,|\,1\le j\le \ell\}$ denote the zeros of~$P$ and~$\{\tau_j\,|\,1\le j\le \ell\}$ the associated parameters multiplied by~$\textstyle{\frac{m}{n}}$, then for all~$(\zeta,\xi)\in S^1\times S^1$ that are not zeros of~$P$, Kenyon, Sun and Wilson show that
\begin{equation}
\label{equ:KSW}
\log P_{mn}(\zeta,\xi)=2mn\,\mathbf{f}_0 + \sum_{j=1}^\ell 2\log \Xi\Big(\frac{\zeta}{z_j^n},\frac{\xi}{w_j^m}\Big|\tau_j\Big)+o(1)
\end{equation}
for~$m$ and~$n$ tending to infinity with~$\textstyle{\frac{m}{n}}$ bounded below and above, where
\begin{equation}
\label{equ:f0}
\mathbf{f}_0=\frac{1}{2}\iint_{S^1\times S^1}\log P(z,w)\frac{dz}{2\pi iz}\frac{dw}{2\pi iw}
\end{equation}
and~$\Xi$ is the explicit function defined by
\[
\Xi(-\exp(2\pi i\phi),-\exp(2\pi i\psi)|\tau)=\left|\frac{\vartheta(\phi\tau-\psi|\tau)\exp(\pi i\tau\phi^2)}{\eta(\tau)}\right|\,.
\]
Here,
\[
\vartheta(\nu|\tau)=\sum_{j\in\Z}\exp(\pi i(j^2\tau + 2j\nu))
\]
is the {\em Jacobi theta function\/} and
\[
\eta(\tau)=\exp(\pi i\textstyle{\frac{\tau}{12}})\prod_{j\ge 1}(1-\exp(2\pi ij\tau))
\]
the {\em Dedekind eta function\/}.

We refer to~\cite{KSW} for a more general form of this result, its proof, and for properties of these special functions.
See also~\cite[Section~3.3]{KSW} for an interpretation of~$\tau$ as the shape parameter of the torus in its natural conformal embedding.
Let us recall that there are three other Jacobi theta functions related to~$\vartheta=\vartheta_{00}$ by
\begin{align*}
\vartheta_{01}(\nu|\tau)&=\vartheta(\nu+\textstyle{\frac{1}{2}}|\tau)\\
\vartheta_{10}(\nu|\tau)&=\exp(\pi i(\nu+\textstyle{\frac{\tau}{4}}))\vartheta(\nu+\textstyle{\frac{\tau}{2}}|\tau)\\
\vartheta_{11}(\nu|\tau)&=i\exp(\pi i(\nu+\textstyle{\frac{\tau}{4}}))\vartheta(\nu+\textstyle{\frac{\tau}{2}}+\textstyle{\frac{1}{2}}|\tau)\,.
\end{align*}
For later use, we also recall the equalities
\begin{equation}
\label{equ:eta}
2\eta(\tau)^3=\vartheta_{00}(\tau)\vartheta_{10}(\tau)\vartheta_{01}(\tau)\,,
\end{equation}
where~$\vartheta_{rs}(\tau)$ stands for~$\vartheta_{rs}(0|\tau)$, and
\begin{equation}
\label{equ:2tau}
\frac{\vartheta_{00}(\tau)\vartheta_{01}(\tau)}{\eta(\tau)^2}=\frac{\vartheta_{01}(2\tau)}{\eta(2\tau)}\,,
\end{equation}
see e.g.~\cite[Lemma~2.5]{KSW}.

\medskip

We will apply Equation~\eqref{equ:KSW} to the characteristic polynomial~$P$,
which by Proposition~\ref{prop:P} satisfies the equality~$P(z,w^{-1})=P(z,w)$.
Furthermore, we know from Proposition~\ref{prop:Q(-1)} that if~$\G$ is bipartite, then any zero~$(z_j,w_j)$ of~$P$ on the unit torus satisfy~$z_j=-1$. We conjecture that this still holds in the non-bipartite case:

\begin{conjecture}
For any graph in the Klein bottle, all the zeroes of the characteristic polynomial~$P$ are positive nodes of the form~$(-1,w_0)$.
\end{conjecture}

This fact is known to hold for Fisher graphs, see Lemma~\ref{lemma:node} below.
Note also that this is coherent with the conjecture of~\cite[Section~1.2]{KSW},
which states that for an arbitrary
non-bipartite toric graph, the associated characteristic polynomial either never vanishes on the unit torus,
or admits zeros that are positive real nodes. 

\begin{figure}[htb]
\labellist\small\hair 2.5pt
\pinlabel {$=$} at 660 200
\endlabellist
\centering
\includegraphics[width=0.6\textwidth]{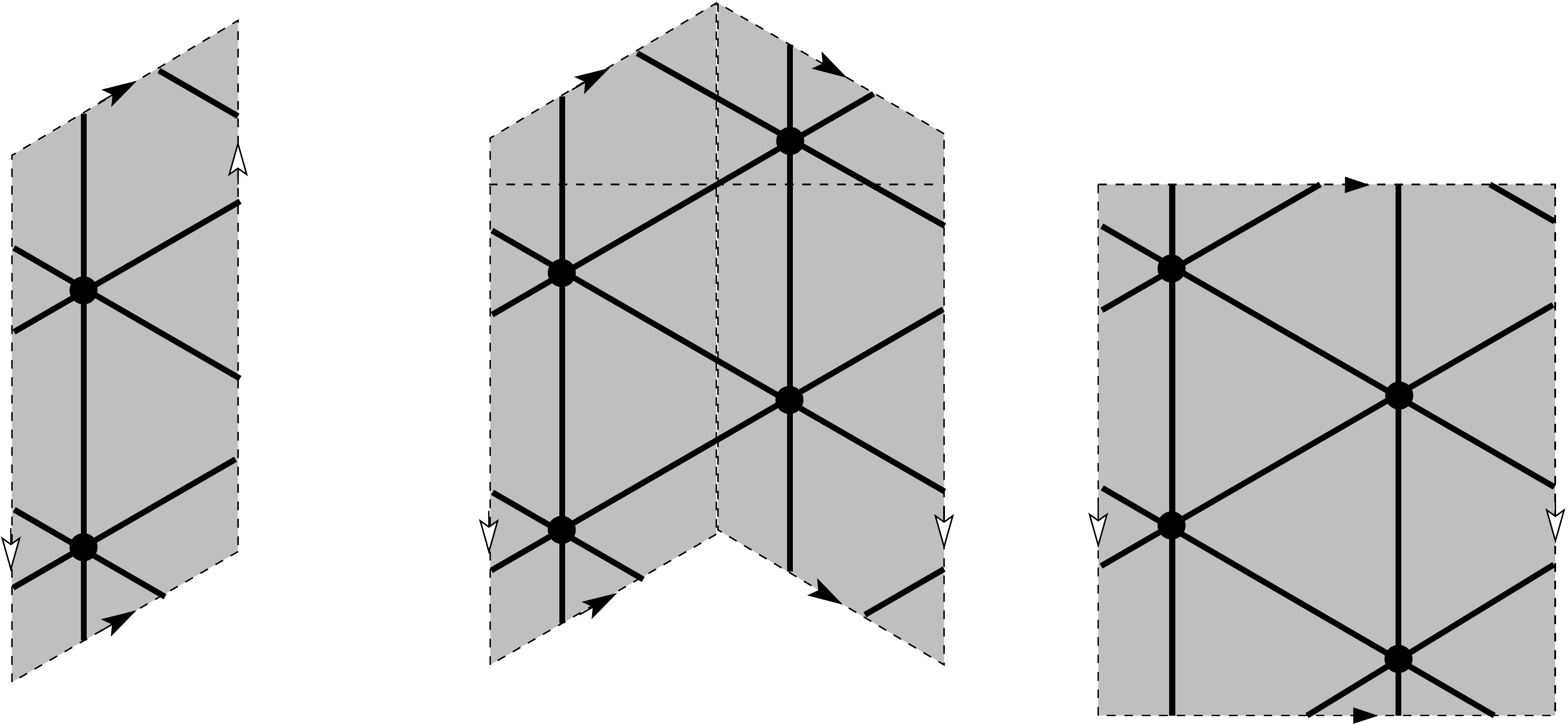}
\caption{A (slanted) Klein bottle together with its torus double cover, which by a cut-and-paste argument is seen to be rectangular. (The triangular lattice is only drawn for concreteness.)
}
\label{fig:tau}
\end{figure}

These two properties of~$P$ easily imply that at any zero of~$P$, the associated
parameter~$\tau$ from Equation~\eqref{equ:tau1} is purely imaginary.
Note that this can also be motivated geometrically, since any torus made up of two copies of a Klein bottle
is rectangular, as illustrated in Figure~\ref{fig:tau}.
(Thanks are due to Andrea Sportiello for this observation.)
More precisely, the parameter~$\tau_j$ associated with the zero~$(-1,w_j)$ of~$P$ is simply given by
\begin{equation}
\label{equ:tau}
\tau_j=i\,\frac{m}{n}\left|\frac{\partial_z^2P(-1,w_j)}{\partial_w^2P(-1,w_j)}\right|^{1/2}\,.
\end{equation}
In the bipartite case, it takes the yet simpler form
\begin{equation}
\label{equ:tau-bip}
\tau_j=i\,\frac{m}{n}\left|\frac{\partial_zQ(-1,w_j)}{\partial_wQ(-1,w_j)}\right|\,.
\end{equation}

Note that we will only need to apply Equation~\eqref{equ:KSW} for~$(\zeta,\xi)=(1,\pm 1)$.
Since~$z_j=-1$ and~$n$ is odd, we will only make use of evaluations at~$(-1,\frac{\xi}{w^m_j})$ of the function~$\Xi$, which are of the form
\begin{equation}
\label{equ:Xi}
\Xi(-1,-\exp(2\pi i\psi)|\tau)=\left|\frac{\vartheta(-\psi|\tau)}{\eta(\tau)}\right|=\frac{\vartheta(\psi|\tau)}{\eta(\tau)}>0\,,
\end{equation}
since~$\eta(\tau)$ and~$\vartheta(\psi|\tau)$ are strictly positive for~$\psi$ real
and~$\tau$ purely imaginary.

\subsection{Product of evaluations of a polynomial at roots of unity}
\label{sub:roots}

For any weighted graph in~$\T^2$, we saw in Proposition~\ref{prop:Ppos} that the associated characteristic polynomial~$P$ is non-negative on the unit torus. Furthermore, it is proved in the bipartite case~\cite{KOS} and conjectured in the general case~\cite{KSW} that all the zeros of~$P$ in the unit torus are positive nodes.
Therefore, in view of Theorem~\ref{thm:Zmn} and Equation~\eqref{equ:KSW} above, we are left with the analysis of the asymptotics of
\[
\mathrm{Arg}\Big(\prod_{z^n=1}R(z,1)\Big)\quad\text{and}\quad\mathrm{Arg}\Big(\prod_{z^n=1}R(z,-1)\Big)
\]
for~$n$ odd tending to infinity. This is the subject of this section.

Let us assume that a given polynomial~$R(z)\in\C[z^{\pm 1}]$ satisfies~$R(iz)\in\R[z^{\pm 1}]$. Then, the argument of its leading coefficient is of the form~$\lambda\frac{\pi}{2}$ for some~$\lambda\in\Z/4\Z$. Furthermore, its roots are either purely imaginary, or paired up as~$\{iz,i\overline{z}\}$. Let us denote by~$p$ the number of such pairs with modulus~$>1$ and by~$r_+$ (resp.~$r_-$) the number of roots in the positive (resp. negative) imaginary axis with modulus~$>1$, counted with multiplicity. Finally, let us assume that the only roots of~$R$ on the unit circle are~$i$ and~$-i$, and write~$m_+$ and~$m_-$ for the respective multiplicities. 

\begin{lemma}
\label{lemma:alpha}
With the notations and assumptions above, we have
\[
\mathrm{Arg}\Big(\prod_{z^n=1}R(z)\Big)=\left((\lambda+2p+r_--r_+)n+(-1)^{\frac{n-1}{2}}\textstyle{\frac{m_--m_+}{2}}\right)\frac{\pi}{2}+o(1)\in \Z/2\pi\Z 
\]
as~$n$ odd tends to infinity.
\end{lemma}
\begin{proof}
If~$\{\mu_k\}_k$ denotes the set of roots of~$R(z)$ and~$m_k$ the multiplicity of~$\mu_k$,
then we have~$R(z)=c\,z^d\,i^\lambda\prod_k(z-\mu_k)^{m_k}$ for some positive real number~$c$ and~$d\in\Z$.
Therefore, the equality~$\prod_{z^n=1}(z-\mu)=1-\mu^n$ leads to
\[
\prod_{z^n=1}R(z)=c^n\Big(\prod_{z^n=1}z\Big)^di^{\lambda n}\prod_k\Big(\prod_{z^n=1}(z-\mu_k)\Big)^{m_k}=c^n\,i^{\lambda n}\prod_k(1-\mu^n_k)^{m_k}\,.
\]
It follows that
\[
\mathrm{Arg}\Big(\prod_{z^n=1}R(z)\Big)=\lambda n \frac{\pi}{2}+\sum_k m_k\,\mathrm{Arg}(1-\mu_k^n)\in\Z/2\pi\Z\,,
\]
and we are left with the analysis of~$\mathrm{Arg}(1-\mu^n)$ as~$n$ tends to infinity for various~$\mu\in\C$.

If a root~$\mu$ has modulus~$|\mu|<1$, then~$\mathrm{Arg}(1-\mu^n)=o(1)$ does not contribute in the limit. On the other hand, if~$\mu=re^{i\varphi}$ with~$r>1$, then
\[
e^{i\mathrm{Arg}(1-\mu^n)}=\frac{1-\mu^n}{|1-\mu^n|}=\frac{1-r^ne^{in\varphi}}{|1-r^ne^{in\varphi}|}=\frac{r^{-n}-e^{in\varphi}}{|r^{-n}-e^{in\varphi}|}=-e^{i(n\varphi+o(1))}\,,
\]
so~$\mathrm{Arg}(1-\mu^n)=(\mathrm{Arg}(\mu)+\pi)n+o(1)$ since~$n$ is odd. Therefore, each pair~$\{iz,i\overline{z}\}$ with~$|z|>1$ contributes~$\mathrm{Arg}(iz)n+\mathrm{Arg}(i\overline{z})n=\pi n$, each root~$\mu\in i\R_{>0}$ with~$|\mu|>1$ contributes~$-\frac{\pi}{2}n$ and each root~$\mu\in i\R_{<0}$ with~$|\mu|>1$ contributes~$\frac{\pi}{2}n$. We end up with the total contribution of the roots of modulus~$>1$ equal to~$(2p-r_++r_-)n\frac{\pi}{2}$, as expected.
Finally, since~$n$ is odd, each root~$\mu=i$ (resp.~$\mu=-i$) contributes~$\mathrm{Arg}(1-i^n)=(-1)^{\frac{n+1}{2}}\frac{\pi}{4}$ (resp.~$\mathrm{Arg}(1+i^n)=(-1)^{\frac{n-1}{2}}\frac{\pi}{4}$). This concludes the proof.
\end{proof}

We need one last lemma.

\begin{lemma}
\label{lemma:cst}
Let~$p_t(z)\in\R[z^{\pm1}]$ be a~$1$-parameter family of non-zero Laurent polynomials having
no roots in the unit circle, with the coefficients of~$p_t$ given by continuous functions of~$t$.
Consider the associated integer
\[
A_t:=(s-1-d)+2p+r_--r_+\,,
\]
where~$d$ is the top-degree of~$p_t$,~$s\in\{\pm 1\}$ denotes the sign of its leading coefficient, and~$r_+$ (resp.~$r_-$,~$p$) the number of real roots~$z>1$ of~$p_t$
(resp. real roots~$z<-1$, resp. the number of pairs of conjugate roots with modulus~$>1$), counted
with multiplicities.
Then, the congruence class of~$A_t$ modulo~$4$ does not depend on~$t$.
\end{lemma}
\begin{proof}
Let us first assume that the leading coefficient of~$p_t$ does not vanish. Since the roots of~$p_t$ depend continuously on~$t$ and are not allowed to cross the unit circle, the congruence class of~$A_t$ modulo~$4$
is indeed constant: this is trivial unless
a pair of conjugate roots merges into a double real root~$z>1$, in which case~$A_t$ changes by~$-4$.

Let us now assume that the leading coefficient~$c_d$ of~$p_t$ vanishes at~$t=t'$,
but the next coefficient~$c_{d-1}$ does not vanish at~$t=t'$.
There are~$4$ cases to be considered, depending on the possible values of
the signs~$s$ of~$c_d$ for~$t\in(t'-\varepsilon,t')$ and~$s'$ of~$c_{d-1}$ at~$t'$.
Let us first assume that~$(s,s')=(1,1)$. Then, as~$t$ tends to~$t'$, the largest root of~$p_t$ tends
to~$-\infty$ along the real axis. Hence, both integers~$r_-$ and~$d$ drop by~$1$ while all the other
integers stay constant, leading to~$A_t=A_{t'}$. The case of~$(s,s')=(-1,-1)$ is identical.
For~$(s,s')=(1,-1)$, the largest root of~$p_t$ tends to~$+\infty$ along the real axis, leading to~$r_+$ and~$d$ dropping by~$1$, which is compensated by~$s$ dropping by~$2$.
The final case is~$(s,s')=(-1,-1)$, where the largest root of~$p_t$ tends to~$-\infty$ along the real axis, leading to~$r_+$ and~$d$ dropping by~$1$ and~$s$ dropping by~$2$. As a result, we have~$A_{t'}=A_t+4$, and
the residue modulo~$4$ is constant indeed.

In general, it might well happen that the coefficients~$c_d,c_{d-1},\dots,c_{d-\ell+1}$ of~$p_t$
simultaneously vanish at~$t=t'$ for some~$\ell\ge 1$, with~$c_{d-\ell}\neq 0$ since~$p_t$ is never
identically zero. However, by a small perturbation of the coefficients, it can be assumed that~$c_d$ vanishes
first, followed by~$c_{d-1}$ at a later time, then by~$c_{d-2}$, and so one. Therefore,~$\ell$ successive applications of the case studied above leads to the proof of the general case. 
\end{proof}

\subsection{Dimer asymptotics in the general case}
\label{sub:as-gen}

We are now ready to state and prove the main result of this section in its most general form,
valid for arbitrary (possibly non-bipartite) graphs in the Klein bottle. The bipartite case is the topic of the next section.

Let~$\G\subset\K$ be a weighted graph embedded in the Klein bottle, and let~$R(z,\pm 1)$ denote the
associated Kleinian polynomials.
Let us write
\begin{equation}
\label{equ:A}
A:=\lambda+2p+r_--r_++\textstyle{\frac{m_--m_+}{2}}\,,
\end{equation}
where the integers~$\lambda,p,r_\pm$ and~$m_\pm$ are associated to the polynomial~$R(z,1)$ as in Lemma~\ref{lemma:alpha}, and let us denote by~$A'$ the
corresponding quantity for~$R(z,-1)$.
We first assume that~$\G\subset\K$ is not ``too small'' in the sense of Remark~\ref{rem:C},
and deal with the ``too small'' case in Remark~\ref{rem:as-gen}~(iii).

\begin{theorem}
\label{thm:as-gen}
Let~$\G\subset\K$ be a weighted graph embedded in the Klein bottle, and let us assume that all the zeros in the unit torus of the associated characteristic polynomial~$P$ are positive nodes of the form~$\{(-1,\exp(2\pi i\psi_j))\,|\,1\le j\le \ell\}$. Then, we have the asymptotic expansion
\[
\log Z_{mn}=mn\frac{\mathbf{f}_0}{2}+\mathsf{fsc}+o(1)
\]
for~$m$ and~$n$ tending to infinity with~$n$ odd and~$m/n$ bounded below and above, where~$\mathbf{f}_0$ is as in~\eqref{equ:f0} and~$\mathsf{fsc}=\log\mathsf{FSC}$ with
\[
\mathsf{FSC}=\left|\sin(A{\textstyle{\frac{\pi}{4}}})\right|\prod_{j=1}^\ell\left(\frac{\vartheta_{01}(m\psi_j|\tau_j)}{\eta(\tau_j)}\right)^{\frac{1}{2}}+\left|\cos(A'{\textstyle{\frac{\pi}{4}}})\right|\prod_{j=1}^\ell\left(\frac{\vartheta_{00}(m\psi_j|\tau_j)}{\eta(\tau_j)}\right)^{\frac{1}{2}}
\]
if~$m$ is odd, and
\[
\mathsf{FSC}=\left|\sin((A-A'){\textstyle{\frac{\pi}{4}}})\right|\prod_{j=1}^\ell\left(\frac{\vartheta_{01}(m\psi_j|\tau_j)}{\eta(\tau_j)}\right)^{\frac{1}{2}}+\prod_{j=1}^\ell\left(\frac{\vartheta_{00}(m\psi_j|\tau_j)}{\eta(\tau_j)}\right)^{\frac{1}{2}}
\]
if~$m$ is even.
Here, the parameter~$\tau_j$ is given
by~$i\frac{m}{n}\left|\frac{\partial_z^2P(-1,w_j)}{\partial_w^2P(-1,w_j)}\right|^{1/2}$,
while~$A$ and~$A'$ are the modulo~$4$ integers determined by~$R(z,1)$ and~$R(z-1)$ as in~\eqref{equ:A}.
Finally, if the dimer weights vary continuously so that~$P(-1,1)\neq 0$ (resp,~$P(-1,-1)\neq 0$), then the modulo~$4$ integer~$A$ (resp.~$A'$) stays constant.
\end{theorem}

\begin{proof}
Let us apply Theorem~\ref{thm:Zmn} to~$\G\subset\K$, and the work of Kenyon, Sun and Wilson in the form of Equation~\eqref{equ:KSW} to~$P(\zeta,\xi)=P(1,\pm 1)$, together with Equations~\eqref{equ:tau} and~\eqref{equ:Xi}. This yields the expected
asymptotic expansion of~$\log Z_{mn}$, with angles~$\alpha_n,\alpha'_n$ to be determined.
 
By Remark~\ref{rem:R}~(i), the polynomial~$R(iz,1)$ belong to~$\R[z^{\pm 1}]$.
Furthermore, the equality~$P(z^2,1)=R(z,1)R(-z,1)$ of Proposition~\ref{prop:P}
shows that the roots of~$R(z,1)$ in~$S^1$ correspond to roots of~$P$ in the unit torus,
i.e. satisfy~$z=\pm i$ by hypothesis. Since such a root is a node of~$P$,
the roots~$i$ and~$-i$ of~$R(z,1)$ have total multiplicity~$m_++m_-=2$, so~$\frac{m_+-m_-}{2}$ is an integer.
Therefore, Lemma~\ref{lemma:alpha} can be applied to~$R(z,1)$, leading to
\[
\alpha_n/2=\mathrm{Arg}\Big(\prod_{z^n=1}R(z,1)\Big)/2=\left((\lambda+2p+r_--r_+)n+(-1)^{\frac{n-1}{2}}\textstyle{\frac{m_--m_+}{2}}\right)\frac{\pi}{4}+o(1)\in \Z/\pi\Z\,,
\]
and similarly for~$R(z,-1)$. Furthermore, since~$n$ is odd, we have the modulo~$4$ equality
\[
(\lambda+2p+r_--r_+)n+(-1)^{\frac{n-1}{2}}\textstyle{\frac{m_--m_+}{2}}\equiv\left(\lambda+2p+r_--r_++\textstyle{\frac{m_--m_+}{2}}\right)n=A\cdot n\in\Z/4\Z\,,
\]
and similarly for~$A'$. The result follows from the observation that~$|\sin(An{\textstyle{\frac{\pi}{4}}})|$ and~$|\cos(An{\textstyle{\frac{\pi}{4}}})|$ are independent of~$n$ odd.

Finally, let us consider the~$1$-parameter family of polynomials~$p_t(z)=R(iz,1)$ given by a
continuous path in the dimer weights of~$\G$, assuming that~$P(-1,1)=R(i,1)R(-i,1)$ never vanishes.
In such as case, we have~$m_+=m_-=0$, and the modulo~$4$ integer~$\lambda$ associated with~$R(z,1)$
translates to~$s-1-d$ for~$p_t(z)$, leading to the identification of~$A$ from Equation~\eqref{equ:A}
with~$A_t$ from Lemma~\ref{lemma:cst}. The coefficients of~$R(iz,1)$ being continuous functions of
the dimer weights, the last sentence of the theorem now follows from Lemma~\ref{lemma:cst}.
The proof for~$A'$ is identical.
\end{proof}

\begin{remarks}
\label{rem:as-gen}
\begin{enumerate}[(i)]
\item{The only contribution of the characteristic polynomial~$R$ to~$\mathsf{FSC}$ is in
the coefficients~$|\sin(A{\textstyle{\frac{\pi}{4}}})|$,~$|\cos(A{\textstyle{\frac{\pi}{4}}})|$
and~$|\sin((A-A'){\textstyle{\frac{\pi}{4}}})|$, which take the values
in~$\{0,\textstyle{\frac{\sqrt{2}}{2}},1\}$.}
\item{It is conjectured in the non-bipartite case and proved in the bipartite case and for Fisher graphs
that~$P$ admits at most two zeros on the unit torus, that are positive nodes of the form~$(-1,w_0)$.
(This follows from~\cite{KOS} together with Proposition~\ref{prop:Q(-1)} in the bipartite case,
and from Lemma~\ref{lemma:node} for Fisher graphs.)
Together with the previous remark, this implies that there are a
finite number of possible finite size corrections. In the bipartite case and for Fisher graphs,
this statement can be made much more precise, see Theorems~\ref{thm:as-bip} and~\ref{thm:as-Ising} below.}
\item{As mentioned in Remarks~\ref{rem:C},~\ref{rem:Kmn} and~\ref{rem:perK},
some graphs are ``too small'', and Theorem~\ref{thm:Zmn} needs some minor adjustment.
In Theorem~\ref{thm:as-gen} above, by Equation~\eqref{equ:small}, this simply
amounts to exchanging the roles of~$A$ and~$A'$.}
\end{enumerate}
\end{remarks}

We conclude this section with two explicit non-bipartite examples.

\begin{example}
\label{ex:fsc-sq}
Consider the~$(1\times 2)$-square lattice of Figure~\ref{fig:ex2-4} with horizontal weights~$x_1=x_2=:x$ and vertical weights~$y_1=y_2=:y$. As computed in Example~\ref{ex:R}, we have
\[
R(z,1)=R(z,-1)=ix^2(z+z^{-1})\,.
\]
This polynomial has leading coefficient of argument~$\frac{\pi}{2}$ and roots~$\pm i$ of multiplicity~$1$; with the notations of Lemma~\ref{lemma:alpha}, this gives~$\lambda=1$,~$p=r_+=r_-=0$,~$m_+=m_-=1$ and leads to~$A=A'=1$.
Furthermore, as computed in Example~\ref{ex:P}, we have
\[
P(z,w)=y^4(w-w^{-1})^4-4x^2y^2(w-w^{-1})^2+x^4(2+z+z^{-1})\,.
\]
This polynomial has two roots in~$S^1\times S^1$, namely~$(z_1,w_1)=(-1,-1)$ and~$(z_2,w_2)=(-1,1)$, both of which are positive nodes with associated parameters~$\tau_1=\tau_2=i\frac{m}{n}\frac{x}{4y}$.
Applying Theorem~\ref{thm:as-gen} to this example (with~$M=m$ and~$N=2n$) together with
Equation~\eqref{equ:2tau}, we get the following result.
(Note that Remark~\ref{rem:as-gen}~(iii) needs to be applied, but has no effect on the final formula.)

\begin{corollary}
\label{cor:square}
For~$N\equiv 2\pmod{4}$, the finite size corrections in the asymptotic expansion of the dimer partition function for the~$(M\times N)$-square lattice in the Klein bottle are given by~$\mathsf{FSC}=\left(\frac{2\vartheta_{01}(2\tau)}{\eta(2\tau)}\right)^{1/2}$ if~$M$ is odd, and by~$\mathsf{FSC}=\frac{\vartheta_{00}(\tau)}{\eta(\tau)}$ if~$M$ is even, where~$\tau=i\frac{Mx}{2Ny}$.\qed
\end{corollary}

Note that this result can be extended to any even~$N$ by considering the more general fundamental domain given by the~$(1\times 2^k)$-square lattice with~$k\ge 1$. Its Kleinian characteristic polynomial is given
by~$R(z,w)=i x^{2^k}(z+z^{-1})$. Its toric one, which can easily be computed inductively, has the
two roots~$(-1,\pm 1)$ in the unit torus, and
satisfies~$\partial_z^2P(-1,\pm 1)=2x^{2^{k+1}}$ and~$\partial_w^2P(-1,\pm 1)=2^{2k+3}x^{2^{k+1}-2}y^2$,
leading to the parameters~$\tau_1=\tau_2=i\frac{m}{n}\frac{x}{2^{k+1}y}$. Setting~$M=m$ and~$N=2^kn$ thus
leads to the same result as above, now valid for any even~$N$.

Using Equality~\eqref{equ:eta} and the other standard notation~$\vartheta_2:=\vartheta_{10}$,
~$\vartheta_3:=\vartheta_{00}$,~$\vartheta_4:=\vartheta_{01}$, the logarithm of these finite size corrections
can be written as
\[
\mathsf{fsc}=\frac{1}{2}\log 2 + \frac{1}{6}\log\frac{2\vartheta_4(2\tau)^2}{\vartheta_2(2\tau)\vartheta_3(2\tau)}\]
for~$M$ odd, and
\[
\mathsf{fsc}=\frac{1}{3}\log\frac{2\vartheta_3(\tau)^2}{\vartheta_2(\tau)\vartheta_4(\tau)}
\]
for~$M$ even. This coincides with the formulas~$(80),(81)$ and~$(76),(77)$ of~\cite{IOH}, respectively.
(For~$M$ even, it does not coincide with~\cite[Equations~(35),(41)]{Lu-Wu99}.)
\end{example}

\begin{example}
\label{ex:tri-sq}
Consider the isotropic triangular lattice of Figure~\ref{fig:ex2-4}. By Example~\ref{ex:R}, we have
\[
R(z,1)=2(z^2+z^{-2})+4i(z+z^{-1})-4\quad\text{and}\quad R(z,-1)=4\,.
\]
The first polynomial has two roots of modulus~$>1$, namely~$\frac{1}{2}(\sqrt{3}+1)(\pm 1-i)$, and two roots of modulus~$<1$, namely~$\frac{1}{2}(\sqrt{3}-1)(\pm 1+i)$. With the notations of Lemma~\ref{lemma:alpha}, this gives~$\lambda=0$,~$p=1$,~$r_+=r_-=m_+=m_-=0$ and leads to~$A=2$.
On the other hand, we obviously have~$A'=0$.
Furthermore, as computed in Example~\ref{ex:P}, we have
\[
P(z,w)=(z^2+z^{-2})(w+w^{-1}+2)+10(w+w^{-1})+w^2+w^{-2}+34\,,
\]
which is strictly positive on the unit torus.
Applying Theorem~\ref{thm:as-gen} to this example, we get the asymptotic expansion
\[
\log Z_{mn}=mn\frac{\mathbf{f}_0}{2}+\log(2)+o(1)\,.
\]
Hence, we see that the finite-size correction term might be non-zero even when the spectral curve does not meet the unit torus. As we shall see in Theorems~\ref{thm:as-bip} and~\ref{thm:as-Ising}, this never occurs
in the bipartite case, but it does occur for Fisher graphs.

Note that this phenomemon is not specific to the Klein bottle, as it does also appear in the toric (non-bipartite) case. The first such example was computed in~\cite{BEP}, thus proving the first sentence of~\cite[Theorem~2]{KSW} to be incorrect, see also Remark~\ref{rem:Ising}~(i) below.
\end{example}

\subsection{Dimer asymptotics in the bipartite case}
\label{sub:as-bip}

Let us now consider a bipartite graph~$\G$ embedded in~$\K$. In this case, we know from the work of Kenyon,
Okounkov and Sheffield~\cite{KOS} together with Proposition~\ref{prop:Q(-1)}
that the associated toric characteristic polynomial~$Q(z,w)$ has at most two (conjugate)
zeros~$(-1,w_0)$ and~$(-1,\overline{w}_0)$ in the unit torus~$S^1\times S^1$,
that might coincide to form a single real node.
This leads to three cases in the asymptotic expansion of the associated dimer partition function,
as already stated in Theorem~\ref{thmintro:as-bip}. We now recall this result for the reader's convenience,
and give the proof.

\begin{theorem}
\label{thm:as-bip}
Let~$\G\subset\K$ be a weighted bipartite graph embedded in the Klein bottle. Then, we have the asymptotic expansion
\[
\log Z_{mn}=mn\frac{\mathbf{f}_0}{2}+\mathsf{fsc}+o(1)
\]
for~$m$ and~$n$ tending to infinity with~$n$ odd and~$m/n$ bounded below and above, with
\[
\mathbf{f}_0=\iint_{S^1\times S^1}\log |Q(z,w)|\frac{dz}{2\pi iz}\frac{dw}{2\pi iw}
\]
and~$\mathsf{fsc}=\log\mathsf{FSC}$ given as follows:
\begin{enumerate}
\item{If~$Q(z,w)$ has no zeros in the unit torus, then~$\mathsf{FSC}=1$.}
\item{If~$Q(z,w)$ has two zeros~$(-1,w_0)\neq(-1,\overline{w}_0)$ in the unit torus
with~$w_0=\exp(2\pi i\psi)$, then
\[
\mathsf{FSC}=\frac{\vartheta_{00}(m\psi|\tau)}{\eta(\tau)}+\frac{\vartheta_{01}(m\psi|\tau)}{\eta(\tau)}
\quad\text{with}\quad
\tau=i\,\frac{m}{n}\left|\frac{\partial_zQ(-1,w_0)}{\partial_wQ(-1,w_0)}\right|\,.
\]}
\item{If~$Q(z,w)$ has a single (real) node at~$(-1,w_0)$ in the unit torus, then
\[
\mathsf{FSC}=\frac{\vartheta_{00}(\tau)}{\eta(\tau)}+\frac{\vartheta_{01}(\tau)}{\eta(\tau)}
\quad\text{with}\quad
\tau=i\,\frac{m}{n}\left|\frac{\partial^2_zQ(-1,w_0)}{\partial^2_wQ(-1,w_0)}\right|^{1/2}\,.
\]}
\end{enumerate}
\end{theorem}

\begin{proof}
Corollary~\ref{cor:Zmn} expresses~$Z_{mn}$ in terms of~$|Q_{mn}(1,\pm 1)|^{1/2}$ and
of~$\mathrm{Arg}\left(\prod_{z^n=1}S(z,\pm 1)\right)$, two contributions whose asymptotics we now analyse.

First, let us apply Equation~\eqref{equ:KSW} to~$P(z,w)=Q(z,w)Q(z^{-1},w^{-1})$. This leads to
\[
\frac{1}{2}\log\left|Q_{mn}(1,\xi)\right|=\frac{1}{4}\log P_{mn}(1,\xi)=mn\,\frac{\mathbf{f}_0}{2} + \mathsf{fsc}(\xi) + o(1)
\]
for~$\xi=\pm 1$, with
\[
\mathbf{f}_0=\frac{1}{2}\iint_{S^1\times S^1}\log P(z,w)\frac{dz}{2\pi iz}\frac{dw}{2\pi iw}=\iint_{S^1\times S^1}\log |Q(z,w)|\frac{dz}{2\pi iz}\frac{dw}{2\pi iw}\,,
\]
and~$\mathsf{fsc}(\xi)=\log\mathsf{FSC}(\xi)$ given as follows.
\begin{enumerate}[(1)]
\item{If~$Q$ has no zeros in~$S^1\times S^1$, then~$\mathsf{FSC}(\xi)=1$.}
\item{If~$Q(z,w)$ has two distinct zeros~$(-1,w_0)\neq(-1,\overline{w}_0)$ in~$S^1\times S^1$, then
\[
\mathsf{FSC}(\xi)=\Xi((-1)^n,\xi\overline{w}_0^m|\tau)^{1/2}\,\Xi((-1)^n,\xi w_0^m|\tau)^{1/2}=\Xi(-1,\xi w_0^m|\tau)\,,
\]
since~$n$ is odd and~$\vartheta$ is an even function.}
\item{If~$Q(z,w)$ has a single real node at~$(-1,w_0)$ in~$S^1\times S^1$, then
\[
\mathsf{FSC}(\xi)=\Xi(-1,\xi w_0^m|\tau)\,.
\]}
\end{enumerate}

We now apply Lemma~\ref{lemma:alpha}, whose notation we assume,
to the polynomial~$S(z,1)\in\C[z^{\pm 1}]$.
By Equation~\eqref{equ:S}, we know that~$S(iz,1)$ belongs to~$\R[z^{\pm 1}]$.
By Proposition~\ref{prop:rootS}, all the roots of~$S(z,1)$ are simple and purely imaginary,
so Lemma~\ref{lemma:alpha} can indeed be applied, with~$p=0$. 
By Proposition~\ref{prop:rootS} and Equation~\eqref{equ:QS}, the roots~$i$ and~$-i$ of~$S(z,1)$ have multiplicities~$m_+=m_-$. With these observations, Lemma~\ref{lemma:alpha} now reads
\[
\beta_n=\mathrm{Arg}\Big(\prod_{z^n=1}S(z,1)\Big)=An\frac{\pi}{2}+o(1)\,,
\]
with~$A=\lambda+r_--r_+$. By Corollary~\ref{cor:Zmn}, we are only concerned with the parity of~$An$.
Since~$n$ is odd, this parity is simply given by~$A=\lambda+r$, where~$r=r_++r_-$ denotes
the number of roots of~$S(z,1)$ with modulus~$>1$.
The same argument obviously holds for~$S(z,-1)$, and we denote by~$A'=\lambda'+r'$ the
corresponding integers.

To summarize, Corollary~\ref{cor:Zmn}, Equation~\eqref{equ:KSW} and Lemma~\ref{lemma:alpha} lead to the statement of
the theorem, with finite size corrections given by
\[
\mathsf{FSC}=
\begin{cases}
|\sin(A\frac{\pi}{2})|\,\mathsf{FSC}(1)+|\cos(A'\frac{\pi}{2})|\,\mathsf{FSC}(-1) & \text{for $m$ odd;}\\
|\sin((A+A')\frac{\pi}{2})|\,\mathsf{FSC}(1)+\mathsf{FSC}(-1) & \text{for $m$ even,}
\end{cases}
\]
where~$A,A'$ are as in Lemma~\ref{lemma:S} and~$\mathsf{FSC}(\pm 1)$ as described above, depending
on the cases~(1)-(3). The statement now follows from Lemma~\ref{lemma:S} together with Equations~\eqref{equ:tau}-\eqref{equ:Xi} and the relation~$\vartheta_{00}(\nu+\frac{1}{2}|\tau)=\vartheta_{01}(\nu|\tau)$.
\end{proof}

We conclude this section with two examples realizing the three cases in the statement of this theorem.

\begin{example}
\label{ex:sqbip}
Consider the~$2\times 1$ square lattice of Figure~\ref{fig:ex1} with~$x_1=x_2=:x$ and~$y_1=y_2=:y$. Its toric characteristic polynomial~$Q(z,w)=x^2(z+z^{-1}+2)+y^2(w+w^{-1}+2)$
has a single real node at~$(-1,-1)$, and the corresponding Kleinian polynomial~$S(z,-1)=x(z+z^{-1})$ has 
simple roots at~$\pm i$. We are therefore in case~(3),
with~$\tau=i\,\frac{m}{n}\left|\frac{\partial^2_zQ(-1,-1)}{\partial^2_wQ(-1,-1)}\right|^{1/2}=i\frac{mx}{ny}$. Applying Theorem~\ref{thm:as-bip} with~$M=2m$ and~$N=n$, we get the following result.

\begin{corollary}
\label{cor:sqbip}
For~$M$ even and~$N$ odd, the finite size correction in the asymptotic expansion of the dimer partition function for the~$(M\times N)$-square lattice in the Klein bottle is given by~$\mathsf{FSC}=\frac{\vartheta_{00}(\tau)}{\eta(\tau)}+\frac{\vartheta_{01}(\tau)}{\eta(\tau)}$, where~$\tau=i\frac{Mx}{2Ny}$.\qed
\end{corollary}

Note that this result does not coincide with the finite size corrections obtained in Equations~(82),(83) of~\cite{IOH}, and for a good reason: as mentioned in Example~\ref{ex:Z}, these computations are based
on the incorrect formula~\cite[Equation~(5)]{Lu-Wu} for the dimer partition function.
\end{example}

\begin{example}
\label{ex:hexa}
Consider the hexagonal lattice of Figure~\ref{fig:ex2-4}, whose characteristic polynomial is given
by~$Q(z,w)=\nu_1^2+\nu_3^2+\nu_1\nu_3(w+w^{-1})+\nu_2^2z$. If the edge weights are so that~$\nu_1+\nu_3<\nu_2$ or~$\nu_2<|\nu_1-\nu_3|$, then~$Q$ never vanishes on the unit torus and we are in case~(1),
so~$\mathsf{FSC}=1$.

On the other hand, if the edge weights satisfy~$|\nu_1-\nu_3|<\nu_2<\nu_1+\nu_3$, then we are in case~(2) where
\[
\mathsf{FSC}=\frac{\vartheta_{00}(m\psi|\tau)}{\eta(\tau)}+\frac{\vartheta_{01}(m\psi|\tau)}{\eta(\tau)}
\]
with~$w_0=\exp(2\pi i\psi)$ such that~$\nu_1^2+\nu_3^2+\nu_1\nu_3(w_0+w_0^{-1})=\nu_2^2$ and~$\tau=i\,\frac{m}{n}\frac{\nu_2^2}{\nu_1\nu_3|1-w_0^2|}$.
For example, the isotropic case~$\nu_1=\nu_2=\nu_3$ leads to~$\mathsf{FSC}=\frac{\vartheta_{00}(\frac{m}{3}|\tau)}{\eta(\tau)}+\frac{\vartheta_{01}(\frac{m}{3}|\tau)}{\eta(\tau)}$ with~$\tau=i\,\frac{m}{n}\frac{\sqrt{3}}{3}$.
\end{example}

\subsection{Consequences for the dimer model}
\label{sub:cons}

We now explore some consequences of Theorems~\ref{thm:as-gen} and~\ref{thm:as-bip} for the dimer model.
Analogous results hold for the Ising model as well, see Remark~\ref{rem:Ising} below.

\subsubsection*{Asymptotic for~$\tau_\text{im}\to\infty$}
In~\cite{BCN}, Bl\"ote, Cardy and Nightingale argue that for a conformally invariant model at criticality
on an infinitely long strip, the amplitude of the finite-size corrections to the free energy is linearly related to the central charge~$c$ of the model. We now compare our results to these predictions.
 
As explained in~\cite[Section~3.4]{KSW}, the function~$\Xi$ satisfies
\[
\Xi(-1,-\exp(2\pi i\psi)|\tau)=\frac{\vartheta(\psi|\tau)}{\eta(\tau)}=\exp(\pi\tau_\text{im}/12+o(1))
\]
in the limit~$\tau_\text{im}\to\infty$. As a consequence, Theorem~\ref{thm:as-bip} implies that if~$\G$ is bipartite and the dimer model is in the liquid phase (i.e., if the spectral curve intersects the unit torus),
then the partition function satisfies the asymptotic expansion
\begin{equation}
\label{equ:CFT}
\log Z_{mn}=mn\frac{\mathbf{f}_0}{2}+\frac{\pi\tau_\text{im}}{12}+C+o(1)
\end{equation}
for~$m$,~$n$ and~$\tau_\text{im}$ tending to infinity with~$n$ odd and~$m/n$ bounded below and above,
with~$C=\log(2)$. Such an expansion also holds for the non-bipartite examples considered in Section~\ref{sub:as-gen}, with~$C=\frac{\log(2)}{2}$ (resp.~$C=0$) for the~$M\times N$ square lattice with~$N$ even and~$M$ odd (resp.~$M$ even).

Let us now compare this result to the CFT predictions of~\cite{BCN}, where the authors consider the asymptotic expansion of the free energy per unit length of an infinitely long strip of width~$L$ at criticality. They claim it to be of the form
\[
F=fL+\frac{\Delta}{L}+f^\times+\dots
\]
with~$f$ the bulk free energy per unit area,~$\frac{1}{2}f^\times$ the surface free energy, and~$\Delta$ a universal term explicitely given by
\[
\Delta=\left\{\begin{array}{ll}
-\pi c/6&\text{for periodic boundary conditions},\\
-\pi c/24&\text{for free or fixed boundary conditions}.
\end{array}\right.
\]

Coming back to our setting, let us assume that the strip has antiperiodic horizontal boundary conditions (and periodic vertical ones). The term~$f^\times$ vanishes as the Klein bottle has no boundary, while~$f$ corresponds to~$-\frac{\mathbf{f}_0}{2}$. Writing~$L'$ for the length of the strip, we have
\begin{equation}
\label{equ:CFT'}
\log Z=-L' F=L'L\frac{\mathbf{f}_0}{2}-\Delta\frac{L'}{L}+\dots\,.
\end{equation}
Observe that the shape parameter~$\tau$ of the corresponding torus is given by~$\tau=i\frac{L'}{2L}$. Therefore,
the comparison of expansions~\eqref{equ:CFT} and~\eqref{equ:CFT'} leads
to~$\Delta=-\pi/24$. This is consistent with ``fixed'' boundary
conditions including antiperiodic ones, and with the value~$c=1$ for the central charge of
a conformal field theory describing the bipartite dimer model.

The analogous discussion applied to the Ising model is presented in Remark~\ref{rem:Ising}~(iii) below.

\subsubsection*{Ratios of partition functions}

Let us now consider an arbitrary weighted graph~$\G$ embedded in the Klein bottle~$\K$,
and denote by~$\widetilde{\G}$ its~$2$-cover embedded in the torus~$\mathbb{T}^2$.
As usual, let us write~$\G_{mn}\subset\K$ and~$\widetilde{\G}_{mn}\subset\mathbb{T}^2$ for the relevant covers, for~$m,n$ integers with~$n$ odd. Finally, let us denote by~$Z(\G_{mn})$ and~$Z(\widetilde{\G}_{mn})$ the
corresponding dimer partition functions.

By the results of Section~\ref{sec:as}, we have
\[
\log Z(\G_{mn})=mn\frac{\mathbf{f}_0}{2}+\mathsf{fsc}_\K(\tau)+o(1)\,,
\]
where~$\mathbf{f}_0$ can be computed via~\eqref{equ:f0} and
$\mathsf{fsc}_\K$ is a function of an explicit parameter~$\tau$, a function which
falls within a finite number of classes
(recall Remark~\ref{rem:as-gen}~(ii)).
Also, by~\cite[Theorem~2]{KSW}, we have
\[
\log Z(\widetilde{\G}_{mn})=mn\,\mathbf{f}_0+\mathsf{fsc}_\mathbb{T}(\tau)+o(1)\,,
\]
with~$\mathbf{f}_0$ as above, and~$\mathsf{fsc}_\mathbb{T}$ a function of \emph{the same}
parameter~$\tau$ which falls within \emph{the same} classes. As a immediate consequence,
the limit
\[
\lim_{m,n\to\infty}\frac{Z(\G_{mn})^2}{Z(\widetilde{\G}_{mn})}=\frac{\mathsf{fsc}_\K(\tau)^2}{\mathsf{fsc}_\mathbb{T}(\tau)}
\]
is given by some explicit function of~$\tau$ which only depends on the relevant class.

The bipartite case is completely described by Theorem~\ref{thm:as-bip} and~\cite[Theorem~2(b,c)]{KSW}, yielding the following result.
\begin{enumerate}
\item{If~$Q(z,w)$ has no zeros in the unit torus, then~$\lim_{m,n\to\infty}\frac{Z(\G_{mn})^2}{Z(\widetilde{\G}_{mn})}=1$.}
\item{If~$Q(z,w)$ has two zeros~$(-1,w_0)\neq(-1,\overline{w}_0)$ in the unit torus
with~$w_0=\exp(2\pi i\psi)$, then
\[
\lim_{m,n\to\infty}\frac{Z(\G_{mn})^2}{Z(\widetilde{\G}_{mn})}=\frac{2(\vartheta_{00}(m\psi|\tau)+\vartheta_{01}(m\psi|\tau))^2}{\vartheta_{00}(m\psi|\tau)^2+\vartheta_{01}(m\psi|\tau)^2+\vartheta_{10}(m\psi|\tau)^2+\vartheta_{11}(m\psi|\tau)^2}\,,
\]
where~$\tau=i\frac{m}{n}\left|\frac{\partial_zQ(-1,w_0)}{\partial_wQ(-1,w_0)}\right|$.}
\item{If~$Q(z,w)$ has a single (real) node at~$(-1,w_0)$ in the unit torus, then
\[
\lim_{m,n\to\infty}\frac{Z(\G_{mn})^2}{Z(\widetilde{\G}_{mn})}=\frac{2(\vartheta_{00}(\tau)+\vartheta_{01}(\tau))^2}{\vartheta_{00}(\tau)^2+\vartheta_{01}(\tau)^2+\vartheta_{10}(\tau)^2}\,,
\quad\text{with}\quad\tau=i\,\frac{m}{n}\left|\frac{\partial^2_zQ(-1,w_0)}{\partial^2_wQ(-1,w_0)}\right|^{1/2}\,.
\]
}
\end{enumerate}

The non-bipartite case yields additional possible limits. For instance, Example~\ref{ex:fsc-sq} and~\cite[Theorem~2(d)]{KSW} yield the following result. If~$\G_{MN}$ is the~$M\times N$-square lattice with~$N$ and~$M$ even, then
\[
\lim_{M,N\to\infty}\frac{Z(\G_{MN})^2}{Z(\widetilde{\G}_{MN})}=
\frac{2\vartheta_{00}(\tau)^2}{\vartheta_{00}(\tau)^2+\vartheta_{01}(\tau)^2+\vartheta_{10}(\tau)^2}
\]
with~$\tau=i\frac{Mx}{2Ny}$, while for~$N$ even and~$M$ odd, we simply get
\[
\lim_{M,N\to\infty}\frac{Z(\G_{MN})^2}{Z(\widetilde{\G}_{MN})}=2\,.
\]
This latter result can be seen as a reality check for our computations,
since the equality~$Z(\G_{MN})^2=2\,Z(\widetilde{\G}_{MN})$ actually holds
for the~$M\times N$-square lattice with~$N$ even and~$M$ odd.
We refer to~\cite[Equation~(66)]{IOH} for the first occurrence of this non-trivial result, and
to~\cite[Theorem~1.6]{Cim-P} for an extension to a wider class of graphs.

\subsection{Asymptotics of the Ising partition function}
\label{sub:as-Ising}

We now apply our results and methods to the study of the Ising model on a graph~$G\subset\K$
as in Section~\ref{sub:intro-Ising}, whose notation we assume.

Let us first note that, with the help of the well-known Equation~\eqref{equ:Ising-dimer},
it is a trivial task to translate Theorem~\ref{thm:Zmn}
from the dimer to the Ising model, thus providing a closed formula for the Ising partition function
on the cover~$G_{mn}$ of an arbitrary weighted graph~$G$ embedded in the Klein bottle.
We shall not state this result explicitely, but directly move to the study of the resulting asymptotic
expansion, in the form of Theorem~\ref{thmintro:Ising} whose statement we now recall.

\begin{theorem}
\label{thm:as-Ising}
Let~$(G,J)$ be a non-degenerate weighted graph embedded in the Klein bottle, and let~$P(z,w)$ be the
characteristic polynomial of the associated Fisher graph~$\widetilde{G^F}\subset\mathbb{T}^2$.
Then, the Ising partition function on~$G_{mn}$ satisfies
\[
\log Z^J_\beta(G_{mn})= mn\frac{\mathbf{f}_0}{2} + \mathsf{fsc}+o(1)
\]
for~$m$ and~$n$ tending to infinity with~$n$ odd and~$m/n$ bounded below and above, with
\[
\mathbf{f}_0=2\sum_{e\in E(G)}\log\cosh(\beta J_e)+\frac{1}{2} \int_{\mathbb T^2} \log P(z,w)\frac{dz}{2\pi iz}\frac{dw}{2\pi iw}
\]
and~$\mathsf{fsc}=0$ in the subcritical regime~$\beta<\beta_c$,~$\mathsf{fsc}=\log(2)$ in the supercritical regime~$\beta>\beta_c$, and
\[
\mathsf{fsc}=\log\left(\left(\frac{\vartheta_{00}(\tau)}{\eta(\tau)}\right)^{1/2}+\left(\frac{\vartheta_{01}(\tau)}{2\eta(\tau)}\right)^{1/2}\right)\,,\quad\text{where}\quad\tau=i\,\frac{m}{n}\left|\frac{\partial^2_zP(-1,1)}{\partial^2_wP(-1,1)}\right|^{1/2}
\]
in the critical regime~$\beta=\beta_c$.
\end{theorem}

We break down the proof into a series of lemmas. The first one is a direct consequence of~\cite{CDC} (see also~\cite{Li}) and of our conventions for Kasteleyn orientations.

\begin{lemma}
\label{lemma:node}
Let~$(G,J)$ be a non-degenerate weighted graph embedded in the Klein bottle, and let~$P(z,w)$ be the
characteristic polynomial of the associated Fisher graph~$\widetilde{G^F}\subset\mathbb{T}^2$.
For~$\beta\neq\beta_c$, this polynomial never vanishes on the unit torus, while for~$\beta=\beta_c$, it
only vanishes at the positive node~$(z_0,w_0)=(-1,1)$.
\end{lemma}
\begin{proof}
Up to a non-vanishing multiplicative constant, the polynomial~$P(z,w)$
coincides with the characteristic polynomial
of~$\widetilde{G}\subset\mathbb{T}^2$ obtained using the Kac-Ward matrix (see~\cite[Section~4.3]{Cim-KW} or~\cite[Section~3.1]{CCK} for details). As showed in~\cite[Theorem~3.1]{CDC}, this polynomial coincides up to a non-vanishing multiplicative constant with the characteristic polynomial~$Q(z,w)$ of an associated {\em bipartite\/} graph~$\widetilde{G}^C\subset\mathbb{T}^2$. (This result is already present in essence in~\cite{Dub}.) Moreover, it is also shown in~\cite{CDC} that~$Q(z,w)$ never vanishes on the unit torus for~$\beta\neq\beta_c$, and does so at a unique real positive node for~$\beta=\beta_c$. It only remains to
check that, with the conventions of Section~\ref{sub:gen}, this zero can only be located at~$(-1,1)$.
(Note that this is consistent with Proposition~\ref{prop:Q(-1)}.)

One way to do so is to use the geometric interpretation of Kasteleyn orientations of~\cite{C-R,Cim}.
In a nutshell, the powers of~$i$ (resp. the signs) appearing in the Pfaffian expansion
of~$R(z_0,w_0)$ (resp.~$P(z_0,w_0)$) with~$z_0,w_0\in\{\pm 1\}$ are described by maps
$q\colon H_1(\K;\Z)\to\Z/4\Z$ (resp~$\tilde{q}\colon H_1(\mathbb{T}^2;\Z)\to\Z/2\Z$) called quadratic forms
(resp. quadratic enhancements). 
Using this language, condition~(ii) in Section~\ref{sub:Klein} can be translated as follows: the quadratic
enhancement corresponding to~$K$ is determined by~$q(a)=q(a')=1$ or~$q(a)=q(a')=3$.
Using~\cite[Section~3]{Cim-P}, the quadratic form corresponding to the associated Kasteleyn orientation~$\widetilde{K}$ satisfies~$\tilde{q}(\tilde{a})=0$ and~$\tilde{q}(\tilde{b})=1$.
As showed in~\cite{CDC} in the context of Kac-Ward determinants, the real node
of~$P$ corresponds to the unique odd quadratic form, i.e. the one mapping~$\tilde{a}$ and~$\tilde{b}$ to~$1$.
By the computation above, this corresponds to changing the parity of~$\tilde{q}(\tilde{a})$, i.e. twisting
the determinant of the Kasteleyn matrix by~$(z_0,w_0)=(-1,1)$. This completes the proof.
\end{proof}

We can now apply Equation~\eqref{equ:Ising-dimer} to~$G_{mn}$ and Theorem~\ref{thm:as-gen} to the dimer model
on~$\G=G^F\subset\K$. This yields the asymptotic expansion of~$\log Z^J_\beta(G_{mn})$ with free energy~$\mathbf{f}_0$ as described in the statement,
and~$\mathsf{fsc}=\log\mathsf{FSC}$ given as follows: for~$\beta\neq\beta_c$, we have
\begin{equation}
\label{equ:FSC1}
\mathsf{FSC}=\left\{
\begin{array}{ll}
\left|\sin(A{\textstyle{\frac{\pi}{4}}})\right|+\left|\cos(A'{\textstyle{\frac{\pi}{4}}})\right|&\text{for~$m$ odd,}\\
\left|\sin((A-A'){\textstyle{\frac{\pi}{4}}})\right|+1&\text{for~$m$ even,}
\end{array}
\right.
\end{equation}
while for~$\beta=\beta_c$, we have
\begin{equation}
\label{equ:FSC2}
\mathsf{FSC}=\left\{
\begin{array}{ll}
\left|\sin(A{\textstyle{\frac{\pi}{4}}})\right|\left(\frac{\vartheta_{01}(\tau)}{\eta(\tau)}\right)^{\frac{1}{2}}+\left|\cos(A'{\textstyle{\frac{\pi}{4}}})\right|\left(\frac{\vartheta_{00}(\tau)}{\eta(\tau)}\right)^{\frac{1}{2}}&\text{for~$m$ odd,}\\
\left|\sin((A-A'){\textstyle{\frac{\pi}{4}}})\right|\left(\frac{\vartheta_{01}(\tau)}{\eta(\tau)}\right)^{\frac{1}{2}}+\left(\frac{\vartheta_{00}(\tau)}{\eta(\tau)}\right)^{\frac{1}{2}}&\text{for~$m$ even,}
\end{array}
\right.
\end{equation}
with~$\tau=i\,\frac{m}{n}\left|\frac{\partial^2_zP(-1,1)}{\partial^2_wP(-1,1)}\right|^{1/2}$.
Hence, we are now left with the computation of~$A,A'\in\Z/4\Z$.

First note that by the last sentence of Theorem~\ref{thm:as-gen} together with Lemma~\ref{lemma:node},~$A'$
is independant of~$\beta$ while~$A$ can only change at~$\beta=\beta_c$.
Furthermore, for extremal values of~$\beta\in[0,\infty]$, these modulo~$4$ integers can be determined 
by the following result.

\begin{lemma}
\label{lemma:fsc}
We have~$\mathsf{fsc}=0$ for~$\beta=0$ and~$\mathsf{fsc}=\log(2)$ for~$\beta=\infty$.
\end{lemma}

\begin{proof}
Recall that we have the asymptotic expansion
\[
\log Z(G^F_{mn},x^F)=mn\frac{I}{4}+\mathsf{fsc}+o(1),\quad\text{with}\quad I=\int_{\mathbb T^2} \log P(z,w)\frac{dz}{2\pi iz}\frac{dw}{2\pi iw}\,.
\]
Let us start with the value~$\beta=0$, corresponding to the dimer weights~$x_e=\tanh(\beta J_e)=0$.
In this case, we have
\[
Z(G^F_{mn},x^F)=2^{|V(G_{mn})|}=2^{mn|V(G)|}\,,
\]
while the characteristic polynomial is constant equal to
\[
P(z,w)=Z(\widetilde{G^F},x^F)^2=\left(2^{|V(\widetilde{G})|}\right)^2=2^{4|V(G)|}\,,
\]
yielding~$I=4|V(G)|\log(2)$.
Comparing this with the asymptotic expansion displayed above,
we see that for~$\beta=0$, we have~$\mathsf{fsc}=0$.

Let us now consider the extremal value~$\beta=\infty$, which gives the dimer weights~$x_e=1$.
By Equation~\eqref{equ:Fisher}, the corresponding dimer partition function is given
by
\[
Z(G_{mn}^F,x^F)=2^{|V(G_{mn})|}|\mathcal{E}(G_{mn})|=2^{mn|V(G)|+\dim H_1(G_{mn})}\,.
\]
The graph~$G_{mn}$ being connected, we have the Euler characteristic computation
\[
1-\dim H_1(G_{mn})=\chi(G_{mn})=|V(G_{mn})|-|E(G_{mn})|=mn(|V(G)|-|E(G)|)\,,
\]
leading to
\[
Z(G_{mn}^F,x^F)=2^{1+mn|E(G)|}\,.
\]
A possible way of computing the corresponding polynomial~$P(z,w)$ is to recall that it is equal
to~$2^{2|V(\widetilde{G})|}$ times the twisted Kac-Ward determinant (see~\cite[Section~3.1]{CCK}),
which for~$x=1$, is equal to~$2^{-|V(\widetilde{G})|+|E(\widetilde{G})|}$ times the characteristic polynomial~$Q(z,w)$ (see~\cite[Theorem~3.1]{CDC}). Since this later polynomial is given by~$Q(z,w)=2^{|F(\widetilde{G})|}$
for~$x=1$, where~$F(\widetilde{G})$ denotes the faces of~$\widetilde{G}\subset\mathbb{T}^2$,
the Euler characteristic computation~$|V(\widetilde{G})|-|E(\widetilde{G})|+|F(\widetilde{G})|=\chi(\mathbb{T}^2)=0$ now leads to
\[
P(z,w)=2^{2|E(\widetilde{G})|}=2^{4|E(G)|}\,,
\]
yielding~$I=4|E(G)|\log(2)$.
Comparing these results with the asymptotic expansion displayed above leads to the value~$\mathsf{fsc}=\log(2)$.
\end{proof}

Theorem~\ref{thm:as-Ising} now follows Equations~\eqref{equ:FSC1} and~\eqref{equ:FSC2}
together with one final lemma. 

\begin{lemma}
\label{lemma:tech-Ising}
\begin{enumerate}[(i)]
\item{For all~$\beta\in[0,\infty]$, we have~$A'=0\in\Z/4\Z$.}
\item{The modulo~$4$ integer~$A$ is equal to~$0$ for~$\beta<\beta_c$, to~$2$ for~$\beta>\beta_c$,
and to~$\pm 1$ for~$\beta=\beta_c$.}
\end{enumerate}
\end{lemma}

\begin{proof}
First observe that Lemma~\ref{lemma:fsc} and Equation~\eqref{equ:FSC1} imply that~$A'=0$ for~$\beta=\infty$.
Since this modulo~$4$ integer is known to be constant, the first point follows.
By the same Lemma~\ref{lemma:fsc} and Equation~\eqref{equ:FSC1}, now with~$A'=0$,
we have~$A=0$ for~$\beta=0$ and~$A=2$ for~$\beta=\infty$.
Since~$A$ remains constant as~$\beta$ varies in~$[0,\infty]\setminus\{\beta_c\}$, the second point
is proved for~$\beta\neq\beta_c$, and we are left with the determination of~$A$ at~$\beta=\beta_c$.
To do so, let us recall the identities
\[
R(z^{1/2},\pm 1)R(-z^{1/2},\pm 1)=P(z,\pm 1)\quad\text{and}\quad P(z,w)=k\cdot Q(z,w)
\]
of Proposition~\ref{prop:P} and of the proof of Lemma~\ref{lemma:node}, with~$k\in\R^*$ and~$Q(z,w)$ the characteristic polynomial of a bipartite toric graph~$\widetilde{G}^C$.
With these two equalities, the proof of Proposition~\ref{prop:Q(-1)}, Proposition~\ref{prop:rootS}
and Lemma~\ref{lemma:S} extend \emph{verbatim}, with~$Q(z,w)$ replaced by~$P(z,w)$
and~$S(z,\pm 1)$ by~$R(z,\pm 1)$.
By Lemma~\ref{lemma:node}, we know that for~$\beta=\beta_c$, the polynomial~$P(z,w)$ has a node in the unit torus. The third point of Lemma~\ref{lemma:S} then implies that~$A$ is odd, i.e. congruent to~$\pm 1$ modulo~$4$. This concludes the proof of the lemma.
\end{proof}

We conclude this article with an explicit example, and remarks.

\begin{example}
\label{ex:Ising}
Consider the only case previously studied in the literature, that of the isotropic square lattice.
Standard computations lead to the value~$\tau=i\frac{m}{2n}$ at the critical temperature~$\beta_c=\frac{1}{2}\log(\sqrt{2}+1)$.
The corresponding value of~$\mathsf{fsc}$ from Theorem~\ref{thm:as-Ising} is easily seen to coincide with~\cite[Equation~(59)]{LWIsing} (and also with the result of~\cite{Ch-P}). However, it does not coincide with~\cite[Equation~(41)]{Iz12}.
\end{example}

\begin{remarks}
\label{rem:Ising}
\begin{enumerate}[(i)]
\item{The statement of Lemma~\ref{lemma:fsc} holds in the case of a toric graph~$G\subset\T^2$ as well,
with the exact same proof. By continuity, this determines the finite size corrections
for all~$\beta\neq\beta_c$, and completes the determination of the asymptotic expansion of
the Ising partition function on~$G_{mn}$ given in~\cite[Corollary~3.5]{Ctech}. In particular,
Fisher graphs in the supercritical regime form a wide class of graphs with~$\mathsf{fsc}\neq 0$ even
though the spectral curve is disjoint from the unit torus,
see also the discussion at the end of Example~\ref{ex:tri-sq}.}
\item{In the spirit of Section~\ref{sub:cons}, we can consider the ratio of partition
functions~$\frac{Z^J_\beta(G_{mn})^2}{Z^J_\beta(\widetilde{G}_{mn})}$ for an arbitrary
non-degenerate graph~$G\subset\K$. Theorem~\ref{thm:as-Ising} and the remark above imply that it
tends to~$1$ for~$\beta\neq\beta_c$, and to the following universal limit at criticality:
\[
\lim_{m,n\to\infty}\frac{Z^J_\beta(G_{mn})^2}{Z^J_\beta(\widetilde{G}_{mn})}
=\frac{2\vartheta_{00}(\tau)+\left(2\vartheta_{00}(\tau)\vartheta_{01}(\tau)\right)^{1/2}+\vartheta_{01}(\tau)}{\vartheta_{00}(\tau)+\vartheta_{01}(\tau)+\vartheta_{10}(\tau)}\,.
\]
}
\item{Considering the~$\tau_\text{im}\to\infty$ limit as in Section~\ref{sub:cons} leads to the asymptotic expansion
\[
\log Z^J_\beta(G_{mn})=mn\frac{\mathbf{f}_0}{2}+\frac{\pi\tau_\text{im}}{48}-\log(2-\sqrt{2})+o(1)\,.
\]
Comparing it with~\eqref{equ:CFT'} leads to~$\Delta=-\pi/24$, i.e. to the value~$c=\frac{1}{2}$ for the central charge of a conformal field theory describing the Ising model. Therefore, our results are in full agreement with the predictions of~\cite{BCN}.}
\end{enumerate}
\end{remarks}

\bibliographystyle{plain}

\bibliography{bibliographie}

\end{document}